%% file: main.tex
\documentclass[final]{llncs}
\pdfoutput=1

\input{headers}
\def\labelindent{\skip52}

\makeatletter
\renewcommand*\l@author[2]{}
\renewcommand*\l@title[2]{}
\makeatletter

\DeclareSymbolFont{greek}     {OML}{cmm}{m}{it}
\DeclareMathSymbol{\alpha}{\mathord}{greek}{"0B}
\DeclareMathSymbol{\beta}{\mathord}{greek}{"0C}
\DeclareMathSymbol{\gamma}{\mathord}{greek}{"0D}
\DeclareMathSymbol{\delta}{\mathord}{greek}{"0E}
\DeclareMathSymbol{\epsilon}{\mathord}{greek}{"22}
\DeclareMathSymbol{\varepsilon}{\mathord}{greek}{"22}
\DeclareMathSymbol{\zeta}{\mathord}{greek}{"10}
\DeclareMathSymbol{\eta}{\mathord}{greek}{"11}
\DeclareMathSymbol{\theta}{\mathord}{greek}{"12}
\DeclareMathSymbol{\iota}{\mathord}{greek}{"13}
\DeclareMathSymbol{\kappa}{\mathord}{greek}{"14}
\DeclareMathSymbol{\lambda}{\mathord}{greek}{"15}
\DeclareMathSymbol{\mu}{\mathord}{greek}{"16}
\DeclareMathSymbol{\nu}{\mathord}{greek}{"17}
\DeclareMathSymbol{\xi}{\mathord}{greek}{"18}
\DeclareMathSymbol{\pi}{\mathord}{greek}{"19}
\DeclareMathSymbol{\rho}{\mathord}{greek}{"1A}
\DeclareMathSymbol{\sigma}{\mathord}{greek}{"1B}
\DeclareMathSymbol{\tau}{\mathord}{greek}{"1C}
\DeclareMathSymbol{\upsilon}{\mathord}{greek}{"1D}
\DeclareMathSymbol{\phi}{\mathord}{greek}{"27}
\DeclareMathSymbol{\chi}{\mathord}{greek}{"1F}
\DeclareMathSymbol{\psi}{\mathord}{greek}{"20}
\DeclareMathSymbol{\omega}{\mathord}{greek}{"21}
\DeclareMathSymbol{\vartheta}{\mathord}{greek}{"23}
\DeclareMathSymbol{\varpi}{\mathord}{greek}{"24}
\DeclareMathSymbol{\varrho}{\mathord}{greek}{"25}
\DeclareMathSymbol{\varsigma}{\mathord}{greek}{"26}
\DeclareMathSymbol{\varphi}{\mathord}{greek}{"27}

\begin{document}

\overfullrule=5pt
\hyphenation{Brow-serID}
\hyphenation{in-fra-struc-ture}
\hyphenation{brow-ser}
\hyphenation{doc-u-ment}
\hyphenation{Chro-mi-um}
\hyphenation{meth-od}
\hyphenation{sec-ond-ary}
\hyphenation{Java-Script}
\hyphenation{Mo-zil-la}
\hyphenation{post-Mes-sage}

\title{Analyzing the BrowserID SSO System\\ with Primary Identity Providers\\ Using an Expressive Model of the Web}
\author{Daniel Fett \and Ralf Küsters \and Guido Schmitz}
\institute{University of Trier, Germany\\\email{\{fett,kuesters,schmitzg\}@uni-trier.de}}

\maketitle

\ifdraft{
\listoftodos
}{ }

\begin{abstract}
  \input{abstract}

\end{abstract}

\tableofcontents
\newpage

\input{intro}

\input{shortmodel}

\input{generalproperties}

\input{overview}

\input{browserid-javascript}
%
%

%
\input{analysis}

\input{modelling-browserid}

\input{securityproperties-informal}

\input{attacks}

\input{theorem}

\input{privacy-analysis}
\input{related-work}

\input{conclusion}
{\footnotesize

\input{bibl}
\bibliographystyle{abbrv}
}

\appendix

\input{appendix-webmodel}

\input{appendix-message-formats}

\input{appendix-browsermodel}

\input{appendix-general-properties}

\input{appendix-browserid-lowlevel}

\input{appendix-browserid-pidp}

\input{securityproperties-formal}

\input{appendix-proof-pidp}

\newpage
\section{BrowserID Login Flow Overviews}
\input{figure-pidp-detailed}
\input{figure-sidp}

\end{document}

%% file: headers.tex
\usepackage[utf8]{inputenc}
\usepackage[T1]{fontenc}

\usepackage{mathptmx}

\usepackage{keyval}
\usepackage{footmisc}
\usepackage{ragged2e}
\usepackage{xspace}
\usepackage[usenames,dvipsnames,svgnames,table]{xcolor}
\usepackage[final,stretch=10,shrink=10]{microtype}
\usepackage[showlabels,sections,floats,textmath,displaymath]{preview}
\usepackage{ifdraft}
\usepackage{enumitem}
\usepackage{footnote}
\usepackage{tikz}
\usetikzlibrary{backgrounds,snakes,arrows,decorations.markings,calc}
\usepackage{varwidth}
\usepackage{algpseudocode}
\usepackage{algorithm}

\usepackage{captcont}
\usepackage{caption}
\DeclareCaptionLabelSeparator{dot}{.\enspace}
\clearcaptionsetup{figure}
\usepackage[nottoc,section]{tocbibind}

\usepackage{macros}
\usepackage{cite}
\usepackage{amsmath}
\usepackage{amssymb}
\usepackage{thmtools}
\setlistdepth{9}
\setlist[itemize,1]{label=$\bullet$}
\setlist[itemize,2]{label=$\bullet$}
\setlist[itemize,3]{label=$\bullet$}
\setlist[itemize,4]{label=$\bullet$}
\setlist[itemize,5]{label=$\bullet$}
\setlist[itemize,6]{label=$\bullet$}
\setlist[itemize,7]{label=$\bullet$}
\setlist[itemize,8]{label=$\bullet$}
\setlist[itemize,9]{label=$\bullet$}

\renewlist{itemize}{itemize}{9}

\usepackage[english]{babel}
\usepackage[final]{hyperref}
\usepackage{color}
\definecolor{linkcolor}{rgb}{0,0,0.5}
\hypersetup{
 colorlinks=true,
  linkcolor=linkcolor,
  anchorcolor=linkcolor,
  citecolor=linkcolor,
  filecolor=linkcolor,
  menucolor=linkcolor,
  runcolor=linkcolor,
  urlcolor=linkcolor,
}

\usepackage{todonotes}

\usepackage{fancybox}
\usepackage{wrapfig}

%% file: abstract.tex
BrowserID is a complex, real-world Single Sign-On (SSO)
System for web applications recently developed by
Mozilla. It employs new HTML5 features (such as web
messaging and web storage) and cryptographic assertions to
provide decentralized login, with the intent to respect
users' privacy. It can operate in a primary and a
secondary identity provider mode. While in the primary mode
BrowserID runs with arbitrary identity providers, in the
secondary mode there is one identity provider only, namely
Mozilla's default identity provider.

We recently proposed an
expressive general model for the web infrastructure and,
based on this web model, analyzed the security of the
secondary identity provider mode of BrowserID. The analysis
revealed several severe vulnerabilities, which have been
fixed by Mozilla.

In this paper, we complement our prior work by analyzing
the even more complex primary identity provider mode of
BrowserID. We do not only study authentication properties
as before, but also privacy properties. During our analysis
we discovered new and practical attacks that do not apply
to the secondary mode: an identity injection attack, which
violates a central authentication property of SSO systems,
and attacks that break an important privacy promise of
BrowserID and which do not seem to be fixable without a
major redesign of the system. Interestingly, some of our
attacks on privacy make use of a browser side channel that,
to the best of our knowledge, has not gained a lot of
attention so far.

For the authentication bug, we propose a fix and formally
prove in a slight extension of our general web model that
the fixed system satisfies all the authentication
requirements we consider. This constitutes the most complex
formal analysis of a web application based on an expressive
model of the web infrastructure so far.

As another contribution, 
we identify and prove important security properties of
generic web features in the extended web model to
facilitate future analysis efforts of web standards and web
applications.

%% file: intro.tex
\section{Introduction}

Single sign-on (SSO) systems have become an important
building block for authentication in the web. Over the last
years, many different SSO systems have been developed, for
example, OpenID, OAuth, and proprietary solutions such as
Facebook Connect. These systems usually allow a user to
identify herself to a so-called relying party (RP), which
provides some service, using an identity that is managed by
an identity provider (IdP), such as Facebook or Google.

Given their role as brokers between IdPs and RPs, the
security of SSO systems is particularly crucial: numerous
attacks have shown that vulnerabilities in SSO systems
compromise the security of many services and users at
once (see, e.g.,
\cite{BansalBhargavanMaffeis-CSF-2012,SantsaiBeznosov-CCS-2012-OAuth,ArmandoEtAl-FMSE-2008,SomorovskyMayerSchwenkKampmannJensen-USENIX-2012,WangChenWang-SP-2012-SSO,SantsaiHaekyBeznosov-CS-2012}).

BrowserID \cite{mozilla/persona} is a relatively new
complex SSO system which allows users to utilize any of
their existing email addresses as an identity. BrowserID,
which is also known by its marketing name \emph{Persona},
has been developed by Mozilla and provides decentralized
and federated login, with the intent to respect users'
privacy: While in other SSO systems (such as OpenID), by
design, IdPs can always see when and where
their users log in, Mozilla's intention behind the design
of BrowserID was that such tracking should not be
possible. Several web applications support BrowserID
authentication. For example, popular content management
systems, such as Drupal and WordPress allow users to log in
using BrowserID. Also Mozilla uses this SSO system on
critical web sites, e.g., their bug tracker Bugzilla and
their developer network MDN. 

The BrowserID implementation is based solely on native web
technologies. It uses many new HTML5 web features, such as
web messaging and web storage. For example, BrowserID uses
the postMessage mechanism for cross-origin inter-frame
communication (i.e., communication within a browser between
different windows) and the web storage concept of modern
browsers to store user data on the client side.

There are two modes for BrowserID: For the best user
experience, email providers (IdPs) can actively support
BrowserID; they are then called \emph{primary IdPs}. For all
other email providers that do not support BrowserID,
the user can register her email address at a default IdP,
namely Mozilla's \nolinkurl{login.persona.org}, the
so-called \emph{secondary IdP}.

In \cite{FettKuestersSchmitz-SP-2014}, we proposed a general and expressive Dolev-Yao style
model for the web infrastructure. This web model is
designed independently of a specific web application and
closely mimics published (de-facto) standards and
specifications for the web, for instance, the HTTP/1.1 and HTML5
standards and associated (proposed) standards (main\-ly RFCs). It is the most
comprehensive web model to date. Among others, HTTP(S)
requests and responses, including several headers, such as
cookie, location, strict transport security (STS), and
origin headers, are modeled. The model of web browsers
captures the concepts of windows, documents, and iframes,
including the complex navigation rules, as well as new
technologies, such as web storage and cross-document
messaging (postMessages). JavaScript is modeled in an
abstract way by so-called scripting processes which can be
sent around and, among others, can create iframes and
initiate XMLHTTPRequests (XHRs). Browsers may be corrupted
dynamically by the adversary.

Based on this general web model, we analyzed the
security of the secondary IdP mode of BrowserID
\cite{FettKuestersSchmitz-SP-2014}. The analysis revealed
several severe vulnerabilities, which have since been fixed
by Mozilla.

\subsubsection{Contributions of this Paper.} The main
contributions of this paper are that we i) analyze
authentication and privacy properties for the primary mode
of BrowserID, where in both cases the analysis revealed new
attacks, ii) identify generic web security properties to
ease future analysis efforts, and iii) slightly extend our
web model. 

As mentioned before, in \cite{FettKuestersSchmitz-SP-2014},
we studied the simpler secondary mode of BrowserID
only. The primary model studied here is much more complex
than the secondary mode (see also the remarks in
Section~\ref{sec:javascript-descr}). It involves more
components (such as an arbitrary set of IdPs, more
iframes), a much more complex communication structure, and
requires weaker trust assumptions (for example, some IdPs,
and hence, the JavaScript they deliver, might be
malicious).  Also, in our previous work, we have not
considered privacy properties, but authentication
properties only.

More specifically, the contributions of this paper can be
summarized as follows.

\paragraph{Extension of the Web Model.} We slightly extend
our web model proposed in
\cite{FettKuestersSchmitz-SP-2014}. We complement the
modeling of the web storage concept of modern browsers by
adding sessionStorage \cite{w3c/webstorage}, which is
(besides the already modeled localStorage) heavily used by
BrowserID in its primary mode.
We also extend the model to include a set of user
identities (e.g., user names or email addresses) in
addition to user secrets.

\paragraph{Authentication Attack and Security Proof for BrowserID.}
The authentication properties we analyze are central to any
SSO system and correspond to those considered in our previous work: i) the attacker should not be able to log in at an RP
as an honest user and ii) the attacker should not be able
to authenticate an honest user/browser to an RP with an ID
not owned by the user (identity injection). While trying to
prove these authentication properties for the primary mode
of BrowserID, we discovered a new attack which violates
property ii). Depending on the service provided by the RP,
this could allow the attacker to track the honest user or
to obtain user secrets. We confirmed the attack on the
actual implementation and reported it to Mozilla, who
acknowledged the attack.  We note that this attack does not
apply to the secondary mode.

We propose a fix and provide a detailed formal proof based
on the (extended) web model which shows that the fixed
system satisfies the mentioned authentication
properties. This constitutes the most complex formal
analysis of a web application based on an expressive model
of the web infrastructure, in fact, as mentioned, the most
comprehensive one to date. We note that other web models
are too limited to be applied to BrowserID (see also
Section~\ref{sec:relatedwork}).

\paragraph{Privacy Attacks on BrowserID.}  As pointed out
before, BrowserID was designed by Mozilla with the explicit
intention to respect users' privacy. Unlike in other SSO
systems, when using BrowserID, IdPs should not learn to
which RP a user logs in. When trying to formally prove this
property, we discovered attacks that show that BrowserID
cannot live up to this claim. Our attacks allow malicious
IdPs to check whether or not a user is logged in at a
specific RP with little effort. Interestingly, one variant
of these attacks exploits a browser side channel which, to
our knowledge, has not received much attention in the literature so
far. Just as for authentication, we have confirmed the
attacks on the actual implementation and reported them to
Mozilla~\cite{mozilla-browserid-primary-privacy-bug-report},
who acknowledged the attacks.  Unfortunately, the attacks
exploit a design flaw of BrowserID that does not seem to be
easily fixable without a major redesign.

\paragraph{Generic Web Security Properties.} Our security
analysis of BrowserID and the case study
 in \cite{FettKuestersSchmitz-SP-2014} show that
certain security properties of the web model need to be
established in most security proofs for web standards and
web applications. As another contribution, we
therefore identify and summarize central security
properties of generic web features in our extension of our
model and formalize them in a general way
such that they can be used in and facilitate future
analysis efforts of web standards and web applications.

\subsubsection{Structure of this Paper.}
In Section~\ref{sec:webmodel}, we present the basic
communication model and the web model, including our
extensions. We deduce general properties of this model,
which are independent of specific web applications, in
Section~\ref{sec:generalproperties}. For our security
analysis, we first, in Section~\ref{sec:browserid}, provide
a description of the BrowserID system, focusing on the
primary mode. We then, in
Section~\ref{sec:analysisbrowserid}, present our attack and
the formal analysis of the authentication properties of
the (fixed) BrowserID system in primary mode. In
Section~\ref{sec:privacyanalysis-pidp}, we present our
attacks on privacy of BrowserID. Related work is discussed
in Section~\ref{sec:relatedwork}. We conclude in
Section~\ref{sec:conclusion}. Full details can be found in
the appendix.

%% file: shortmodel.tex
\section{The Web Model}\label{sec:webmodel}

In this section, we present the model of the web
infrastructure as proposed in
\cite{FettKuestersSchmitz-SP-2014}, along with our
extensions (sessionStorage and user identities) mentioned
in the introduction, with full details, most of which taken
from \cite{FettKuestersSchmitz-SP-2014}, provided in
Appendices~\ref{app:communication-model}
to~\ref{sec:deta-descr-brows}. We first present the generic
Dolev-Yao style communication model which the model is
based on.

\subsection{Communication Model}\label{sec:communicationmodel}
The main entities in the communication model are
\emph{atomic processes}, which will be used to model web
browsers, web servers, DNS servers as well as web and
network attackers. Each \ap has a list of addresses
(representing IP addresses) it listens to.  A set of \aps
forms what is called a \emph{system}. The different atomic
processes in such a system can communicate via events,
which consist of a message as well as a receiver and a
sender address.  In every step of a run, one event is
chosen non-deterministically from the current ``pool'' of
events and is delivered to an \ap that listens to the
receiver address of that event; if different atomic
processes can listen to the same address, the atomic
process to which the event is delivered is chosen
non-deterministically among the possible processes. The
(chosen) atomic process can then process the event and
output new events, which are added to the pool of events,
and so on. More specifically, messages, processes, etc.~are
defined as follows.

\subsubsection{Terms, Messages and Events.} 
As usual in Dolev-Yao models (see, e.g.,
\cite{AbadiFournet-POPL-2001}), messages are expressed as
formal terms over a signature. Later messages may, for
instance, represent HTTP(S) requests and responses.

The signature $\Sigma$ for the terms and messages
considered in this work is the union of the following
pairwise disjoint sets of function symbols: (1)
constants $C = \addresses\,\cup\,
\mathbb{S}\cup \{\True,\bot,\notdef\}\cup\, \nonces$
($\addresses$ for (IP) addresses, $\mathbb{S}$ for ASCII
strings, and $\nonces$ for an infinite set of nonces) where
the four sets are pairwise disjoint, (2) function symbols
for public keys, asymmetric/symmetric
encryption/decryption, and digital signatures:
$\mathsf{pub}(\cdot)$, $\enc{\cdot}{\cdot}$,
$\dec{\cdot}{\cdot}$, $\encs{\cdot}{\cdot}$,
$\decs{\cdot}{\cdot}$, $\sig{\cdot}{\cdot}$,
$\checksig{\cdot}{\cdot}$, $\unsig{\cdot}$, (3) $n$-ary
sequences $\an{}, \an{\cdot}, \an{\cdot,\cdot},
\an{\cdot,\cdot,\cdot},$ etc., and (4) projection symbols
$\pi_i(\cdot)$ for all $i \in \mathbb{N}$. \emph{Ground
  terms} over this signature are terms that do not contain
variables. These terms represent messages. By $\messages$
we denote the set of messages. An \emph{event (over $\addresses$ and $\messages$)} is of
the form $(a{:}f{:}m)$, for $a, f\in \addresses$ and $m \in
\messages$, where $a$ is interpreted to be the receiver
address and $f$ is the sender address.

For example, $k\in \nonces$ and $\pub(k)$ are messages,
where $k$ typically models a private key and $\pub(k)$ the
corresponding public key. For strings $\str{a},\str{b}\in \mathbb{S}$
and the nonce $k\in \nonces$, the message
$\enc{\an{\str{a},\str{b}}}{\pub(k)}$ is interpreted to be the message
$\an{\str{a},\str{b}}$ (the sequence of strings $\str{a}$ and $\str{b}$) encrypted
under the public key $\pub(k)$.

The \emph{equational theory} associated with the signature
$\Sigma$ is defined as usual in Dolev-Yao models. It
captures the meaning of the function symbols in
$\Sigma$. For instance, one equation is $\dec{\enc
  x{\pub(y)}}{y}=x$ and another
$\pi_i(\an{x_1,\ldots,x_n})=x_i$ for $1\le i\le n$. We have
that
$\pi_1(\dec{\enc{\an{\str{a},\str{b}}}{\pub(k)}}{k})\equiv
\str{a}$.

\subsubsection{Atomic Processes, Systems and Runs.} Atomic
Dolev-Yao processes, systems, and runs of systems are
defined as follows.

A \emph{(generic) atomic process} is a tuple $p = (I^p,
Z^p, R^p, s^p_0)$ where $I^p$ is a set of addresses (the
set of address the process listens to), $Z^p$ is a set of
states (formally, terms), $s^p_0\in Z^p$ is an initial
state, and $R^p$ is a relation that takes an event and a
state as input and (non-deterministically) returns a new
state and a set of events. This relation models a
non-deterministic computation step of the process, which
upon receiving an event in a given state
non-deterministically moves to a new state and outputs a
set of messages (events).

In the web model, we consider \emph{atomic Dolev-Yao (DY)
  processes} only. For these processes it is required that
the events and states that they output can be computed
(more formally, derived in the usual Dolev-Yao style) from
the current input event and state (see
Appendix~\ref{app:communication-model}). The rest of this
paper will consider DY processes only.

The so-called \emph{attacker process} is an atomic DY
process which records all messages it receives and outputs
all messages it can possibly derive from its recorded
messages.  Hence, an attacker process is the maximally
powerful DY process. It carries out all attacks any DY
process could possibly perform and is parametrized by the
set of sender addresses it may use.

A \emph{system} is a (possibly infinite) set of atomic
processes. Its state (i.e., the states of all atomic
processes in the system) together with a multi-set of waiting
events is called a \emph{configuration}.

A \emph{run} of a system for an initial set $E_0$ of events
is a sequence of configurations, where each configuration
(except for the first one, which consists of $E_0$ and the
initial states of the atomic processes) is obtained by
delivering one of the waiting events of the preceding
configuration to an atomic process $p$ (which listens to the
receiver address of the event), and which in turn
performs a computation step according to its relation
$R^p$.

\subsection{Scripting Processes}

For the web model, we also define scripting processes,
which model client-side scripting technologies, such as
JavaScript.

A \emph{scripting process} (or simply, a \emph{script}) is
defined similarly to a DY process. It is called by the
browser in which it runs. The browser provides it with a
(fresh, infinite) set $N$ of nonces and state information
$s$. The script then outputs a term $s'$, which represents
the new internal state and some command which is
interpreted by the browser (see
Section~\ref{sec:web-browsers} for details). Again, it is
required that a script's output $s'$ is derivable from its
input $(s,N)$.

Similarly to an attacker process, the so-called
\emph{attacker script} $\Rasp$ may output everything that
is derivable from the input.

\subsection{Web System}\label{sec:websystem}

In \cite{FettKuestersSchmitz-SP-2014}, we formalize the web infrastructure and web applications
by what they call a web system. A web system contains,
among others, a (possibly infinite) set of DY processes,
modeling web browsers, web servers, DNS servers, and
attackers (which may corrupt other entities, such as
browsers).

\subsubsection{Web System.}
A \emph{web system is a tuple $(\!\websystem\!,
  \scriptset\!,\mathsf{script}, E_0\!)$} with its
components defined as follows:

The first component, $\websystem$, denotes a system (a set
of DY processes) and is partitioned into the sets
$\mathsf{Hon}$, $\mathsf{Web}$, and $\mathsf{Net}$, where
in $\mathsf{Hon}$ the set of honest DY processes and in
$\mathsf{Web}$ and $\mathsf{Net}$ attacker processes (see
Section~\ref{sec:communicationmodel}) are specified,
\emph{web attacker} and \emph{network attacker processes},
respectively. While a web attacker can listen to and send
messages from its own addresses only, a network attacker
may listen to and spoof all addresses. Hence, it is the
maximally powerful attacker. Attackers may corrupt other
parties. In the analysis of a concrete web system, we
typically have one network attacker only and no web
attackers (as they are subsumed by the network attacker) or
one or more web attackers but then no network attacker.
Honest processes (in $\mathsf{Hon}$) can either be web
servers, web browsers, or DNS servers. In our security
analysis of authentication properties, DNS servers will be
subsumed by the attacker, and hence, we do not need to
model them for the analysis of these properties. Our
attacks on privacy work with honest DNS servers. As the
details of the modeling of these servers is not essential
to understand these attacks, we refer to
\cite{FettKuestersSchmitz-SP-2014} for the model of DNS
servers. The modeling of a \emph{web server} heavily
depends on the specific web application. Our concrete
models for the web servers of the BrowserID system are
provided in Sections~\ref{sec:browserid} and following.
Below, we present the modeling of web browsers, including
our extensions, which is independent of a specific web
application, with full details provided in
Appendices~\ref{app:message-data-formats}
and~\ref{sec:deta-descr-brows}.

The second component, $\scriptset$, is a finite set of
scripts, which include the attacker script $\Rasp\in
\scriptset$. In a concrete model of a web application, such as 
our BrowserID model, the set $\scriptset\setminus\{\Rasp\}$
typically describes the set of honest scripts used in the
considered application. Malicious scripts are modeled by
the ``worst-case'' malicious script, $\Rasp$.

The third component, $\mathsf{script}$, is an injective
mapping from $\scriptset$ to $\mathbb{S}$, i.e., by
$\mathsf{script}$ every $s\in \scriptset$ is assigned its
string representation $\mathsf{script}(s)$.  Finally, $E_0$
is a multi-set of events, containing an infinite number of
events of the form $(a{:}a{:}\trigger)$ for every process
$a$ in the web system. A \emph{run} of the web system
is a run of $\websystem$ initiated by $E_0$.

\subsection{HTTP Messages}\label{sec:http-messages}
HTTP requests are represented as ground terms containing
a nonce, a method (e.g., $\mGet$ or $\mPost$), a domain name,
a path, URL parameters, request headers (such as
$\str{Cookie}$), and a message body. For example, a $\mGet$
request for the URL \url{http://ex.com/show?p=1} can be modeled
as the term 
\[ \mi{r} := \langle \cHttpReq, n_1, \mGet,
\str{ex.com}, \str{/show},
\an{\an{\str{p},1}}, \an{}, \an{} \rangle \]
where headers
and body are empty. A response contains a nonce (the same
as in the request), a status code, headers, and a body. A
response to $r$ would be 
\[\mi{s} := \langle \cHttpResp,
n_1, \str{200}, \an{\an{\str{Set{\mhyphen}Cookie},
    \an{\str{SID},\an{n_2,\bot,\True,\bot}}}}, 
\an{\str{script1}, \mi{init}} \rangle
 \] 
where
$\an{\str{SID},\an{n_2,\bot,\True,\bot}}$ is a cookie with
the name/value pair $\str{SID}=n_2$ and the attributes
$\str{httpOnly}$, $\str{secure}$, $\str{session}$ set or
not set, and $\an{\str{script1},\mi{init}}$ is the body, in
this case an HTML document that is to be delivered to the
browser (modeled by the string representation of a script
and its initial state, see below).

A corresponding HTTPS request for $r$ as above would be
$\ehreqWithVariable{r}{k'}{\pub(k_\text{ex.com})}$, where
$k'$ is a fresh symmetric key (a nonce) generated by the
sender of the request. The responder is supposed to use
this key to encrypt the response, which, hence, is of the
form $\ehrespWithVariable{s}{k'}$.

\subsection{Web Browsers}\label{sec:web-browsers}

An honest browser is thought to be used by one honest
user. The honest user is modeled as part of the
browser. User actions are modeled as non-deterministic
actions of the web browser. For example, the web browser
itself can non-deterministically follow the links provided
by a web page. User data (i.e., passwords and identities)
is stored in the initial state of the browser (see below)
and is given to a web page when needed, similar to the
AutoFill feature in browsers. As detailed below, browsers
can be corrupted, i.e., taken over by web and network
attackers.

A web browser $p$ is modeled as a DY process $(I^p, Z^p,
R^p, s^p_0, N^p)$ where $I^p\subseteq \addresses$ is a
finite set of addresses $p$ may listen to and $N^p\subseteq
\nonces$ is an infinite set of nonces $p$ may use. The set
of states $Z^p$, the initial state $s^p_0$, and the
relation $R^p$ are defined next.

\subsubsection{Browser State: $Z^p$ and
  $s^p_0$.}\label{sec:browserstate}
The set $Z^p$ of states of a browser 
consists of terms of the form
\begin{align*}
  \langle
  &\mi{windows}, \allowbreak
  \mi{ids}, \allowbreak
  \mi{secrets}, \allowbreak
  \mi{cookies}, \allowbreak
  \mi{localStorage}, \allowbreak
  \mi{sessionStorage}, \allowbreak
  \mi{keyMapping}, \allowbreak \\
  &\mi{sts}, \allowbreak
  \mi{DNSaddress}, \allowbreak
  \mi{nonces}, \allowbreak
  \mi{pendingDNS}, \allowbreak
  \mi{pendingRequests},\allowbreak
  \mi{isCorrupted}
\rangle\!.
\end{align*}%

\paragraph{Windows and documents.} The most important part
of the state are windows and documents, both stored in the
subterm $\mathit{windows}$.  A browser may have several
windows open at any time (resembling the tabs and windows
in a real browser), each containing a list of documents
(the history of visited web pages) of which one is
``active'', namely the one currently presented to the user
in that window.  A window may be navigated forward and
backward (modeling navigation buttons), deactivating one
document and activating its successor or
predecessor. Intuitively, a document represents a loaded
HTML page. More formally, a document contains (the string
representation of) a script, which is meant to model both
the static HTML code (e.g., links and forms) as well as
JavaScript code. When called by the browser, a script
outputs a command which is then interpreted by the browser,
such as following a link or issuing an \xhr (see
below). Documents may also contain iframes, which are
represented as windows (\emph{subwindows}) nested inside of
document terms. This creates a tree of windows and
documents.

\paragraph{Secrets and IDs.} This subterm holds the
secrets and the identities of the user of the web
browser. Secrets (such as passwords) are modeled as nonces
and they are indexed by origins (where an origin is a
domain name plus the information whether the connection to
this domain is via HTTP or HTTPS).  Secrets are only
released to documents (scripts) with the corresponding
origin, similarly to the AutoFill mechanism in
browsers. Identities are arbitrary terms that model public
information of the user's identity, such as email
addresses. Identities are released to any origin. As
mentioned in the introduction, identities were not
considered in \cite{FettKuestersSchmitz-SP-2014}.

\paragraph{Cookies, localStorage, and sessionStorage.}
These subterms contain the cookies (indexed by domains),
\ls data (indexed by origins), and sessionStorage data
(indexed by origins and top-level window references) stored
in the browser. As mentioned in the introduction,
sessionStorage was not modeled in
\cite{FettKuestersSchmitz-SP-2014}.

\paragraph{KeyMapping.} This term is the equivalent to a
certificate authority (CA) certificate store in the
browser. Since, for simplicity, the model currently does
not formalize CAs, this term simply encodes a mapping
assigning domains $d\in \dns$ to their respective public
keys $\pub(k_d)$.

\paragraph{STS.} Domains that are listed in this term are
contacted by the web browser over HTTPS only. Connection
attempts over HTTP are transparently rewritten to HTTPS
requests. Servers can employ the
$\str{Strict\mhyphen{}Transport\mhyphen{}Security}$ header
to add their domain to this list.

\paragraph{DNSaddress.} This term defines the address of
the DNS server used by the browser.

\paragraph{Nonces, pendingDNS, and pendingRequests.}
These terms are used for bookkeeping purposes, recording
the nonces that have been used by the browser so far, the
HTTP(S) requests that await successful DNS resolution, and
HTTP(S) requests that await a response, respectively.

\paragraph{IsCorrupted.} This term indicates
whether the browser is corrupted ($\not=\bot$) or not
($=\bot$). A corrupted browser behaves like a web attacker.

\paragraph{Initial state $s_0^p$ of a web browser.}  In
the browser's initial state, $\mi{keyMapping}$, $\mi{DNSAddress}$, $\mi{secrets}$, and $\mi{ids}$ are defined as needed, $\mi{isCorrupted}$ is
set to $\bot$, and all other subterms are $\an{}$.

\begin{figure}[t!]
  \centering
  \begin{minipage}{\textwidth}
    \footnotesize
    { \underline{\textsc{Processing Input Message
          \hlExp{$m$}}}
      \begin{itemize}[noitemsep,nolistsep,label=,leftmargin=0pt]
      \item $m = \fullcorrupt$:
        $\mi{isCorrupted} := \fullcorrupt$
      \item $m = \closecorrupt$:
        $\mi{isCorrupted} := \closecorrupt$
      \item $m = \trigger$:
        non-deterministically choose $\mi{action}$ from
        $\{1,2\}$
        \begin{itemize}[noitemsep,nolistsep,label=,leftmargin=1em]
        \item $\mi{action} = 1$: Call
          script
          of some active document. Outputs 
          new state and $\mi{command}$.
          \begin{itemize}[noitemsep,nolistsep,label=,leftmargin=1em]
          \item $\mi{command} = \tHref$:
            $\rightarrow$ \emph{Initiate request}
          \item $\mi{command} = \tIframe$:
            Create subwindow, $\rightarrow$ \emph{Initiate
              request}
          \item $\mi{command} = \tForm$:
            $\rightarrow$ \emph{Initiate request}
          \item $\mi{command} = \tSetScript$:
            Change script in given document.
          \item $\mi{command} =
                \tSetScriptState$: Change state of script
            in given document.
          \item $\mi{command} =
                \tXMLHTTPRequest$: $\rightarrow$
            \emph{Initiate request}
          \item $\mi{command} = \tBack$ or $\tForward$: Navigate given window.
          \item $\mi{command} = \tClose$:
            Close given window.
          \item $\mi{command} =
                \tPostMessage$:
            Send \pm to specified document.
          \end{itemize}
        \item $\mi{action} = 2$:
          $\rightarrow$ \emph{Initiate request to some URL
            in new window}
        \end{itemize}
    \item $m=$ DNS response: send
      corresponding HTTP request
    \item $m=$ HTTP(S) response:
      (decrypt,) find reference.
      \begin{itemize}[noitemsep,nolistsep,label=,leftmargin=1em]
      \item reference to window: create document
        in window
      \item reference to
          document: add response
          body to document's script input
      \end{itemize}
    \end{itemize}
  }
\end{minipage}%
  \caption{The basic structure of the web browser relation
    $R^p$ with an extract of the most important processing
    steps, in the case that the browser is not already
    corrupted.}\label{fig:browser-structure}
\end{figure}

\subsubsection{Web Browser Relation
  $R^p$.}\label{sec:browserrelation}
This relation, outlined in
Figure~\ref{fig:browser-structure}, specifies how the web
browser processes incoming messages.  The browser may
receive special messages that cause it to become corrupted (first two lines in Figure~\ref{fig:browser-structure}), in
which case it acts like the attacker process. There are two
types of corruption: If the browser gets fully corrupted,
the attacker learns the entire current state of the
browser. If it gets close-corrupted, any open windows, documents and used nonces (in
particular, HTTPS encryption keys) are discarded from the
browser's state before it is handed over to the
attacker. This models that a user closed the browser, but a
malicious user now uses the browser (and all information
left in the browser's state).

The browser can receive a special trigger message
$\trigger$, upon which the browser non-deterministically
chooses one of two actions: i) Select one of the current
documents, trigger its JavaScript, and evaluate the output
of the script. Scripts can change the state of the browser
(e.g., by setting cookies) and can trigger specific actions
(e.g., following a link or creating an iframe), which are
modeled as \emph{commands} issued by the script (see the
list in Figure~\ref{fig:browser-structure}). ii) Follow
some URL, with the intuition that it was entered by the
user.

As mentioned, some of the above actions can cause the
browser to generate new HTTP(S) requests. In this case, the
browser first asks the configured DNS server for the IP
address belonging to the domain name in the HTTP(S)
request. As soon as the DNS response arrives, the browser
sends the HTTP(S) request to the respective IP address.

If the HTTP(S) response arrives, its headers are evaluated
and the body of the request becomes the script of a newly
created document that is then inserted at an appropriate
place in the window/document tree. However, if the HTTP(S)
response is a response to an \xhr (triggered by a
script in a document), the body of the response is added to
the corresponding document and can later be processed by
the script of that document.

%% file: generalproperties.tex
\section{General Security Properties}\label{sec:generalproperties}

We have identified central application independent security properties
of web features in the web model and formalized them in a general way
such that they can be used in and facilitate future analysis efforts
of web standards and web applications. In this section, we provide a
brief overview of these properties, with precise formulations and
proofs presented in Appendix~\ref{app:generalproperties}.

The first set of properties concerns encrypted connections (HTTPS): We
show that HTTP requests that were encrypted by an honest browser for
an honest receiver cannot be read or altered by the attacker (or any
other party). This, in particular, implies correct behavior on the
browser's side, i.e., that browsers that are not fully corrupted never
leak a symmetric key used for an HTTPS connection to any other party.
We also show that honest browsers set the host header in their
requests properly, i.e., the header reflects an actual domain name of
the receiver,  and that only the designated receiver can successfully
respond to HTTPS requests.

The second set of properties concerns origins and origin
headers. Using the properties stated above, we show that browsers
cannot be fooled about the origin of an (HTTPS) document in their
state: If the origin of a document in the browser's state is a secure
origin (HTTPS), then the document was actually sent by that
origin. Moreover, for requests which contain an origin header with a
secure origin we prove that such requests were actually initated by a
script that was sent by that origin to the browser. In other words, in
this case, the origin header works as expected.

%% file: overview.tex
\section{The BrowserID System}\label{sec:browserid}

BrowserID \cite{mozilla/persona/mdn} is a decentralized
single sign-on (SSO) system developed by Mozilla for user
authentication on web sites. It is a complex full-fledged
web application deployed in practice, with currently
{\raise.17ex\hbox{$\scriptstyle\mathtt{\sim}$}}47k LOC
(excluding some libraries).  It allows web sites to
delegate user authentication to email providers, identifying users by their email addresses.
BrowserID makes use of a broad variety of
browser features, such as \xhrs, \pm, local- and
sessionStorage, cookies, various headers, etc.

We first, in Section~\ref{sec:browseridoverview}, provide a
high-level overview of the BrowserID system. A more
detailed description of the BrowserID implementation is
then given in Section~\ref{sec:javascript-descr}. The
description of the BrowserID system presented in the
following as well as our BrowserID model (see
Section~\ref{sec:modelbrowseridpidp}) is extracted mainly
from the BrowserID source
code~\cite{mozilla/persona/source} and the (very
high-level) official BrowserID
documentation~\cite{mozilla/persona/mdn}.

\subsection{Overview}\label{sec:browseridoverview}

\input{highlevel-figure}
The BrowserID system knows three distinct parties: the
user, who wants to authenticate herself using a browser,
the relying party (RP) to which the user wants to
authenticate (log in) with one of her 
 email addresses (say,
\nolinkurl{user@idp.com}), and the identity/email address
provider, the IdP\@. If the IdP (\nolinkurl{idp.com})
supports BrowserID directly, it is called a \emph{primary
  IdP}. Otherwise, a Mozilla-provided service, the
so-called \emph{secondary IdP}, takes the role of the
IdP\@. As mentioned before, here we concentrate on the
primary IdP mode as the secondary IdP mode was described in
detail in~\cite{FettKuestersSchmitz-SP-2014}. However, we
briefly discuss the differenes between the two modes at the
end of Section~\ref{sec:javascript-descr}.

A primary IdP provides information about its 
setup in a so-called \emph{support document}, which it
provides at a fixed URL derivable from the email domain,
e.g., \url{https://idp.com/.well-known/browserid}.

A user who wants to log in at an RP with an email address
for some IdP has to present two signed documents to the RP:
A \emph{user certificate} (UC) and an \emph{identity
  assertion} (IA). The UC contains the user's email address
and the user's public key. It is signed by the IdP\@.  The
IA contains the origin of the RP and is signed with the
user's private key. Both documents have a limited validity
period. A pair consisting of a UC and a matching IA is
called a \emph{certificate assertion pair} (CAP) or a
\emph{backed identity assertion}. Intuitively, the UC in
the CAP tells the RP that (the IdP certified that) the
owner of the email address is (or at least claims to be)
the owner of the public key. By the IA contained in the CAP
the RP is ensured that the owner of the given public key
(i.e., the one who knows the corresponding private key)
wants to log in. Altogether, given a valid CAP, RP would
consider the user (identified by the email address in the
CAP) to be logged in.

The BrowserID authentication process (with a primary IdP)
consists of three phases (see
Figure~\ref{fig:browserid-highlevel}):
\refbigprotophase{provisioning} UC provisioning,
\refbigprotophase{authentication} CAP creation, and
\refbigprotophase{verification} CAP verification.

In Phase \refbigprotophase{provisioning}, (the browser of)
the user creates a public/pri\-vate key
pair~\refalphprotostep{gen-key-pair}. She then sends her
public key as well as the email address she wants to use to
log in at some RP to the respective
IdP~\refalphprotostep{req-uc}. The IdP now creates the
UC~\refalphprotostep{create-uc}, which is then sent to the
user~\refalphprotostep{recv-uc}. The above requires the
user to be logged in at IdP.

With the user having received the UC, Phase
\refbigprotophase{authentication} can start. The user wants
to authenticate to an RP, so she creates the
IA~\refalphprotostep{gen-ia}. The UC and the IA are
concatenated to a CAP, which is then sent to the
RP~\refalphprotostep{send-cap}.

In Phase \refbigprotophase{verification}, the RP checks the
authenticity of the CAP\@. For this purpose, the RP could
use an external verification service provided by Mozilla or
check the CAP itself as follows: First, the RP fetches the
public key of the IdP \refalphprotostep{recv-idp-pubkey},
which is contained in the support document. Afterwards, the
RP checks the signatures of the UC and the
IA~\refalphprotostep{verify-cap}. If this check is
successful, the RP can, as mentioned before, consider the
user to be logged in with the given email address and send
her some token (e.g., a cookie with a session ID), which we
refer to as an \emph{RP service token}.

%% file: highlevel-figure.tex
\begin{figure}[t]\centering
  \scriptsize{
  \begin{tikzpicture}

    \matrix [row sep=3.2ex, column sep={12ex}]
    {
     \node[anchor=base,fill=Gainsboro,rounded corners](rp){RP}; & \node[anchor=base,fill=Gainsboro,rounded corners](b){Browser}; & \node[anchor=base,fill=Gainsboro,rounded corners](idp){IdP}; \\
     & \node(b-gen-key){}; & \\
     & \node(b-send-pubkey){}; & \node(idp-recv-pubkey){}; \\
     & & \node(idp-create-uc){}; \\
     & \node(b-recv-cert){}; & \node(idp-send-cert){}; \\  [-1ex]
     \node(phase-1-2-left){}; & & \node(phase-1-2-right){}; \\ [-1ex]
     & \node(b-gen-ia){}; & \\
     \node(rp-recv-cap){}; & \node(b-send-cap){}; & \\ [-1ex]
     \node(phase-2-3-left){}; & & \node(phase-2-3-right){}; \\ [-1ex]
     \node(rp-recv-pubkey){}; & & \node(idp-send-pubkey){}; \\
     \node(rp-vrfy-cap){}; & & \\
     \node(rp-end){}; & \node(b-end){}; & \node(idp-end){}; \\
    };

    \begin{pgfonlayer}{background}
    \draw [color=gray] (rp) -- (rp-end);
    \draw [color=gray] (b) -- (b-end);
    \draw [color=gray] (idp) -- (idp-end);
    \end{pgfonlayer}

    \node(phase-1-2-leftleft)[left of=phase-1-2-left]{};
    \node(phase-1-2-rightright)[right of=phase-1-2-right]{};
    \draw [dashed] (phase-1-2-leftleft) -- (phase-1-2-rightright);
    \node(phase-2-3-leftleft)[left of=phase-2-3-left]{};
    \node(phase-2-3-rightright)[right of=phase-2-3-right]{};
    \draw [dashed] (phase-2-3-leftleft) -- (phase-2-3-rightright);

    \node at (b-gen-key) [draw,fill=white,rounded corners] {\alphprotostep{gen-key-pair} gen.\ key pair};
    \draw [->] (b-send-pubkey) to node[draw,fill=white]{\alphprotostep{req-uc} $\text{pk}_\text{b}$, email} (idp-recv-pubkey);
    \node at (idp-create-uc) [draw,fill=white,rounded corners] {\alphprotostep{create-uc} create UC};
    \draw [->] (idp-send-cert) to node[draw,fill=white]{\alphprotostep{recv-uc} UC} (b-recv-cert);
    \node at (b-gen-ia) [draw,fill=white,rounded corners]{\alphprotostep{gen-ia} gen. IA};
    \draw [->] (b-send-cap) to node[draw,fill=white]{\alphprotostep{send-cap} CAP} (rp-recv-cap);
    \draw [->] (idp-send-pubkey) to node[draw,fill=white]{\alphprotostep{recv-idp-pubkey} $\text{pk}_{\text{IdP}}$} (rp-recv-pubkey);
    \node at (rp-vrfy-cap) [draw,fill=white,rounded corners]{\alphprotostep{verify-cap} verify CAP};

    \node(phase-1-label) at ($(phase-1-2-left)!0.5!(rp)$) [left=1ex] {\bigprotophase{provisioning}};
    \node(phase-2-label) at ($(phase-2-3-left)!0.5!(phase-1-2-left)$) [left=1ex] {\bigprotophase{authentication}};
    \node(phase-3-label) at ($(rp-end)!0.75!(phase-2-3-left)$) [left=1ex] {\bigprotophase{verification}};

  \end{tikzpicture}}\vspace{-1.2em}
  \caption{BrowserID login: basic overview}
  \label{fig:browserid-highlevel}
\end{figure}
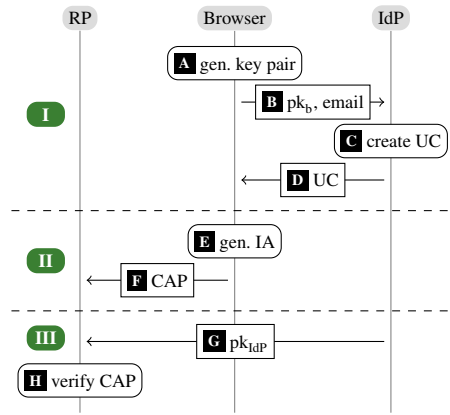%

%% file: browserid-javascript.tex
\subsection{Implementation Details}\label{sec:javascript-descr}

\input{lowlevel-ld-figure}

We now provide a more detailed description of the BrowserID
implementation. Since the system is very complex, with many
HTTPS requests, \xhrs, and postMessages sent between
different entities (servers as well as windows and iframes
within the browser), we here describe mainly the phases of
the login process without explaining every single message
exchange done in the implementation. A more detailed
step-by-step description can be found in
Appendix~\ref{app:browserid-lowlevel}. Note that
BrowserID's specification of IdPs fixes the interface to
BrowserID only, but otherwise does not further detail the
specification of IdPs. Therefore, in what follows, we
consider a typical IdP, namely the example implementation
provided by Mozilla \cite{mozilla/persona/source}.

In addition to the parties mentioned so far, the actual
BrowserID implementation uses another party, Mozilla's
\nolinkurl{login.persona.org} (LPO). Among others, LPO
provides HTML and JavaScript files that, for security and
privacy reasons, cannot be delivered by either IdP or
RP. An overview of the implementation is given in
Figure~\ref{fig:browserid-lowlevel-ld}. For brevity of
presentation, several messages and components, such as the
CIF (see below), are omitted in the figure (see
Figure~\ref{fig:browserid-lowlevel-pidp-detailed-1} on
Pages~\pageref{fig:browserid-lowlevel-pidp-detailed-1} and \pageref{fig:browserid-lowlevel-pidp-detailed-2} for a detailed version of Figure~\ref{fig:browserid-lowlevel-ld}).

\subsubsection{Windows and iframes in the Browser.} By
\emph{RP-Doc} we
denote the window containing the document loaded from some
RP, at which the user wants to log in with an email address
hosted by some IdP\@. RP-Doc typically includes JavaScript
from LPO and contains a button ``Login with
BrowserID''. The LPO JavaScript running in RP-Doc opens an
auxiliary window called the \emph{login dialog} (LD). Its
content is provided by LPO and it handles the interaction
with the user. During the login process, a temporary
invisible iframe called the \emph{provisioning iframe}
(PIF) can be created in the LD\@. The PIF is loaded from
IdP\@. It is used by LD to communicate (cross-origin) with
the IdP via postMessages: As the BrowserID implementation
mainly runs under the origin of LPO, it cannot directly
communicate with the IdP, thus it uses the PIF as a proxy.
Temporarily, the LD may navigate itself to a web page at
IdP to allow for direct user interaction with the IdP. We
then call this window the \emph{authentication dialog}
(AD).

\subsubsection{Login Process.} To describe the login
process, for the sake of presentation we assume for now
that the user uses a ``fresh'' browser, i.e., the user has
not been logged in before. As mentioned, the process starts
by the user visiting a web site of some RP\@.  After the
user has clicked on the login button in RP-Doc, the LD is
opened and the interactive login flow is started. We can
divide this login flow into seven phases: In
Phase~\refprotophase{ld-start-1}, the LD is initialized and
the user is prompted to provide her email address. Also, LD
fetches the support document (see
Section~\ref{sec:browseridoverview}) of the IdP via
LPO\@. In Phase~\refprotophase{ld-prov-1}, LD creates the
PIF from the \emph{provisioning URL} provided in the
support document. As (by our assumption) the user is not
logged in yet, the PIF notifies LD that the user is not
authenticated to the IdP\@. In
Phase~\refprotophase{ld-auth}, LD navigates itself away to
the \emph{authentication URL} which is also provided in the
support document and links to the IdP\@. Usually, this
document will show a login form in which the user enters
her password to authenticate to the IdP\@.  After the user
has been authenticated to IdP (which typically implies that
the IdP sets a session cookie in the browser), the window
is navigated back to LPO\@. 

Now, the login flow continues in Phase
\refprotophase{ld-start-2}, which basically repeats Phase
\refprotophase{ld-start-1}. However, the user is not
prompted for her email address (it has previously been
saved in the localStorage under the origin of LPO along
with a nonce, where the nonce is stored in the
sessionStorage). In Phase~\refprotophase{ld-prov-2}, which
essentially repeats Phase~\refprotophase{ld-prov-1}, the PIF
detects that the user is now authenticated to the IdP and the
provisioning phase is started
(\refbigprotophase{provisioning} in
Figure~\ref{fig:browserid-highlevel}): The user's keys are
created by LD and stored in the localStorage under the
origin of LPO\@. The PIF forwards the certification request
to the IdP, which then creates the UC and sends it back to the
PIF\@. The PIF in turn forwards it to the LD, which stores
it in the localStorage under the origin of LPO.

In Phases~\refprotophase{ld-lpo-auth} and
\refprotophase{ld-cap}, mainly the IA is generated by LD
for the origin of RP-Doc and sent (together with the UC) to
RP-Doc (\refbigprotophase{authentication} in
Figure~\ref{fig:browserid-highlevel}). In the localStorage,
LD stores that the user's email address is logged in at
RP\@. Moreover, to log the user in at LPO, LD generates an IA for the origin of
LPO and sends the UC and IA to LPO.

\subsubsection{Automatic CAP Creation.} In addition to the
interactive login presented above, BrowserID also contains
an automatic, non-interactive way for RPs to obtain a
freshly generated CAP: During initialization within RP-Doc,
an invisible iframe called the \emph{communication iframe}
(CIF) is created inside RP-Doc. The CIF's JavaScript is
loaded from LPO and behaves similar to LD, but without user
interaction. The CIF automatically issues a fresh CAP and
sends it to RP-Doc under specific conditions: among others,
the email address must be marked as logged in at RP in the
\ls. If necessary, a new key pair is created and a
corresponding new UC is requested at the IdP. For this
purpose, a PIF is created inside the CIF.

\subsubsection{Differences to the Secondary IdP Mode.} In
the secondary IdP mode there are three parties involved
only: RP, Browser, and LPO, where LPO also takes the role
of an IdP; LPO is the only IdP that is present, rather than
an arbitrary set of (external) IdPs.  Consequently, in the
secondary IdP mode the PIF and the AD do not
exist. Moreover, in the primary mode, the behavior of the
CIF and the LD is more complex than in the secondary
mode. For example, in the primary mode, just like the LD,
the CIF might contain a PIF (iframe in iframe) and interact
with it via postMessages. Altogether, the secondary IdP
case requires much less communication between
parties/components and trust assumptions are simpler: in
the secondary IdP mode LPO (which is the only IdP in this
mode) has to be trusted, in the primary IdP mode some
external IdPs might be malicious (and hence, also the
scripts they deliver for the PIF and the AD).  To
illustrate the difference between the secondary and the
primary IdP mode, in the appendix both modes are
illustrated in more detail, see
Figure~\ref{fig:browserid-lowlevel-sidp} on
Page~\pageref{fig:browserid-lowlevel-sidp} for the
secondary IdP mode and
Figure~\ref{fig:browserid-lowlevel-pidp-detailed-1} on
Pages~\pageref{fig:browserid-lowlevel-pidp-detailed-1}
and~\pageref{fig:browserid-lowlevel-pidp-detailed-2} for
the primary IdP mode.

%
%
%
%
%

%

%% file: lowlevel-ld-figure.tex
\begin{figure}[t!]\centering
\begin{tikzpicture}
\scriptsize

\matrix[column sep={6pc,between origins}, row sep=1.6ex] {
\node[anchor=base,fill=Gainsboro,rounded corners](lpo){LPO}; & \node[anchor=base,fill=Gainsboro,rounded corners](idp){IdP}; & \node[draw,anchor=base](rpdoc){RP-Doc}; & \node(ld-top){};                 & \node(pif-top){}; \\
                                   &                                    & \node(rpdoc-ld-open){};                 & \node(ld)[draw,anchor=base]{LD}; & \\
\node(phase-ld-init-1-left){};     &                                    &                                         &                                  & \node(phase-ld-init-1-right){}; \\ [2pt]
\node(lpo-ld-init-1){};            &                                    &                                         & \node(ld-init-1){};              & \\ [1pt]
                                   &                                    & \node(rpdoc-ld-recv-ready){};           & \node(ld-rpdoc-send-ready){};    & \\
                                   &                                    & \node(rpdoc-ld-send-request){};         & \node(ld-rpdoc-recv-request){};  & \\
\node(lpo-ld-ctx-1){};             &                                    &                                         & \node(ld-lpo-ctx-1){};           & \\
                                   &                                    &                                         & \node(ld-user-email){};          & \\
\node(lpo-ld-addr-info-1-top){};   &                                    &                                         & \node(ld-lpo-addr-info-1-top){}; & \\ [1pt]
\node(lpo-idp-wk-1){};             & \node(idp-lpo-wk-1){};             &                                         &                                  & \\
\node(lpo-ld-addr-info-1-btm){};   &                                    &                                         & \node(ld-lpo-addr-info-1-btm){}; & \\
\node(phase-prov-left){};          &                                    &                                         &                                  & \node(phase-prov-right){}; \\
                                   &                                    &                                         & \node(ld-pif-open){};            & \node(pif)[draw,anchor=base]{PIF}; \\
                                   & \node(idp-pif-init-1){};           &                                         &                                  & \node(pif-init-1){}; \\ [1pt]
                                   &                                    &                                         & \node(ld-pif-pms){};             & \node(pif-ld-pms){}; \\
                                   &                                    &                                         & \node(ld-pif-close){};           & \node(pif-end)[draw,anchor=base]{/PIF}; \\
\node(phase-auth-left){};          &                                    &                                         &                                  & \node(phase-auth-right){}; \\
                                   & \node(idp-ld-auth){};              &                                         & \node(ld-auth){}; & \\
\node(phase-ld-init-2-left){};     &                                    &                                         &                                  & \node(phase-ld-init-2-right){}; \\
\node(lpo-ld-init-2-top){};        & \node(idp-ld-init-2-top){};        & \node(rpdoc-ld-init-2-top){};           & \node(ld-ld-init-2-top){};       & \node(pif-ld-init-2-top){}; \\
\node(lpo-ld-init-2-btm){};        & \node(idp-ld-init-2-btm){};        & \node(rpdoc-ld-init-2-btm){};           & \node(ld-ld-init-2-btm){};       & \node(pif-ld-init-2-top){}; \\
\node(phase-prov-cont-left){};     &                                    &                                         &                                  & \node(phase-prov-cont-right){}; \\
\node(lpo-prov-cont-top){};        & \node(idp-prov-cont-top){};        & \node(rpdoc-prov-cont-top){};           & \node(ld-prov-cont-top){};       & \node(pif-prov-cont-top){}; \\
                                   &                                    &                                         & \node(ld-prov-cont-key-pair){};  & \\ [3ex]
                                   &                                    &                                         & \node(ld-prov-cont-pkb){};       & \node(pif-prov-cont-pkb){};\\
                                   & \node(idp-prov-cont-req-uc){};     &                                         &                                  & \node(pif-prov-cont-req-uc){};\\
                                   & \node(idp-prov-cont-cert-uc){};    &                                         &                                  & \\
                                   & \node(idp-prov-cont-send-uc){};    &                                         &                                  & \node(pif-prov-cont-send-uc){};\\
                                   &                                    &                                         & \node(ld-prov-cont-recv-uc){};   & \node(pif-prov-cont-recv-uc){};\\
\node(lpo-prov-cont-btm){};        & \node(idp-prov-cont-btm){};        & \node(rpdoc-prov-cont-btm){};           & \node(ld-prov-cont-btm){};       & \node(pif-prov-cont-btm){}; \\
\node(phase-auth-lpo-left){};      &                                    &                                         &                                  & \node(phase-auth-lpo-right){}; \\
                                   &                                    &                                         & \node(ld-gen-cap-lpo){};         & \\ [3ex]
\node(lpo-ld-auth){};              &                                    &                                         & \node(ld-lpo-auth){};            & \\ [-2pt]
\node(phase-cap-left){};           &                                    &                                         &                                  & \node(phase-cap-right){}; \\ [2pt]
\node(lpo-ld-list-emails){};       &                                    &                                         & \node(ld-lpo-list-emails){};     & \\ [2pt]
\node(lpo-ld-addr-info-3){};       &                                    &                                         & \node(ld-lpo-addr-info-3){};     & \\
                                   &                                    &                                         & \node(ld-gen-cap){};             & \\ [3ex]
                                   &                                    & \node(rpdoc-ld-cap){};                  & \node(ld-rpdoc-cap){};           & \\
                                   &                                    & \node(rpdoc-ld-close){};                & \node[draw,anchor=base](ld-end){/LD}; & \\
\node(lpo-end){};                  & \node(idp-end){};                  & \node(rpdoc-end){};                     &                                  & \node(pif-col-end){}; \\
};

\node(phase-1-label) at ($(phase-prov-left)!0.5!(phase-ld-init-1-left)$) [left=1ex] {\protophase{ld-start-1}};
\node(phase-2-label) at ($(phase-auth-left)!0.5!(phase-prov-left)$) [left=1ex] {\protophase{ld-prov-1}};
\node(phase-3-label) at ($(phase-ld-init-2-left)!0.5!(phase-auth-left)$) [left=1ex] {\protophase{ld-auth}};
\node(phase-4-label) at ($(phase-prov-cont-left)!0.5!(phase-ld-init-2-left)$) [left=1ex] {\protophase{ld-start-2}};
\node(phase-5-label) at ($(phase-auth-lpo-left)!0.5!(phase-prov-cont-left)$) [left=1ex] {\protophase{ld-prov-2}};
\node(phase-6-label) at ($(phase-cap-left)!0.5!(phase-auth-lpo-left)$) [left=1ex] {\protophase{ld-lpo-auth}};
\node(phase-7-label) at ($(lpo-end)!0.5!(phase-cap-left)$) [left=1ex] {\protophase{ld-cap}};

\tikzstyle{xhrArrow} = [color=blue,decoration={markings, mark=at position 1 with {\arrow[color=blue]{triangle 45}}}, preaction = {decorate}]

\draw [->,snake=snake,segment amplitude=0.2ex] (rpdoc-ld-open.40) to node [above=-2pt] {\protostep{ld-open} open} (ld);

\draw [->] (ld-init-1.160) to node [above=-2pt]{\protostep{ld-init-1} GET LD} (lpo-ld-init-1.20);
\draw [->] (lpo-ld-init-1.340) -- (ld-init-1.200);

\draw [->,color=red,dashed] (ld-rpdoc-send-ready) to node [above=-2pt]{\protostep{ld-rpdoc-ready-1} ready} (rpdoc-ld-recv-ready);

\draw [->,color=red,dashed] (rpdoc-ld-send-request) to node [above=-2pt]{\protostep{rpdoc-ld-request-1} request} (ld-rpdoc-recv-request);

\draw [->,color=blue,>=latex] (ld-lpo-ctx-1.160) to node [above=-2pt]{\protostep{ld-ctx-1} GET session\_context} (lpo-ld-ctx-1.20);
\draw [->,color=blue,>=latex] (lpo-ld-ctx-1.340) -- (ld-lpo-ctx-1.200);

\node (ld-user-email-drawn) at (ld-user-email) [draw,rounded corners,fill=Gainsboro]{\protostep{ld-user-email} email address };

\draw [->,color=blue,>=latex] (ld-lpo-addr-info-1-top) to node [above=-2pt]{\protostep{ld-addrinfo-1} GET address\_info} (lpo-ld-addr-info-1-top);

\draw [->] (lpo-idp-wk-1.20) to node [above=-2pt]{\protostep{lpo-idp-wk-1} GET wk} (idp-lpo-wk-1.160);
\draw [->] (idp-lpo-wk-1.200) -- (lpo-idp-wk-1.340);

\draw [->,color=blue,>=latex] (lpo-ld-addr-info-1-btm) to node [above=-2pt]{\protostep{ld-addrinfo-1-resp}} (ld-lpo-addr-info-1-btm);

\draw [->,snake=snake,segment amplitude=0.2ex] (ld-pif-open.40) to node [above=-2pt] {\protostep{ld-pif-open-1} create} (pif);

\draw [->] (pif-init-1.160) to node [above=-2pt]{\protostep{pif-init-1} GET PIF} (idp-pif-init-1.20);
\draw [->] (idp-pif-init-1.340) -- (pif-init-1.200);

\draw [implies-implies,color=red,dashed,double] (pif-ld-pms) to node [above=-2pt]{\protostep{pif-ld-pms-1} PMs} (ld-pif-pms);

\draw [->,snake=snake,segment amplitude=0.2ex] (ld-pif-close.40) to node [above=-2pt] {\protostep{ld-pif-close-1} close} (pif-end);

\node (ld-auth-drawn) at (ld-auth) [draw]{\protostep{idp-ld-auth} auth IdP};
\draw [implies-implies,double] (ld-auth-drawn) -- (idp-ld-auth);

\node[right of=phase-4-label,node distance=5em] {repeat \resizebox{!}{0.8\baselineskip}{\refprotophase{ld-start-1}}};

\node[right of=phase-5-label,node distance=3em,rotate=90] {repeat \resizebox{!}{0.8\baselineskip}{\refprotophase{ld-prov-1}}};

\node at (pif-prov-cont-top) [draw,fill=Gainsboro] {PIF};

\node at (ld-prov-cont-key-pair) [draw,rounded corners,fill=Gainsboro]{\protostep{gen-key-pair} gen. key pair};

\draw [->,color=red,dashed] (ld-prov-cont-pkb) to node [above=-2pt]{\protostep{pubkey-ld-pif} $\text{pk}_\text{b}$, email} (pif-prov-cont-pkb);

\draw [->,color=blue,>=latex] (pif-prov-cont-req-uc) to node [above=-2pt]{\protostep{req-uc} $\text{pk}_\text{b}$, email} (idp-prov-cont-req-uc);

\node (certify-uc) at (idp-prov-cont-cert-uc) [draw,rounded corners,fill=white]{\protostep{certify-uc} create UC};

\draw [->,color=blue,>=latex] (idp-prov-cont-send-uc) to node [above=-2pt]{\protostep{send-uc} UC} (pif-prov-cont-send-uc);

\draw [->,color=red,dashed] (pif-prov-cont-recv-uc) to node [above=-2pt]{\protostep{recv-uc} UC} (ld-prov-cont-recv-uc);

\node at (pif-prov-cont-btm) [draw,fill=Gainsboro] {/PIF};

\node at (ld-gen-cap-lpo) [draw,rounded corners,fill=Gainsboro]{\protostep{ld-gen-cap-lpo} gen. $\text{IA}_\text{LPO}$};

\draw [->,color=blue,>=latex] (ld-lpo-auth.160) to node [above=-2pt]{\protostep{ld-lpo-auth} POST auth\_with\_assertion ($\text{CAP}_\text{LPO}$)} (lpo-ld-auth.20);
\draw [->,color=blue,>=latex] (lpo-ld-auth.340) -- (ld-lpo-auth.200);

\draw [->,color=blue,>=latex] (ld-lpo-list-emails.160) to node [above=-2pt]{\protostep{ld-lpo-list-emails} GET list\_emails} (lpo-ld-list-emails.20);
\draw [->,color=blue,>=latex] (lpo-ld-list-emails.340) -- (ld-lpo-list-emails.200);

\draw [->,color=blue,>=latex] (ld-lpo-addr-info-3.160) to node [above=-2pt]{\protostep{ld-addrinfo-3} GET address\_info} (lpo-ld-addr-info-3.20);
\draw [->,color=blue,>=latex] (lpo-ld-addr-info-3.340) -- (ld-lpo-addr-info-3.200);

\node at (ld-gen-cap) [draw,rounded corners,fill=Gainsboro]{\protostep{ld-gen-cap} gen. $\text{IA}_\text{RP}$};

\draw [->,color=red,dashed] (ld-rpdoc-cap) to node [above=-2pt]{\protostep{ld-rpdoc-cap} response ($\text{CAP}_\text{RP}$)} (rpdoc-ld-cap);

\draw [->,snake=snake,segment amplitude=0.2ex] (rpdoc-ld-close.40) to node [above=-2pt] {\protostep{ld-close} close} (ld-end);

\begin{pgfonlayer}{background}
 \node (rpdoc-a) [above of=rpdoc, node distance=2ex]{};
 \node (rpdoc-al) [left of=rpdoc-a, node distance=6ex]{};
 \node (pif-col-end-b) [below of=pif-col-end, node distance=2ex]{};
 \node (pif-col-end-br) [right of=pif-col-end-b, node distance=4ex]{};
 \filldraw [color=Gainsboro,rounded corners] (rpdoc-al) rectangle (pif-col-end-br);

 \draw [color=gray] (lpo.270) -- (lpo-ld-init-2-top);
 \draw [color=gray,dotted,thick] (lpo-ld-init-2-top) -- (lpo-ld-init-2-btm);
 \draw [color=gray] (lpo-ld-init-2-btm) -- (lpo-prov-cont-top);
 \draw [color=gray,dotted,thick] (lpo-prov-cont-top) -- (lpo-prov-cont-btm);
 \draw [color=gray] (lpo-prov-cont-btm) -- (lpo-end);
 \draw [color=gray] (idp.270) -- (idp-ld-init-2-top);
 \draw [color=gray,dotted,thick] (idp-ld-init-2-top) -- (idp-ld-init-2-btm);
 \draw [color=gray] (idp-ld-init-2-btm) -- (idp-prov-cont-top);
 \draw [color=gray,dotted,thick] (idp-prov-cont-top) -- (idp-prov-cont-btm);
 \draw [color=gray] (idp-prov-cont-btm) -- (idp-end);
 \draw [color=gray] (rpdoc.270) -- (rpdoc-ld-init-2-top);
 \draw [color=gray,dotted,thick] (rpdoc-ld-init-2-top) -- (rpdoc-ld-init-2-btm);
 \draw [color=gray] (rpdoc-ld-init-2-btm) -- (rpdoc-prov-cont-top);
 \draw [color=gray,dotted,thick] (rpdoc-prov-cont-top) -- (rpdoc-prov-cont-btm);
 \draw [color=gray] (rpdoc-prov-cont-btm) -- (rpdoc-end);
 \draw [color=gray] (ld.270) -- (ld-auth-drawn);
 \draw [color=gray] (ld-auth-drawn) -- (ld-ld-init-2-top);
 \draw [color=gray,dotted,thick] (ld-ld-init-2-top) -- (ld-ld-init-2-btm);
 \draw [color=gray] (ld-ld-init-2-btm) -- (ld-prov-cont-top);
 \draw [color=gray,dotted,thick] (ld-prov-cont-top) -- (ld-prov-cont-btm);
 \draw [color=gray] (ld-prov-cont-btm) -- (ld-end);
 \draw [color=gray] (pif.270) -- (pif-end);
 \draw [color=gray,dotted,thick] (pif-prov-cont-top) -- (pif-prov-cont-btm);
\end{pgfonlayer}

\draw [dashed] (phase-ld-init-1-left.180) -- (phase-ld-init-1-right.0);
\draw [dashed] (phase-prov-left.180) -- (phase-prov-right.0);
\draw [dashed] (phase-auth-left.180) -- (phase-auth-right.0);
\draw [dashed] (phase-ld-init-2-left.180) -- (phase-ld-init-2-right.0);
\draw [dashed] (phase-prov-cont-left.180) -- (phase-prov-cont-right.0);
\draw [dashed] (phase-auth-lpo-left.180) -- (phase-auth-lpo-right.0);
\draw [dashed] (phase-cap-left.180) -- (phase-cap-right.0);

\node at ($(ld-top)!0.7!(pif-top)$) {Browser};

\end{tikzpicture}%

\raisebox{0.5ex}{\tikz{\draw [->] (0,0) -- (0.4,0);}}
   HTTPS messages, \raisebox{0.5ex}{\tikz{\draw
       [->,color=blue,>=latex] (0,0) -- (0.4,0);}} \xhrs (over HTTPS),
   \raisebox{0.5ex}{\tikz{\draw [->,color=red,dashed] (0,0) --
       (0.4,0);}} \pms, \raisebox{0.5ex}{\tikz{\draw
       [->,snake=snake,segment length=2ex,segment amplitude=0.2ex] (0,0) -- (0.4,0);}}
   browser commands%
 \caption[foo]{Simplified BrowserID implementation
   overview. CIF omitted for brevity.}
\label{fig:browserid-lowlevel-ld}

\end{figure}

%% file: analysis.tex

\section{Analysis of BrowserID: Authentication
  Properties}\label{sec:analysisbrowserid}

In this section, we present the analysis of the BrowserID system with
primary IdPs and with respect to authentication properties. As already
mentioned, in \cite{FettKuestersSchmitz-SP-2014}, we analyzed the
simpler case with a secondary IdP. We first, in
Section~\ref{sec:modelbrowseridpidp}, describe our model of BrowserID
with primary IdPs, with two central authentication properties one
would expect any SSO system to satisfy formalized in
Section~\ref{sec:securitypropsBrowserID}.  Due to the many differences
between the secondary and primary mode as described above, the model
had to be written from scratch in most parts. As mentioned in the
introduction, during the analysis of BrowserID it turned out that one
of the security properties is not satisfied and that in fact there is
an attack on BrowserID. We confirmed that this attack, which was
acknowledged by Mozilla, works on the actual implementation of
BrowserID. In Section~\ref{sec:attackbrowserid}, the attack is
presented along with a fix. (Our BrowserID model presented in
Appendix~\ref{sec:analysisbrowserid-pidp} contains this fix.) In
Section~\ref{sec:securefixedBrowserID}, we prove that the fixed
BrowserID system with primary IdPs satisfies both authentication
properties.

%% file: modelling-browserid.tex
\subsection{Modeling of BrowserID with Primary IdPs}\label{sec:modelbrowseridpidp}

We model the BrowserID system with primary IdPs as a web
system (in the sense of Section~\ref{sec:webmodel}). Note
that while in Section~\ref{sec:browserid} we give only a
brief overview of the BrowserID system, our modeling and
analysis considers the complete system with primary IdPs,
where we have extracted the model from the BrowserID source
code~\cite{mozilla/persona/source}.

We call a web system $\bidwebsystem=(\bidsystem,
\scriptset, \mathsf{script}, E_0)$ a \emph{BrowserID web
  system} if it is of the form described in
Appendix~\ref{sec:analysisbrowserid-pidp} and briefly
outlined here.

The system $\bidsystem=\mathsf{Hon}\cup \mathsf{Web} \cup
\mathsf{Net}$ consists of the (network) attacker process
$\fAP{attacker}$, the web server for $\fAP{LPO}$, a finite
set $\fAP{B}$ of web browsers, a finite set $\fAP{RP}$ of
web servers for the relying parties, and a finite set
$\fAP{IDP}$ of web servers for the identity providers, with
$\mathsf{Hon} := \fAP{B} \cup \fAP{RP} \cup \fAP{IDP} \cup
\{\fAP{LPO}\}$, $\mathsf{Web} := \emptyset$, and
$\mathsf{Net} := \{\fAP{attacker}\}$. DNS servers are
assumed to be dishonest, and hence, are subsumed by
$\fAP{attacker}$. IdPs and RPs can become corrupted
(similar to browsers, by a special message); LPO is assumed
to be honest.

The set $\addresses$ of IP addresses (see
Section~\ref{sec:communicationmodel}) contains one address
for each party in $\bidsystem$.  The set
$\dns\subseteq \mathbb{S}$  contains one or more domains
for each party in $\bidsystem$, except for browsers.

The definition of the processes in $\websystem$ follows the
description in Section~\ref{sec:javascript-descr}. 
For RP, we
explicitly follow the security considerations in
\cite{mozilla/persona/mdn} (Cross-site Request Forgery
protection, e.g., by checking origin headers and HTTPS only
with STS enabled). When RP receives a valid CAP (see
below), RP responds with a fresh \emph{RP service token for
  ID $i$} where $i$ is the ID (email address) for which the
CAP was issued. Intuitively, a client having such a token
can use the service of the RP.

Each browser $b\in \fAP{B}$ owns a set of email addresses
(identities) of the form $\left<name,d\right>$ with
$name\in \mathbb{S}$ and $d\in \dns$ (belonging to an IdP) and associated passwords
(i.e., nonces). 

A UC $\mi{uc}$ for a user $u$ with email address
$\an{\mi{name}, d}$ and public key (verification key)
$\pub(k_u)$, where $d \in \mapDomain(y)$ is a domain of the
IdP $y$ that issued the UC and $k_u$ is the private
(signing) key of $u$, is a term of the form
$\mi{uc}=\sig{\an{\an{\mi{name}, d},
    \pub(k_u)}}{\mapSignKey(y)}$, with $\mapSignKey(y)$
being the signing key of $y$. An IA $ia$ for an origin
$\mi{o}$ is a message of the form $ia=\sig{o}{k_u}$. A CAP
is of the form $\an{\mi{uc},\mi{ia}}$. Note that time
stamps are omitted both from the UC and the IA, modeling
that UC and IA never expire. In reality, as explained
in Section~\ref{sec:browserid}, they are valid for a
certain period of time. So our modeling is a safe
overapproximation.

The set $\scriptset$ of $\bidwebsystem$ contains six
scripts, with their string representations defined by
$\mathsf{script}$: the honest scripts running in RP-Doc,
CIF, LD, AD, and PIF, respectively, and the malicious
script $\Rasp$. The scripts for CIF and LD (issued by
$\fAP{LPO}$) are defined in a straightforward way following
the implementation outlined in
Section~\ref{sec:browserid}. The script for RP-Doc (issued
by RP) also includes the script that is (in reality) loaded
from LPO. In particular, this script creates the CIF and
the LD (sub)windows, whose contents (scripts) are
loaded from LPO. The scripts for the AD and PIF are modeled
following the example implementation provided by Mozilla
\cite{mozilla/persona/source}. Full formal specifications
of all the above mentioned scripts are provided in
Appendix~\ref{sec:analysisbrowserid-pidp}.

%% file: securityproperties-informal.tex
\subsection{Authentication Properties of the BrowserID System}\label{sec:securitypropsBrowserID}

While the documentation of BrowserID does not contain explicit
security goals, here we state two fundamental authentication
properties every SSO system should satisfy. These properties are
adapted from \cite{FettKuestersSchmitz-SP-2014}.

Informally, these properties can be stated as follows: \textbf{(A)}~\emph{The attacker should not be able to use a service of RP as an
  honest user.} In other words, the attacker should not get hold of
(be able to derive from his current knowledge) an RP service token for
an ID of an honest user (browser), even if the browser was closed and
then later used by a malicious user (i.e., after a
$\closecorrupt$). \textbf{(B)}~\emph{The attacker should not be able
  to authenticate an honest browser to an RP with an ID that is not
  owned by the browser (identity injection).} We refer the reader to
Appendix~\ref{app:form-secur-prop} for the formal definition of these
properties.

We call a BrowserID web system $\bidwebsystem$ \emph{secure
  (w.r.t.~authentication)} if the above conditions are satisfied in
all runs of the system.

%% file: attacks.tex
\subsection{Identity Injection Attack on BrowserID with Primary IdPs}\label{sec:attackbrowserid} 

While trying to prove the above mentioned authentication
properties of BrowserID with primary IdPs in our model, we
discovered a serious attack, which is sketched below and
does not apply to the case with secondary IdPs. We
confirmed the attack on the actual implementation and
reported it to
Mozilla~\cite{mozilla-browserid-primary-identity-injection-bug-report},
who acknowledged it.

During the provisioning phase \refprotophase{ld-prov-2}
(see Figure~\ref{fig:browserid-lowlevel-ld}), the IdP
issues a UC for the user's identity and public key provided
in \refprotostep{pubkey-ld-pif}. This UC is sent to
the LD by the PIF in \refprotostep{recv-uc}. 

If the IdP is malicious, it can issue a UC with different
data. In particular, it could replace the email address by
a different one, but keep the original public key. This
(malicious) UC is then later included in the CAP by LD. The
CAP will still be valid, because the public key is
unchanged. Now, as the RP determines the user's identity by
the UC contained in the CAP, RP issues a service token for
the spoofed email address. As a result, the honest user
will use RP's service (and typically will be logged in to
RP) under an ID that belongs to the attacker, which, for
example, could allow the attacker to track actions of the
honest user or obtain user secrets.  This violates
Condition~\textbf{(B)}.

To fix this problem, upon receipt of the UC in
\refprotostep{recv-uc}, LD should check whether it contains
the correct email address and public key, i.e., the one
requested by LD in \refprotostep{pubkey-ld-pif}. The same
is true for the CIF, which behaves similarly to the LD. The
formal model of BrowserID presented in
Appendix~\ref{sec:analysisbrowserid-pidp} contains these
fixes.

%% file: theorem.tex
\subsection{Security of the Fixed System}\label{sec:securefixedBrowserID}

For the fixed BrowserID system with primary IdPs, we have
proven the following theorem, which says that a fixed
BrowserID web system (i.e., the system where the above
described fix is applied) satisfies the security properties
\textbf{(A)} and \textbf{(B)}.

\begin{theorem}\label{thm:secur-fixed-syst}
  Let $\bidwebsystem$ be a fixed BrowserID web
  system. Then, $\bidwebsystem$ is secure (w.r.t.~authentication).
\end{theorem}
We prove Conditions \textbf{(A)} and \textbf{(B)}
separately. For both conditions, we assume that they are
not satisfied and lead this to a contradiction. In our
proofs, we make use of the general security properties of
the web model presented in
Section~\ref{sec:generalproperties}, which helped a lot in
making the proof for the primary IdP model more modular and
concise.  The complete proof with all details is provided
in Appendix~\ref{app:proofbrowserid-pidp}.

%% file: privacy-analysis.tex
\section{Privacy of BrowserID}\label{sec:privacyanalysis-pidp}

In this section, we study the privacy guarantees of the BrowserID
system with primary IdPs. Regarding privacy, Mozilla states that
``...the BrowserID protocol never leaks tracking information back to
the Identity Provider.''~\cite{mozilla/persona/faq} and ``Unlike other
sign-in systems, BrowserID does not leak information back to any
server [\dots] about which sites a user
visits.''~\cite{mozilla/persona/introducing-browserid}.\footnote{Clearly,
  in the current state of BrowserID a malicious LPO server could gather
  information about users' log in history. However, an integration of
  the code currently delivered by LPO into the browser, as envisioned,
  would avoid this issue. Currently, Mozilla's LPO needs to be
  trusted.}  While this is not a formal definition of the level of
privacy that BrowserID is supposed to provide, these and other
statements\footnote{see, e.g.,
  \url{https://developer.mozilla.org/en-US/Persona/Why_Persona} and \url{http://identity.mozilla.com/post/7669886219}.}
make it certainly clear that, unlike for other SSO systems, IdPs
should not be able to learn to which RPs their users log in.

In the process of formalizing this intuition in our model of BrowserID
and trying to prove this property, we found severe attacks against the
privacy of BrowserID which made clear that BrowserID does not provide
even a rather weak privacy property in the presence of a malicious
IdP.  Intuitively, the property says that a malicious IdP (which acts
as a web attacker) should not be able to tell whether a user logs in
at an honest RP $r$ or some other honest RP $r'$. In other words, a
run in which the user logs in at $r$ at some point should be
indistinguishable (from the point of view of the IdP) from the run in
which the user logs in at $r'$ instead.  Indistinguishability means
that the two sequences of messages received by the web attacker in the
two runs are statically equivalent in the usual sense of Dolev-Yao
models (see \cite{AbadiFournet-POPL-2001}), i.e., a Dolev-Yao attacker
cannot distinguish between the two sequences. Details of the privacy
definition are not important here since our attacks clearly show that
privacy is broken for any reasonable definition of
privacy. Unfortunately, our attacks are not caused by a simple
implementation error, but rather a fundamental design flaw in the
BrowserID protocol. Fixes for this flaw are conceivable, but not
without major changes to the design of BrowserID as discussed in
Section~\ref{sec:priv-privacy-fix}. Such a redesign of BrowserID and a
proof of privacy of the redesigned system are therefore out of the
scope of this paper, which focuses on the existing and deployed
version of BrowserID.

\input{figure-privacy-attack}

\subsection{Privacy Attacks on BrowserID}
\label{sec:priv-attack-brows}

For our attacks to work, it suffices that the IdP is a web
attacker. They work even if all DNS servers, RPs, and LPO are honest,
and all parties use encrypted connections. In what follows, we present
several variants of attacks on privacy.

\paragraph{PostMessage-Based  Attack.} The adversary is a malicious IdP that is interested to learn whether a user is logged in at
RP $r$. Figure~\ref{fig:privacy-attack} illustrates the main steps:

\noindent
\textbf{Step~\refattackstep{privacy-attack-1}.} First, the victim
visits her IdP. In BrowserID, email providers serve as IdPs, and
therefore it is not unlikely that a user visits this web site (e.g.,
for checking email or to use other services). As the IdP usually has
some cookie set at the user's browser, it learns the identity of the
victim. The IdP now creates a hidden iframe containing the login page
of $r$. 

\noindent
\textbf{Step~\refattackstep{privacy-attack-2}.} The login page
of $r$ (now loaded as an iframe within IdP's web site) includes and
runs the BrowserID script. As defined in the BrowserID protocol, the
script creates the communication iframe (see ``Automatic CAP
Creation'' in Section~\ref{sec:javascript-descr}), which in turn
checks whether the email address is marked as logged in at $r$ in the
localStorage of the user's browser. Only then it will try
to create a new CAP, for which it needs a PIF (the
same as in Phase~\refprotophase{ld-prov-1} in
Figure~\ref{fig:browserid-lowlevel-ld}).

\noindent
\textbf{Step~\refattackstep{privacy-attack-3}.} The PIF is loaded from
the IdP. Note that from this action alone, the IdP does not learn
where the user wants to log in. However, instead of the original
(honest) PIF document, the IdP can send a modified one that sends a
postMessage to the parent of the parent of the parent of its own
window, which in this setting is the IdP document that was opened by
the user in Step~\refattackstep{privacy-attack-1}. When the IdP
receives this message in the document from
Step~\refattackstep{privacy-attack-1}, it knows that the PIF was
loaded, and therefore, that the user is currently logged in at $r$.

Note that the IdP can repeatedly apply the above as long as the user
stays on the IdP's web site. During this period, the IdP can see
whether or not the user is logged in at the targeted RP. Clearly, the
IdP can simultaneously run the attack for different RPs in order to
track the user's login status for all such RPs. In particular, the IdP
can distinguish whether a user is logged in at RP $r$ or $r'$, which
violates the privacy property sketched above. In our formal model, the
malicious IdP would run the attacker script $\Rasp$ in
\nolinkurl{idp.com/index} and in \nolinkurl{idp.com/pif} (see
Figure~\ref{fig:privacy-attack}) in order to carry out the attack.

\paragraph{Variant 1: Waiting for UC requests.} The IdP first acts as
in Step~\refattackstep{privacy-attack-1}. Now, it could passively wait
for incoming requests for the PIF document or UC requests on its
server, which tell the IdP that a provisioning flow (probably
initiated by Step~\refattackstep{privacy-attack-1}) was started. This
variant cannot be executed in parallel and is less reliable in
practice, though.

\paragraph{Variant 2: PIF as Attack Source.} 
Step~\refattackstep{privacy-attack-1} can also be launched from within
a PIF itself (i.e., the PIF also takes the role of
\nolinkurl{idp.com/index} above). This way, while the user logs in at
some $r_1$, the IdP could check whether the user is logged in at
$r_2$, for any $r_2$.

\paragraph{Variant 3: Scanning the Window Structure (I).} Instead of using
a postMessage to alert the IdP's outer document about the existence of
the inner PIF document, the outer document could as well repeatedly
scan the window tree of the iframe containing $r$'s web site: While
the IdP sees almost no information about $r$'s document in the
iframe (as it is not same origin), it can see the list of subwindows
(i.e., the CIF, and possibly other iframes). For these frames, again,
it would see the subwindows, especially the PIF, which it could
identify uniquely by checking whether it is same origin with the IdPs
outer window.

\paragraph{Variant 4: Scanning the Window Structure (II).} In Variant
2, using a same-origin check, the malicious IdP can uniquely identify
the PIF in the window structure. This same-origin check could be
skipped and it could only be checked whether a PIF is generated, based
on the window structure alone. While this is less reliable, this
attack could be launched by \emph{any} third party web attacker (not
only the IdP to which the user's email address belongs) to check
whether the victim is logged in at $r$ or not.

\smallskip

We verified (all variants of) the attacks in our model as well as in a
real-world BrowserID setup. Implementing proofs-of-concept required
only a few lines of (trivial) JavaScript. In most attack
variants, we directly or indirectly use the structure of the windows
inside the web browser as a side channel. To our knowledge, this is
the first description of this side channel for breaking privacy in
browsers. The attacks have been reported to and confirmed by Mozilla~\cite{mozilla-browserid-primary-privacy-bug-report}.

\subsection{Fixing the Privacy of BrowserID}
\label{sec:priv-privacy-fix}

Fixing the privacy of BrowserID seems to require a substantial
redesign of the system. Regarding the presented attacks, BrowserID's
main weakness is the window structure. The most obvious mitigation,
modifying the CIF such that it always creates the PIF (even if the
user has not logged in before), does not work: To open the PIF, the
CIF looks up (in the localStorage) the user's identity at the current
RP to derive the address of the PIF. If the
user has not logged in before, this information is not available.

Another approach would be to use cross-origin \xhrs to replace the
features of the PIF. This solution would require a major revision in
the inner workings of BrowserID and would not protect against
Variant~1.

%% file: figure-privacy-attack.tex
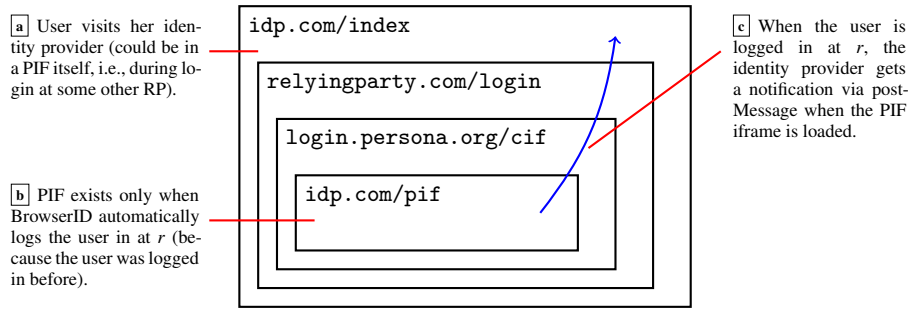
\begin{figure}[t!]
  \centering
\begin{tikzpicture}[scale=1]
\usetikzlibrary{arrows}

\draw[thick] (6,0) rectangle (0,4) node[anchor=north west] {\nolinkurl{idp.com/index}};

\draw[thick] (5.5,0.25) rectangle (0.25,3.25) node[anchor=north west] {\nolinkurl{relyingparty.com/login}};

\draw[thick] (5,0.5) rectangle (0.5,2.5) node[anchor=north west] {\nolinkurl{login.persona.org/cif}};

\draw[thick] (4.5,0.75) rectangle (0.75,1.75)  node[anchor=north west] {\nolinkurl{idp.com/pif}};

\draw [->, thick, color=blue, bend angle=15, bend right]  (4,1.25) to node[pos=0.35](arrowmid){} (5,3.6);

\draw [thick, color=red] (0.25,3.4) to (-0.4,3.4);

\draw [thick, color=red] (1,1.15) to (-0.4,1.15);

\draw [thick, color=red] (arrowmid) to (6.4,3.4);

\node(text-1) at (-3.15,4) [anchor=north west]{
  \begin{minipage}{2.5cm}
    \scriptsize
    \attackstep{privacy-attack-1}\hspace{1ex}User visits her identity provider (could be in a PIF itself, i.e., during login at some other RP). 
  \end{minipage}
};

\node(text-2) at (-3.15,1.75) [anchor=north west]{
  \begin{minipage}{2.5cm}
    \scriptsize
    \attackstep{privacy-attack-2}\hspace{1ex}PIF exists only when BrowserID automatically logs the user in at $r$ (because the user was logged in before).
  \end{minipage}
};

\node(text-3) at (6.45,4) [anchor=north west]{
  \begin{minipage}{2.3cm}
    \scriptsize \attackstep{privacy-attack-3}\hspace{1ex}When the user
    is logged in at $r$, the identity provider gets a notification via
    postMessage when the PIF iframe is loaded.
  \end{minipage}
};

\end{tikzpicture}
\caption{The three main steps of the privacy attack.
  Using a specially crafted PIF document, a malicious IdP can notify
  itself via postMessage when the user is logged in at some RP $r$.}

\label{fig:privacy-attack}
\end{figure}

%% file: related-work.tex
\section{Related Work}\label{sec:relatedwork}

The formal treatment of the security of the web
infrastructure and web applications based on this infrastructure is a
young discipline.  Of the few works in this area
even less are based on a general model that incorporates
essential mechanisms of the web.

Early works in formal web security analysis (see, e.g.,
\cite{kerschbaum-SP-2007-XSRF-prevention,Jackson-TACAS-2002-Alloy,ArmandoEtAl-FMSE-2008,SantsaiHaekyBeznosov-CS-2012,ChariJutlaRoy-IACR-2011})
are based on very limited models developed specifically for
the application under scrutiny. The first work to consider
a general model of the web, written in the finite-state
model checker Alloy, is the work by Akhawe et
al.\cite{AkhawBarthLamMitchellSong-CSF-2010}. Inspired by
this work, Bansal et
al.\cite{BansalBhargavanetal-POST-2013-WebSpi,BansalBhargavanMaffeis-CSF-2012}
built a more expressive model, called WebSpi, in ProVerif
\cite{Blanchet-CSFW-2001}, a tool for symbolic
cryptographic protocol analysis. These models have
successfully been applied to web standards and
applications. Recently, Kumar \cite{Kumar-RAID-2014}
presented a high-level Alloy model and applied it to SAML
single sign-on. However, compared to our model in \cite{FettKuestersSchmitz-SP-2014} and its extensions
considered here, on the one hand, all above mentioned
models are formulated in the specification languages of
specific analysis tools, and hence, are tailored towards
automation (while we perform manual
analysis). On the other hand, the models considered in
these works are much less expressive and precise. For
example, these models do not incorporate a precise handling
of windows, documents, or iframes; cross-document messaging
(postMessages) or session storage are not included at all.
In fact, several general web features and technologies that
have been crucial for the analysis of BrowserID are not
supported by these models, and hence, these models cannot
be applied to BrowserID. Moreover, the complexity of
BrowserID exceeds that of the systems analyzed in these
other works in terms of the use of web technologies and the
complexity of the protocols. For example, BrowserID in
primary mode is a protocol consisting of 48 different
(network and inter-frame) messages compared to typically
about 10--15 in the protocols analyzed in other models.

The BrowserID system in the primary mode has been analyzed
before using the AuthScan tool developed by Bai et
al.~\cite{baiLeietal-NDSS-2013-authscan}. Their work
focusses on the automated extraction of a model from a
protocol implementation. This tool-based analysis did not
reveal the identity injection attack, though; privacy
properties have not been studied there. Dietz and Wallach
demonstrated a technique to secure BrowserID when specific
flaws in TLS are considered~\cite{DietzWallach-NDSS-2014}.

%% file: conclusion.tex
\section{Conclusion}\label{sec:conclusion}

In this paper, we slightly extended our existing
web model, resulting in the most comprehensive model of the
web so far. It contains many security-relevant features and
is designed to closely mimic standards and specifications
for the web. As such, it constitutes a solid basis for the
analysis of a broad range of web standards and
applications.

Based on this model, we presented a detailed analysis
of the BrowserID SSO system in the  primary IdP
mode.  During the security proof of the fundamental
authentication requirements {\bf (A)} and {\bf (B)}, we
found a flaw in BrowserID that does not apply to its
secondary mode and leads to an identity injection attack,
and hence, violates property {\bf (B)}.  We confirmed the
attack on the actual BrowserID implementation and reported
it to Mozilla, who acknowledged it. We proposed a fix and
formally proved that the fixed system fulfills both {\bf
  (A)} and {\bf (B)}. Among the so far very few efforts on
formally analyzing web applications and standards in
expressive web models, our analysis constitutes the most
complex formal analysis of a web application to date. It
illustrates that (manual) security analysis of complex
real-world web applications in a detailed web model, while
laborious, is feasible and yields meaningful and
practically relevant results.

During an attempt to formally analyze the privacy promise
of the BrowserID system, we again found practical
attacks. These attacks have been reported to and confirmed
by Mozilla and, unfortunately, show that BrowserID would
have to undergo a substantial redesign in order to fulfill
its privacy promise. Interestingly, for our attacks we use
a side channel that exploits information about the
structure of windows in a browser. To the best of our
knowledge, such side channel attacks have not gained much
attention so far in the literature.

Finally, we have identified and proven important security
properties of general application independent web features
in order to facilitate future analysis efforts of web
standards and web applications in the web model.

%% file: appendix-webmodel.tex
\section{Communication Model}\label{app:communication-model}

Extending Section~\ref{sec:communicationmodel}, we here present
details and definitions on the basic concepts of the communication
model. For readability, some parts from
Section~\ref{sec:communicationmodel} are repeated.

\subsection{Terms, Messages and Events} 
The signature $\Sigma$ for the terms and
messages considered in this work is the union of the following
pairwise disjoint sets of function symbols:
\begin{itemize}
\item constants $C = \addresses\,\cup\, \mathbb{S}\cup
  \{\True,\bot,\notdef\}$ where the three sets are pairwise disjoint,
  $\mathbb{S}$ is interpreted to be the set of ASCII strings
  (including the empty string $\varepsilon$), and $\addresses$ is
  interpreted to be a set of (IP) addresses,\footnote{For brevity of presentation, in Section~\ref{sec:communicationmodel} the set $C$ contained also the set of nonces $\nonces$. Here nonces are considered separately (see Definition~\ref{def:terms}).}
\item function symbols for public keys, (a)symmetric
  en\-cryp\-tion/de\-cryp\-tion, and signatures: $\mathsf{pub}(\cdot)$,
  $\enc{\cdot}{\cdot}$, $\dec{\cdot}{\cdot}$, $\encs{\cdot}{\cdot}$,
  $\decs{\cdot}{\cdot}$, $\sig{\cdot}{\cdot}$,
  $\checksig{\cdot}{\cdot}$, and $\unsig{\cdot}$,
\item $n$-ary sequences $\an{}, \an{\cdot}, \an{\cdot,\cdot},
  \an{\cdot,\cdot,\cdot},$ etc., and
\item projection symbols $\pi_i(\cdot)$ for all $i \in \mathbb{N}$.
\end{itemize}

\begin{definition}\label{def:terms}
  Let $X=\{x_0,x_1,\dots\}$ be a set of variables and $\nonces$ be an
  infinite set of constants (\emph{nonces}) such that $\Sigma$, $X$,
  and $\nonces$ are pairwise disjoint. For $N\subseteq\nonces$, we
  define the set $\gterms_N(X)$ of \emph{terms} over $\Sigma\cup N\cup
  X$ inductively as usual: (1) If $t\in N\cup X$, then $t$ is a
  term. (2) If $f\in \Sigma$ is an $n$-ary function symbol in $\Sigma$
  for some $n\ge 0$ and $t_1,\ldots,t_n$ are terms, then
  $f(t_1,\ldots,t_n)$ is a term.
\end{definition}

By $\gterms_N=\gterms_N(\emptyset)$, we denote the set of all terms
over $\Sigma\cup N$ without variables, called \emph{ground terms}. The
set $\messages$ of messages (over $\nonces$) is defined to be the set
of ground terms $\gterms_{\nonces}$. 

\begin{example}
  For example, $k\in \nonces$ and $\pub(k)$ are messages, where $k$
  typically models a private key and $\pub(k)$ the corresponding
  public key. For constants $a$, $b$, $c$ and the nonce $k\in
  \nonces$, the message $\enc{\an{a,b,c}}{\pub(k)}$ is interpreted to
  be the message $\an{a,b,c}$ (the sequence of constants $a$, $b$,
  $c$) encrypted by the public key $\pub(k)$.
\end{example}

For strings (elements in $\mathbb{S}$), we use a
specific font. For example, $\cHttpReq$ and $\cHttpResp$
are strings. We denote by $\dns\subseteq \mathbb{S}$ the
set of domains, e.g., $\str{example.com}\in \dns$.  We
denote by $\methods\subseteq \mathbb{S}$ the set of methods
used in HTTP requests, e.g., $\mGet$, $\mPost\in \methods$.

The equational theory associated with the signature
$\Sigma$ is given in Figure~\ref{fig:equational-theory}.

\begin{figure}
\begin{align}
\dec{\enc x{\pub(y)}}{y} &= x\\
\decs{\encs x{y}}{y} &= x\\
\unsig{\sig{x}{y}} &= x\\
\checksig{\sig{x}{y}}{\pub(y)} &= \True\\
\pi_i(\an{x_1,\dots,x_n}) &= x_i \text{\;\;if\ } 1 \leq i \leq n \\
\proj{j}{\an{x_1,\dots,x_n}} &= \notdef \text{\;\;if\ } j
\not\in \{1,\dots,n\}\\
\proj{j}{t} &= \notdef \text{\;\;if $t$ is not a sequence}
\end{align}
\caption{Equational theory for $\Sigma$.}\label{fig:equational-theory}
\end{figure}

By $\equiv$ we denote the congruence relation on $\terms(X)$ induced
by this theory. For example, we have that
$\pi_1(\dec{\enc{\an{\str{a},\str{b}}}{\pub(k)}}{k})\equiv \str{a}$.

\begin{definition}
  An \emph{event (over $\addresses$ and $\messages$)} is of the form
  $(a{:}f{:}m)$, for $a, f\in \addresses$ and $m \in \messages$, where
  $a$ is interpreted to be the receiver address and $f$ is the sender
  address.  We denote by $\events$ the set of all events.
\end{definition}

\subsection{Atomic Processes, Systems and Runs} We here
define atomic processes, systems, and runs of systems.  

An atomic process takes its current state and an
event as input, and then (non-deterministi\-cally) outputs a new state
and a set of events.
\begin{definition}\label{def:atomic-process-and-process}
  A \emph{(generic) \ap} is a tuple $p = (I^p, Z^p, R^p,
  s^p_0)$ where $I^p \subseteq \addresses$, $Z^p$ is a set
  of states, $R^p\subseteq (\events \times Z^p) \times
  (2^\events \times Z^p)$, and $s^p_0\in Z^p$ is the
  initial state of $p$.  We write $(e,z)R(E,z')$ instead of
  $((e,z),(E,z'))\in R$.

  A \emph{system} $\process$ is a (possibly infinite) set
  of \aps.
\end{definition}

\begin{definition}
  A \emph{configuration of a system $\process$} is a tuple $(S, E)$
  where $S$ maps every atomic process $p\in \process$ to its current
  state $S(p)\in Z^p$ and $E$ is a (possibly infinite) multi-set of
  events waiting to be delivered.
\end{definition}

\begin{definition}
  A \emph{processing step of the system $\process$} is of the form
  \[(S,E) \xrightarrow[p \rightarrow E_{\text{out}}]{e \rightarrow p}
  (S', E')\] such that (1) there exists an event $e = (a{:}f{:}m) \in E$,
  $E_\text{out} \subseteq E'$, and a process $p \in \process$ with $(e,
  S(p))R^p(E_\text{out}, S'(p))$ and $a \in I^p$, (2) $S'(p') = S(p')$
  for all $p' \neq p$, and (3) $E' = (E\setminus \{e\}) \cup
  E_\text{out}$ (multi-set operations).  We may omit the superscript
  and/or subscript of the arrow.
\end{definition}

\begin{definition}
  Let $\process$ be a system and $E_0$ be a multi-set of events. A
  \emph{run $\rho$ of a system $\process$ initiated by $E_0$} is a
  finite sequence of configurations $(S_0, E_0),\dots,(S_n, E_n)$ or
  an infinite sequence of configurations $(S_0, E_0),\dots$ such that
  $S_0(p) = s_0^p$ for all $p \in \process$ and $(S_i, E_i)
  \xrightarrow{} (S_{i+1}, E_{i+1})$ for all $0 \leq i < n$ (finite
  run) or for all $i \geq 0$ (infinite run).
\end{definition}
\subsection{Atomic Dolev-Yao Processes}  We next define
atomic Dolev-Yao processes, for which we require that the
messages and states that they output can be computed (more
formally, derived) from the current input event and
state. For this purpose, we first define what it means to
derive a message from given messages.

\begin{definition}
  Let $N\subseteq \nonces$, $\tau \in \gterms_N(\{x_1,\ldots,x_n\})$,
  and $t_1,\ldots,t_n\in \gterms_N$. By
  $\tau[t_1\!/\!x_1,\ldots,t_n\!/\!x_n]$ we denote the (ground) term obtained
  from $\tau$ by replacing all occurrences of $x_i$ in $\tau$ by
  $t_i$, for all $i\in \{1,\ldots,n\}$.
\end{definition}

\begin{definition}
  Let $M\subseteq \messages$ be a set of messages. We say that \emph{a
    message $m$ can be derived from $M$ with nonces $N$} if there
  exist $n\ge 0$, $m_1,\ldots,m_n\in M$, and $\tau\in
  \gterms_N(\{x_1,\ldots,x_n\})$ such that $m\equiv
  \tau[m_1/x_1,\ldots,m_n/x_n]$. We denote by $d_N(M)$ the set of all
  messages that can be derived from $M$ with nonces $N$.
\end{definition}
For example, $a\in d_{\{k\}}(\{\enc{\an{a,b,c}}{\pub(k)}\})$.

\begin{definition} \label{def:adyp} An \emph{atomic Dolev-Yao process
    (or simply, a DY process)} is a tuple $p = (I^p, Z^p,$ $R^p,
  s^p_0, N^p)$ such that $(I^p, Z^p, R^p, s^p_0)$ is an atomic process
  and (1) $N^p\subseteq\nonces$ is an (initial) set of nonces, (2)
  $Z^p \subseteq \gterms_{\nonces}$ (and hence, $s^p_0\in
  \gterms_{\nonces}$), and (3) for all $a, a', f, f'\in \addresses$,
  $m, m', s, s'\in \gterms_{\nonces}$, set of events $E$ with
  $((a{:}f{:}m), s)R(E, s')$ and $(a'{:}f'{:}m') \in E$ it holds true
  that $m',s' \in d_N(\{m,s\})$. (Note that $a',f'\in d_N(\{m,s\})$.)
\end{definition}

\begin{definition}\label{def:atomicattacker}
  An \emph{(atomic) attacker process for a set of sender
    addresses $A\subseteq \addresses$} is an atomic DY
  process $p = (I, Z, R, s_0, N)$ such that for all $a,
  f\in \addresses$, $m\in \gterms_{\nonces}$, and $s\in
  Z$ we have that $((a{:}f{:}m), s)R(E,s')$ iff $s'=\an{\an{a,
      f, m}, s}$ and $E=\{(a'{:}f'{:}m')\mid a'\in \addresses$,
  $f'\in A$, $m'\in d_N(\{m,s\})\}$.
\end{definition}

\subsection{Scripting Processes}
We define scripting processes, which model client-side scripting
technologies, such as JavaScript. Scripting processes are defined
similarly to DY processes.

\begin{definition} \label{def:sp} A \emph{scripting
    process} (or simply, a \emph{script}) is a relation
  $R\subseteq (\terms \times 2^{\nonces})\times \terms$
  such that for all $s, s' \in \terms$ and $N\subseteq
  \nonces$ with $(s,N)\,R\,s'$ it follows that $s'\in
  d_N(s)$.
\end{definition}
A script is called by the browser which provides it with a
(fresh, infinite) set $N$ of nonces and state information
$s$. The script then outputs a term $s'$, which represents
the new internal state and some command which
is interpreted by the browser.

Similarly to an attacker process, we define the
\emph{attacker script} $\Rasp$. This script outputs
everything that is derivable from the input, i.e.,
$\Rasp=\{((s,N),s')\mid s\in \terms, N\subseteq \nonces,
s'\in d_N(s)\}$.

%% file: appendix-message-formats.tex
\section{Message and Data
  Formats}\label{app:message-data-formats}

We now provide some more details about data and message
formats that are needed for the formal treatment of the web
model and the analysis of BrowserID presented in the rest
of the appendix.

\subsection{Notations}\label{app:notation}

\begin{definition}[Sequence Notations]
  For a sequence $t = \an{t_1,\dots,t_n}$ and a set $s$ we
  use $t \subsetPairing s$ to say that $t_1,\dots,t_n \in
  s$.  We define $\left. x \inPairing t\right. \iff \exists
  i: \left. t_i = x\right.$.
  We write $t \plusPairing y$ to denote the sequence
  $\an{t_1,\dots,t_n,y}$. For a sequence $t=
  \an{t_1,\ldots,t_n}$ we define $|t| = n$. If $t$ is not a
  sequence, we set $|t| = \notdef$. For a finite set $M$
  with $M = \{m_1, \dots,m_n\}$ we use $\an{M}$ to denote
  the term of the form $\an{m_1,\dots,m_n}$. (The order of
  the elements does not matter; one is chosen arbitrarily.)
\end{definition}

\begin{definition}\label{def:dictionaries}
  A \emph{dictionary over $X$ and $Y$} is a term of the
  form \[\an{\an{k_1, v_1}, \dots, \an{k_n,v_n}}\] where
  $k_1, \dots,k_n \in X$, $v_1,\dots,v_n \in Y$, and the
  keys $k_1, \dots,k_n$ are unique, i.e., $\forall i\neq j:
  k_i \neq k_j$. We call every term $\an{k_i,v_i}$, $i\in
  \{1,\ldots,n\}$, an \emph{element} of the dictionary with
  key $k_i$ and value $v_i$.  We often write $\left[k_1:
    v_1, \dots, k_i:v_i,\dots,k_n:v_n\right]$ instead of
  $\an{\an{k_1, v_1}, \dots, \an{k_n,v_n}}$. We denote the
  set of all dictionaries over $X$ and $Y$ by $\left[X
    \times Y\right]$.
\end{definition}
We note that the empty dictionary is equivalent to the
empty sequence, i.e.,  $[] = \an{}$.  Figure
\ref{fig:dictionaries} shows the short notation for
dictionary operations that will be used when describing the
browser atomic process. For a dictionary $z = \left[k_1:
  v_1, k_2: v_2,\dots, k_n:v_n\right]$ we write $k \in z$ to
say that there exists $i$ such that $k=k_i$. We write
$z[k_j] := v_j$ to extract elements. If $k \not\in z$, we
set $z[k] := \an{}$.

\begin{figure}[htb!]\centering
  \begin{align}
    \left[k_1: v_1, \dots, k_i:v_i,\dots,k_n:v_n\right][k_i] = v_i%
  \end{align}\vspace{-2.5em}
  \begin{align}
    \nonumber \left[k_1: v_1, \dots, k_{i-1}:v_{i-1},k_i: v_i, k_{i+1}:v_{i+1}\dots,k_n: v_n\right]-k_i =\\
         \left[k_1: v_1, \dots, k_{i-1}:v_{i-1},k_{i+1}:v_{i+1}\dots,k_n: v_n\right]
  \end{align}
  \caption{Dictionary operators with $1\le i\le n$.}\label{fig:dictionaries}
\end{figure}

Given a term $t = \an{t_1,\dots,t_n}$, we can refer to any
subterm using a sequence of integers. The subterm is
determined by repeated application of the projection
$\pi_i$ for the integers $i$ in the sequence. We call such
a sequence a \emph{pointer}:

\begin{definition}\label{def:pointer}
  A \emph{pointer} is a sequence of non-negative
  integers. We write $\tau.\ptr{p}$ for the application of
  the pointer $\ptr{p}$ to the term $\tau$. This operator
  is applied from left to right. For pointers consisting of
  a single integer, we may omit the sequence braces for
  brevity.
\end{definition}

\begin{example}
  For the term $\tau = \an{a,b,\an{c,d,\an{e,f}}}$ and the
  pointer $\ptr{p} = \an{3,1}$, the subterm of $\tau$ at
  the position $\ptr{p}$ is $c =
  \proj{1}{\proj{3}{\tau}}$. Also, $\tau.3.\an{3,1} =
  \tau.3.\ptr{p} = \tau.3.3.1 = e$.
\end{example}

To improve readability, we try to avoid writing, e.g.,
$\compn{o}{2}$ or $\proj{2}{o}$ in this document. Instead,
we will use the names of the components of a sequence that
is of a defined form as pointers that point to the
corresponding subterms. E.g., if an \emph{Origin} term is
defined as $\an{\mi{host}, \mi{protocol}}$ and $o$ is an
Origin term, then we can write $\comp{o}{protocol}$ instead
of $\proj{2}{o}$ or $\compn{o}{2}$. See also
Example~\ref{ex:url-pointers}.

In the pseudocode, we will write, for example, 

\medskip

\begin{algorithmic}
  \LetST{$x,y$}{$\an{\str{Constant},x,y} \equiv
    t$}{doSomethingElse}
\end{algorithmic} \setlength{\parindent}{1em}

\medskip

\noindent for some variables $x,y$, a string
$\str{Constant}$, and some term $t$ to express that $x :=
\proj{2}{t}$, and $y := \proj{3}{t}$ if $\str{Constant}
\equiv \proj{1}{t}$ and if $|\an{\str{Constant},x,y}| =
|t|$,  and that otherwise
$x$ and $y$ are not set and doSomethingElse is executed.

\subsection{URLs}\label{sec:urls}

\begin{definition}
  A \emph{URL} is a term of the form $\an{\tUrl, \mi{protocol},
    \mi{host}, \mi{path}, \mi{params}}$ with $\mi{protocol}$
  $\in \{\http, \https\}$ (for \textbf{p}lain (HTTP) and
  \textbf{s}ecure (HTTPS)), $\mi{host} \in \dns$,
  $\mi{path} \in \mathbb{S}$ and $\mi{params} \in
  \dict{\mathbb{S}}{\terms}$. The set of all valid URLs
  is $\urls$.
\end{definition}

\begin{example} \label{ex:url-pointers}
  For the URL $u = \an{\tUrl, a, b, c, d}$, $\comp{u}{protocol} =
  a$. If, in the algorithm described later, we say $\comp{u}{path} :=
  e$ then $u = \an{\tUrl, a, b, c, e}$ afterwards. 
\end{example}

\subsection{Origins}\label{sec:origins}
\begin{definition} An \emph{origin} is a term of the form
  $\an{\mi{host}, \mi{protocol}}$ with $\mi{host} \in
  \dns$ and $\mi{protocol} \in \{\http, \https\}$. We write
  $\origins$ for the set of all origins.  See
  Example~\ref{ex:window} for an example of an
  origin.
\end{definition}

\subsection{Cookies}\label{sec:cookies}

\begin{definition} A \emph{cookie} is a term of the form
  $\an{\mi{name}, \mi{content}}$ where $\mi{name} \in
  \terms$, and $\mi{content}$ is a term of the form
  $\an{\mi{value}, \mi{secure}, \mi{session},
    \mi{httpOnly}}$ where $\mi{value} \in
  \terms$,  $\mi{secure}$, $\mi{session}$,
  $\mi{httpOnly} \in \{\True, \bot\}$.  We write $\cookies$
  for the set of all cookies.
\end{definition}

If the $\mi{secure}$ attribute of a cookie is set, the
browser will not transfer this cookie over unencrypted HTTP
connections. If the $\mi{session}$ flag is set, this cookie
will be deleted as soon as the browser is closed. The
$\mi{httpOnly}$ attribute controls whether JavaScript has
access to this cookie.

Note that cookies of the form described here are only
contained in HTTP(S) requests. In responses, only the
components $\mi{name}$ and $\mi{value}$ are transferred as
a pairing of the form $\an{\mi{name}, \mi{value}}$.

\subsection{HTTP Messages}\label{sec:http-messages-full}
\begin{definition}
  An \emph{HTTP request} is a term of the form shown in
  (\ref{eq:default-http-request}). An \emph{HTTP response}
  is a term of the form shown in
  (\ref{eq:default-http-response}).
  \begin{align}
    \label{eq:default-http-request}
    & \hreq{ nonce=\mi{nonce}, method=\mi{method},
      xhost=\mi{host}, xpath=\mi{path},
      parameters=\mi{parameters}, headers=\mi{headers},
      xbody=\mi{body}
    } \\
    \label{eq:default-http-response}
    & \hresp{ nonce=\mi{nonce}, status=\mi{status},
      headers=\mi{headers}, xbody=\mi{body} }
  \end{align}
  The components are defined as follows:
  \begin{itemize}
  \item $\mi{nonce} \in \nonces$ serves to map each
    response to the corresponding request 
  \item $\mi{method} \in \methods$ is one of the HTTP
    methods.
  \item $\mi{host} \in \dns$ is the host name in the HOST
    header of HTTP/1.1.
  \item $\mi{path} \in \mathbb{S}$ is a string indicating
    the requested resource at the server side
  \item $\mi{status} \in \mathbb{S}$ is the HTTP status
    code (i.e., a number between 100 and 505, as defined by
    the HTTP standard)
  \item $\mi{parameters} \in
    \dict{\mathbb{S}}{\terms}$ contains URL parameters
  \item $\mi{headers} \in \dict{\mathbb{S}}{\terms}$,
    containing request/response headers. The dictionary
    elements are terms of one of the following forms: 
    \begin{itemize}
    \item $\an{\str{Origin}, o}$ where $o$ is an origin
    \item $\an{\str{Set{\mhyphen}Cookie}, c}$ where $c$ is
      a sequence of cookies
    \item $\an{\str{Cookie}, c}$ where $c \in
      \dict{\mathbb{S}}{\terms}$ (note that in this header,
      only names and values of cookies are transferred)
    \item $\an{\str{Location}, l}$ where $l \in \urls$
    \item $\an{\str{Strict{\mhyphen}Transport{\mhyphen}Security},\True}$
    \end{itemize}
  \item $\mi{body} \in \terms$ in requests and responses. 
  \end{itemize}
  We write $\httprequests$/$\httpresponses$ for the set of
  all HTTP requests or responses, respectively.
\end{definition}

\begin{example}[HTTP Request and Response]
  \begin{align}
    \label{eq:ex-request}
    \nonumber \mi{r} := & \langle
                   \cHttpReq,
                   n_1,
                   \mPost,
                   \str{example.com},
                   \str{/show},
                   \an{\an{\str{index,1}}},\\ & \quad
                   [\str{Origin}: \an{\str{example.com, \https}}],
                   \an{\str{foo}, \str{bar}}
                \rangle \\
    \label{eq:ex-response} \mi{s} := & \hresp{ nonce=n_1,
      status=200,
      headers=\an{\an{\str{Set{\mhyphen}Cookie},\an{\an{\str{SID},\an{n_2,\bot,\bot,\True}}}}},
      xbody=\an{\str{somescript},x}}
  \end{align}
  \noindent
  An HTTP $\mGet$ request for the URL
  \url{http://example.com/show?index=1} is shown in
  (\ref{eq:ex-request}), with an Origin header and a body
  that contains $\an{\str{foo},\str{bar}}$. A possible
  response is shown in (\ref{eq:ex-response}), which
  contains an httpOnly cookie with name $\str{SID}$ and
  value $n_2$ as well as the string representation
  $\str{somescript}$ of the scripting process
  $\mathsf{script}^{-1}(\str{somescript})$ (which should be
  an element of $\scriptset$) and its initial state
  $x$.
\end{example}

\subsubsection{Encrypted HTTP
  Messages.} \label{sec:http-messages-encrypted-full}
For HTTPS, requests are encrypted using the public key of
the server.  Such a request contains an (ephemeral)
symmetric key chosen by the client that issued the
request. The server is supported to encrypt the response
using the symmetric key.

\begin{definition} An \emph{encrypted HTTP request} is of
  the form $\enc{\an{m, k'}}{k}$, where $k$, $k' \in
  \nonces$ and $m \in \httprequests$. The corresponding
  \emph{encrypted HTTP response} would be of the form
  $\encs{m'}{k'}$, where $m' \in \httpresponses$. We call
  the sets of all encrypted HTTP requests and responses
  $\httpsrequests$ or $\httpsresponses$, respectively.
\end{definition}

\begin{example}
  \begin{align}
    \label{eq:ex-enc-request} \ehreqWithVariable{r}{k'}{\pub(k_\text{example.com})} \\
    \label{eq:ex-enc-response} \ehrespWithVariable{s}{k'}
  \end{align} The term (\ref{eq:ex-enc-request}) shows an
  encrypted request (with $r$ as in
  (\ref{eq:ex-request})). It is encrypted using the public
  key $\pub(k_\text{example.com})$.  The term
  (\ref{eq:ex-enc-response}) is a response (with $s$ as in
  (\ref{eq:ex-response})). It is encrypted symmetrically
  using the (symmetric) key $k'$ that was sent in the
  request (\ref{eq:ex-enc-request}).
\end{example}

\subsection{DNS Messages}\label{sec:dns-messages}
\begin{definition} A \emph{DNS request} is a term of the form
$\an{\cDNSresolve, \mi{domain}, \mi{n}}$ where $\mi{domain}$ $\in
\dns$, $\mi{n} \in \nonces$. We call the set of all DNS requests
$\dnsrequests$.
\end{definition}

\begin{definition} A \emph{DNS response} is a term of the form
$\an{\cDNSresolved, \mi{result}, \mi{n}}$ with $\mi{result} \in
\addresses$, $\mi{n} \in \nonces$. We call the set of all DNS
responses $\dnsresponses$.
\end{definition}

DNS servers are supposed to include the nonce they received
in a DNS request in the DNS response that they send back so
that the party which issued the request can match it with
the request.

%% file: appendix-browsermodel.tex
\section{Detailed Description of the Browser Model}
\label{sec:deta-descr-brows}
Following the informal description of the browser model in
Section~\ref{sec:web-browsers}, we now present a formal
model. We start by introducing some notation and
terminology. 

\subsection{Notation and Terminology (Web Browser State)}

Before we can define the state of a web browser, we first
have to define windows and documents. Concrete window and
document terms are shown in Example~\ref{ex:window}.

\begin{definition} A \emph{window} is a term of the form $w
  = \an{\mi{nonce}, \mi{documents}, \mi{opener}}$ with
  $\mi{nonce} \in \nonces$, $\mi{documents} \subsetPairing
  \documents$ (defined below), $\mi{opener} \in \nonces
  \cup \{\bot\}$ where $\comp{d}{active} = \True$ for
  exactly one $d \inPairing \mi{documents}$ if
  $\mi{documents}$ is not empty (we then call $d$ the
  \emph{active document of $w$}).  We write $\windows$ for
  the set of all windows.  We write
  $\comp{w}{activedocument}$ to denote the active document
  inside window $w$ if it exists and $\an{}$ else.
\end{definition}
We will refer to the window nonce as \emph{(window)
  reference}.

The documents contained in a window term
to the left of the active document are the previously
viewed documents (available to the user via the ``back''
button) and the documents in the window term to the right
of the currently active document are documents available
via the ``forward'' button, as will be clear from the
description of web browser model (see
Section~\ref{sec:descr-web-brows}).

A window $a$ may have opened a top-level window
$b$ (i.e., a window term which is not a subterm of a
document term). In this case, the \emph{opener} part of the
term $b$ is the nonce of $a$, i.e., $\comp{b}{opener} =
\comp{a}{nonce}$.

\begin{definition} A \emph{document} $d$ is a term of the
  form 
  \begin{align*}
  \an{\mi{nonce}, \mi{origin}, \mi{script},
  \mi{scriptstate},\mi{scriptinput}, \mi{subwindows},
  \mi{active}}  
  \end{align*}
 where $\mi{nonce} \in \nonces$,
  $\mi{origin} \in \origins$, $\mi{script} \in \terms$,
  $\mi{scriptstate} \in \terms$, $\mi{scriptinput} \in \terms$,
  $\mi{subwindows} \subsetPairing \windows$, $\mi{active}
  \in \{\True, \bot\}$.  A \emph{limited document} is a
  term of the form $\an{\mi{nonce}, \mi{subwindows}}$ with
  $\mi{nonce}$, $\mi{subwindows}$ as above.  A window $w
  \inPairing \mi{subwindows}$ is called a \emph{subwindow}
  (of $d$).  We write $\documents$ for the set of all
  documents.
\end{definition}
We will refer to the document nonce as \emph{(document)
  reference}.  

\begin{example}\label{ex:window}
  The following is an example of a window term with reference
  $n_1$, two documents, and an opener ($n_4$):
  \begin{align*}
    \an{n_1, &\an{\an{n_2,\! \an{\str{example.com}, \http},
          \str{script1}, \an{}, \an{}, \an{}, \bot},
        \\ &\,\,\an{n_3,\! \an{\str{example.com}, \https},
          \str{script2}, \an{}, \an{}, \an{}, \True}}, n_4}
  \end{align*}
  The first document has the reference $n_2$. It was loaded
  from the origin $\an{\str{example.com}, \http}$, which
  translates into \url{http://example.com}. Its scripting
  process has the string representation $\str{script1}$,
  the last state and the input history of this process are
  empty. The document does not have subwindows and is
  inactive ($\bot$). The second document has the reference
  $n_3$, its origin corresponds to
  \url{https://example.com}, the scripting process is
  represented by $\str{script2}$, and the document is
  active ($\top$).  All other components are empty.
\end{example}

We can now define the set of states of web browsers. Note
that we use the dictionary notation that we introduced in
Definition~\ref{def:dictionaries}.

\begin{definition} Let $\mi{OR} := \left\{\an{o,r}
    \middle|\, o \in \origins,\, r \in \nonces\right\}$. The
  \emph{set of states $Z^p$ of a web browser atomic process}
  $p$ consists of the terms of the form
  \begin{align*} \langle\mi{windows}, \mi{ids},
    \mi{secrets}, \mi{cookies}, \mi{localStorage},
    \mi{sessionStorage}, \mi{keyMapping}, \\\mi{sts},
    \mi{DNSaddress}, \mi{nonces}, \mi{pendingDNS},
    \mi{pendingRequests}, \mi{isCorrupted}\rangle
  \end{align*} where
  \begin{itemize}
  \item $\mi{windows} \subsetPairing \windows$,
  \item $\mi{ids} \subsetPairing \terms$,
  \item $\mi{secrets} \in \dict{\origins}{\nonces}$,
  \item $\mi{cookies}$ is a dictionary over $\dns$ and
    dictionaries of $\cookies$,
  \item $\mi{localStorage} \in \dict{\origins}{\terms}$,
  \item $\mi{sessionStorage} \in \dict{\mi{OR}}{\terms}$,
  \item $\mi{keyMapping} \in \dict{\dns}{\terms}$,
  \item $\mi{sts} \subsetPairing \dns$,
  \item $\mi{DNSaddress} \in \addresses$,
  \item $\mi{nonces} \subsetPairing \nonces$,
  \item $\mi{pendingDNS} \in \dict{\nonces}{\terms}$,
  \item $\mi{pendingRequests} \in$ $\terms$,
  \item and $\mi{isCorrupted} \in \{\bot, \fullcorrupt,$ $
    \closecorrupt\}$.
  \end{itemize} 
\end{definition}

\begin{definition} For two window terms $w$ and $w'$ we
  write $w \windowChildOf w'$ if \\
  \[w \inPairing \comp{\comp{w'}{activedocument}}{subwindows}\text{\ .}\]
We write
  $\windowChildOfX$ for the transitive closure.
\end{definition}

In the following description of the web browser relation
$R^p$ we will use the helper functions
$\mathsf{Subwindows}$, $\mathsf{Docs}$, $\mathsf{Clean}$,
$\mathsf{CookieMerge}$ and $\mathsf{AddCookie}$. 

Given a browser state $s$, $\mathsf{Subwindows}(s)$ denotes
the set of all pointers\footnote{Recall the definition of a
  pointer in Definition~\ref{def:pointer}.} to windows in
the window list $\comp{s}{windows}$, their active
documents, and (recursively) the subwindows of these
documents. We exclude subwindows of inactive documents and
their subwindows. With $\mathsf{Docs}(s)$ we denote the set
of pointers to all active documents in the set of windows
referenced by $\mathsf{Subwindows}(s)$.
\begin{definition} 
  For a browser state $s$ we denote by
  $\mathsf{Subwindows}(s)$ the minimal set of
  pointers that satisfies the
  following conditions: (1) For all windows $w \inPairing
  \comp{s}{windows}$ there is a $\ptr{p} \in
  \mathsf{Subwindows}(s)$ such that $\compn{s}{\ptr{p}} =
  w$. (2) For all $\ptr{p} \in \mathsf{Subwindows}(s)$, the
  active document $d$ of the window $\compn{s}{\ptr{p}}$
  and every subwindow $w$ of $d$ there is a pointer
  $\ptr{p'} \in \mathsf{Subwindows}(s)$ such that
  $\compn{s}{\ptr{p'}} = w$.

  Given a browser state $s$, the set $\mathsf{Docs}(s)$ of
  pointers to active documents is the minimal set such that
  for every $\ptr{p} \in \mathsf{Subwindows}(s)$, there is
  a pointer $\ptr{p'} \in \mathsf{Docs}(s)$ with
  $\compn{s}{\ptr{p'}} =
  \comp{\compn{s}{\ptr{p}}}{activedocument}$.
\end{definition}

The function $\mathsf{Clean}$ will be used to determine
which information about windows and documents the script
running in the document $d$ has access to.
\begin{definition} Let $s$ be a browser state and $d$ a
  document.  By $\mathsf{Clean}(s, d)$ we denote the term
  that equals $\comp{s}{windows}$ but with all inactive
  documents removed (including their subwindows etc.) and
  all subterms that represent non-same-origin documents
  w.r.t.~$d$ replaced by a limited document $d'$ with the
  same nonce and the same subwindow list. Note that
  non-same-origin documents on all levels are replaced by
  their corresponding limited document.
\end{definition}

The function $\mathsf{CookieMerge}$ merges two sequences of
cookies together: When used in the browser,
$\mi{oldcookies}$ is the sequence of existing cookies for
some origin, $\mi{newcookies}$ is a sequence of new cookies
that was output by some script. The sequences are merged
into a set of cookies using an algorithm that is based on
the \emph{Storage Mechanism} algorithm described in
RFC6265.
\begin{definition} \label{def:cookiemerge} For a sequence
  of cookies (with pairwise different names)
  $\mi{oldcookies}$ and a sequence of cookies
  $\mi{newcookies}$, the set
  $\mathsf{CookieMerge}(\mi{oldcookies}, \mi{newcookies})$
  is defined by the following algorithm: From
  $\mi{newcookies}$ remove all cookies $c$ that have
  $c.\str{content}.\str{httpOnly} \equiv \True$. For any
  $c$, $c' \inPairing \mi{newcookies}$, $\comp{c}{name}
  \equiv \comp{c'}{name}$, remove the cookie that appears
  left of the other in $\mi{newcookies}$. Let $m$ be the
  set of cookies that have a name that either appears in
  $\mi{oldcookies}$ or in $\mi{newcookies}$, but not in
  both. For all pairs of cookies $(c_\text{old},
  c_\text{new})$ with $c_\text{old} \inPairing
  \mi{oldcookies}$, $c_\text{new} \inPairing
  \mi{newcookies}$, $\comp{c_\text{old}}{name} \equiv
  \comp{c_\text{new}}{name}$, add $c_\text{new}$ to $m$ if
  $\comp{\comp{c_\text{old}}{content}}{httpOnly} \equiv
  \bot$ and add $c_\text{old}$ to $m$ otherwise. The result
  of $\mathsf{CookieMerge}(\mi{oldcookies},
  \mi{newcookies})$ is $m$.
\end{definition}

The function $\mathsf{AddCookie}$ adds a cookie $c$
received in an HTTP response to the sequence of cookies
contained in the sequence $\mi{oldcookies}$. It is again
based on the algorithm described in RFC6265 but simplified
for the use in the browser model.
\begin{definition} \label{def:addcookie} For a sequence of cookies (with different
  names) $\mi{oldcookies}$ and a cookie $c$, the sequence
  $\mathsf{AddCookie}(\mi{oldcookies}, c)$ is defined by the
  following algorithm: Let $m := \mi{oldcookies}$. Remove
  any $c'$ from $m$ that has $\comp{c}{name} \equiv
  \comp{c'}{name}$. Append $c$ to $m$ and return $m$.
\end{definition}

The function $\mathsf{NavigableWindows}$ returns a set of
windows that a document is allowed to navigate. We closely
follow \cite{html5}, Section~5.1.4 for this definition.
\begin{definition} The set $\mathsf{NavigableWindows}(\ptr{w}, s')$
  is the set $\ptr{W} \subseteq
  \mathsf{Subwindows}(s')$ of pointers to windows that the
  active document in $\ptr{w}$ is allowed to navigate. The
  set $\ptr{W}$ is defined to be the minimal set such that
  for every $\ptr{w'}
  \in \mathsf{Subwindows}(s')$ the following is true: 
\begin{itemize}
\item If
  $\comp{\comp{\compn{s'}{\ptr{w}'}}{activedocument}}{origin}
  \equiv
  \comp{\comp{\compn{s'}{\ptr{w}}}{activedocument}}{origin}$
  (i.e., the active documents in $\ptr{w}$ and $\ptr{w'}$ are
  same-origin), then $\ptr{w'} \in \ptr{W}$, and
\item If ${\compn{s'}{\ptr{w}} \childof
    \compn{s'}{\ptr{w'}}}$ $\wedge$ $\nexists\, \ptr{w}''
  \in \mathsf{Subwindows}(s')$ with $\compn{s'}{\ptr{w}'}
  \childof \compn{s'}{\ptr{w}''}$ ($\ptr{w'}$ is a
  top-level window and $\ptr{w}$ is an ancestor window of
  $\ptr{w'}$), then $\ptr{w'} \in \ptr{W}$, and
\item If $\exists\, \ptr{p} \in \mathsf{Subwindows}(s')$
  such that $\compn{s'}{\ptr{w}'} \windowChildOfX
  \compn{s'}{\ptr{p}}$ \\$\wedge$
  $\comp{\comp{\compn{s'}{\ptr{p}}}{activedocument}}{origin}
  =
  \comp{\comp{\compn{s'}{\ptr{w}}}{activedocument}}{origin}$
  ($\ptr{w'}$ is not a top-level window but there is an
  ancestor window $\ptr{p}$ of $\ptr{w'}$ with an active
  document that has the same origin as the active document
  in $\ptr{w}$), then $\ptr{w'} \in \ptr{W}$, and
\item If $\exists\, \ptr{p} \in \mathsf{Subwindows}(s')$ such
  that $\comp{\compn{s'}{\ptr{w'}}}{opener} =
  \comp{\compn{s'}{\ptr{p}}}{nonce}$ $\wedge$ $\ptr{p} \in
  \ptr{W}$ ($\ptr{w'}$ is a top-level window---it has an
  opener---and $\ptr{w}$ is allowed to navigate the opener
  window of $\ptr{w'}$, $\ptr{p}$), then $\ptr{w'} \in
  \ptr{W}$. 
\end{itemize}
\end{definition}

\subsection{Description of the Web Browser Atomic
  Process}\label{sec:descr-web-brows}
We will now describe the relation $R^p$ of a standard HTTP
browser $p$.  For a tuple $r =
\left(\left(\left(a{:}f{:}m\right), s\right), \left(M,
    s'\right)\right)$ we define $r$ to belong to $R^p$
if\/f the non-deterministic algorithm presented in
Section~\ref{app:mainalgorithmwebbrowserprocess}, when
given $\left(\left(a{:}f{:}m\right), s\right)$ as input,
terminates with \textbf{stop}~$M$,~$s'$, i.e., with output
$M$ and $s'$. Recall that $\left(a{:}f{:}m\right)$ is an
(input) event and $s$ is a (browser) state, $M$ is a set of
(output) events, and $s'$ is a new (browser) state.

The notation $\textbf{let}\ n \leftarrow N$ is used to
describe that $n$ is chosen non-de\-ter\-mi\-nis\-tic\-ally from the
set $N$.  We write $\textbf{for each}\ s \in M\ \textbf{do}$
to denote that the following commands (until \textbf{end
  for}) are repeated for every element in $M$, where the
variable $s$ is the current element. The order in which the
elements are processed is chosen non-deterministically.

We first define some functions which will be used in the
main algorithm presented in
Section~\ref{app:mainalgorithmwebbrowserprocess}.

\subsubsection{Functions.} \label{app:proceduresbrowser} In
the description of the following functions we use $a$,
$f$, $m$, $s$ and $N^p$ as read-only global input
variables. Also, the functions use the set $N^p$ as a
read-only set. All other variables are local variables or
arguments.

$\mathsf{TAKENONCE}$ returns a nonce from the set of unused
nonces and modifies the browser state such that the nonce
is added to the sequence of used nonces. Note that this
function returns two values, the nonce $n$ and the modified
state $s'$.
\captionof{algorithm}{\label{alg:takenonce}
  Non-deterministically choose a fresh nonce.}
\begin{algorithmic}[1]
  \Function{$\mathsf{TAKENONCE}$}{$s'$}
    \LetND{$n$}{$\left\{x \middle| x \in N^p \wedge x \not\inPairing \comp{s'}{nonces} \right\}$} %
    \Append{$n$}{$\comp{s'}{nonces}$}
    \State \Return $n, s'$
  \EndFunction
\end{algorithmic} \setlength{\parindent}{1em}

The following function, $\mathsf{GETNAVIGABLEWINDOW}$, is
called by the browser to determine the window that is
\emph{actually} navigated when a script in the window
$s'.\ptr{w}$ provides a window reference for navigation
(e.g., for opening a link). When it is given a window
reference (nonce) $\mi{window}$,
$\mathsf{GETNAVIGABLEWINDOW}$ returns a pointer to a
selected window term in $s'$:
\begin{itemize}
\item If $\mi{window}$ is the string $\wBlank$, a new
  window is created and a pointer to that window is
  returned.
\item If $\mi{window}$ is a nonce (reference) and there is
  a window term with a reference of that value in the
  windows in $s'$, a pointer $\ptr{w'}$ to that window term
  is returned, as long as the window is navigable by the
  current window's document (as defined by
  $\mathsf{NavigableWindows}$ above).
\end{itemize}
In all other cases, $\ptr{w}$ is returned instead (the
script navigates its own window).
\captionof{algorithm}{\label{alg:getnavigablewindow}
  Determine window for navigation.}
\begin{algorithmic}[1]
  \Function{$\mathsf{GETNAVIGABLEWINDOW}$}{$\ptr{w}$, $\mi{window}$, $s'$}
    \If{$\mi{window} \equiv \wBlank$} \Comment{Open a new window when $\wBlank$ is used}
      \Let {$n$, $s'$}{\textsf{TAKENONCE}$(s')$}
      \Let{$w'$}{$\an{n, \an{}, \comp{\compn{s'}{\ptr{w}}}{nonce} }$}
      \Append{$w'$}{$\comp{s'}{windows}$}  \breakalgohook{2}\textbf{and} let
      $\ptr{w}'$ be a pointer to this new element in $s'$
      \State \Return{$(\ptr{w}', s')$}
    \EndIf
    \LetNDST{$\ptr{w}'$}{$\mathsf{NavigableWindows}(\ptr{w},
      s')$}{$\comp{\compn{s'}{\ptr{w}'}}{nonce} \equiv
      \mi{window}$\breakalgohook{1}}{\textbf{return} $(\ptr{w}, s')$}
    \State \Return{$(\ptr{w'}, s')$}
  \EndFunction
\end{algorithmic} \setlength{\parindent}{1em}

The following function takes a window reference as input
and returns a pointer to a window as above, but it checks
only that the active documents in both windows are
same-origin. It creates no new windows.
\captionof{algorithm}{\label{alg:getwindow} Determine same-origin window.}
\begin{algorithmic}[1]
  \Function{$\mathsf{GETWINDOW}$}{$\ptr{w}$, $\mi{window}$, $s'$}
    \LetNDST{$\ptr{w}'$}{$\mathsf{Subwindows}(s')$}{$\comp{\compn{s'}{\ptr{w}'}}{nonce} \equiv \mi{window}$\breakalgohook{1}}{\textbf{return} $(\ptr{w}, s')$}
    \If{
      $\comp{\comp{\compn{s'}{\ptr{w}'}}{activedocument}}{origin}
      \equiv
      \comp{\comp{\compn{s'}{\ptr{w}}}{activedocument}}{origin}$
    }
      \State \Return{$(\ptr{w}', s')$}
    \EndIf
    \State \Return{$(\ptr{w}, s')$}
  \EndFunction
\end{algorithmic} \setlength{\parindent}{1em}

The next function is used to stop any pending
requests for a specific window. From the pending requests
and pending DNS requests it removes any requests with the
given window reference $n$.
\captionof{algorithm}{\label{alg:cancelnav} Cancel pending requests for given window.}
\begin{algorithmic}[1]
  \Function{$\mathsf{CANCELNAV}$}{$n$, $s'$}
    \State \textbf{remove all} $\an{n, \mi{req}, \mi{key}, \mi{f}}$ \textbf{ from } $\comp{s'}{pendingRequests}$ \textbf{for any} $\mi{req}$, $\mi{key}$, $\mi{f}$
    \State \textbf{remove all} $\an{x, \an{n, \mi{message}, \mi{protocol}}}$ \textbf{ from } $\comp{s'}{pendingDNS}$\breakalgohook{1} \textbf{for any} $\mi{x}$, $\mi{message}$, $\mi{protocol}$
    \State \Return{$s'$}
  \EndFunction
\end{algorithmic} \setlength{\parindent}{1em}

The following function takes an HTTP request
$\mi{message}$ as input, adds cookie and origin headers to
the message, creates a DNS request for the hostname given
in the request and stores the request in
$\comp{s'}{pendingDNS}$ until the DNS resolution
finishes. For normal HTTP requests, $\mi{reference}$ is a
window reference. For \xhrs, $\mi{reference}$ is a value of
the form $\an{\mi{document}, \mi{nonce}}$ where
$\mi{document}$ is a document reference and $\mi{nonce}$ is
some nonce that was chosen by the script that initiated the
request. $\mi{protocol}$ is either $\http$ or
$\https$. $\mi{origin}$ is the origin header value that is
to be added to the HTTP request.

\captionof{algorithm}{\label{alg:send} Prepare headers, do DNS resolution, save message. }
\begin{algorithmic}[1]
  \Function{$\mathsf{SEND}$}{$\mi{reference}$, $\mi{message}$, $\mi{protocol}$, $\mi{origin}$, $s'$}
    \If{$\comp{\mi{message}}{host} \inPairing \comp{s'}{sts}$}
      \Let{$\mi{protocol}$}{$\https$}
    \EndIf
    \Let{ $\mi{cookies}$}{$\langle\{\an{\comp{c}{name}, \comp{\comp{c}{content}}{value}} | c\inPairing \comp{s'}{cookies}\left[\comp{\mi{message}}{host}\right]$} \label{line:assemble-cookies-for-request} \breakalgohook{1} $\wedge \left(\comp{\comp{c}{content}}{secure} \implies \left(\mi{protocol} = \https\right)\right) \}\rangle$ \label{line:cookie-rules-http}
    \Let{$\comp{\mi{message}}{headers}[\str{Cookie}]$}{$\mi{cookies}$}
    \If{$\mi{origin} \not\equiv \bot$}
      \Let{$\comp{\mi{message}}{headers}[\str{Origin}]$}{$\mi{origin}$}
    \EndIf
    \Let{$n, s'$}{\textsf{TAKENONCE}$(s')$} 
    \Let{$\comp{s'}{pendingDNS}[n]$}{$\an{\mi{reference},
        \mi{message}, \mi{protocol}}$} \label{line:add-to-pendingdns}
    \State \textbf{stop} $\{(\comp{s'}{DNSaddress} : a :
    \an{\cDNSresolve, \mi{host}, n})\}$, $s'$
  \EndFunction
\end{algorithmic} \setlength{\parindent}{1em}
\noindent

The following two functions have informally been described
in Section~\ref{sec:browserrelation}.

The function $\mathsf{RUNSCRIPT}$ performs a script
execution step of the script in the document
$\compn{s'}{\ptr{d}}$ (which is part of the window
$\compn{s'}{\ptr{w}}$). A new script and document state is
chosen according to the relation defined by the script and
the new script and document state is saved. Afterwards, the
$\mi{command}$ that the script issued is interpreted.  Note
that \textbf{for each} (Line~\ref{line:for-each}) works in
a non-deterministic order.

\captionof{algorithm}{\label{alg:runscript} Execute a script.}
\begin{algorithmic}[1]
  \Function{$\mathsf{RUNSCRIPT}$}{$\ptr{w}$, $\ptr{d}$, $s'$}
    \Let {$n$, $s'$}{\textsf{TAKENONCE}$(s')$}
    \Let{$\mi{tree}$}{$\mathsf{Clean}(s', \compn{s'}{\ptr{d}})$} \label{line:clean-tree}

    \Let{$\mi{cookies}$}{$\langle\{\an{\comp{c}{name}, \comp{\comp{c}{content}}{value}} | c \inPairing \comp{s'}{cookies}\left[  \comp{\comp{\compn{s'}{\ptr{d}}}{origin}}{host}  \right]$
     \breakalgohook{1} $\wedge\,\comp{\comp{c}{content}}{httpOnly} = \bot$ \breakalgohook{1} $\wedge\,\left(\comp{\comp{c}{content}}{secure} \implies \left(\comp{\comp{\compn{s'}{\ptr{d}}}{origin}}{protocol} \equiv \https\right)\right) \}\rangle$} \label{line:assemble-cookies-for-script}
    \LetND{$\mi{tlw}$}{$\comp{s'}{windows}$ \textbf{such that} $\mi{tlw}$ is the top-level window containing $\ptr{d}$} 
    \Let{$\mi{sessionStorage}$}{$\comp{s'}{sessionStorage}\left[\an{\comp{\compn{s'}{\ptr{d}}}{origin}, \comp{\mi{tlw}}{nonce}}\right]$}
    \Let{$\mi{localStorage}$}{$\comp{s'}{localStorage}\left[\comp{\compn{s'}{\ptr{d}}}{origin}\right]$}
    \Let{$\mi{secret}$}{$\comp{s'}{secrets}\left[\comp{\compn{s'}{\ptr{d}}}{origin}\right]$} \label{line:browser-secrets}
    \State \textbf{let} $\mi{nonces}$ be an infinite subset of $\left\{x \middle| x \in N^p \wedge x \not\inPairing \comp{s'}{nonces} \right\}$
    \LetND{$R$}{$\mathsf{script}^{-1}(\comp{\compn{s'}{\ptr{d}}}{script})$}
    \Let{$\mi{in}$}{$\langle\mi{tree}$, $\comp{\compn{s'}{\ptr{d}}}{nonce}, \comp{\compn{s'}{\ptr{d}}}{scriptstate}$, $\comp{\compn{s'}{\ptr{d}}}{scriptinput}$, $\mi{cookies},$ \breakalgohook{1}  $\mi{localStorage}$, $\mi{sessionStorage}$, $\comp{s'}{ids}$, $\mi{secret}\rangle$}
    \LetND{$\mi{state}'$}{$\terms$, \breakalgohook{1}
      $\mi{cookies}' \gets \mathsf{Cookies}$, \breakalgohook{1}
      $\mi{localStorage}' \gets \terms$,\breakalgohook{1} $\mi{command}
      \gets \terms$, \breakalgohook{1} $\mi{out} := \an{\mi{state}', \mi{cookies}', \mi{localStorage}',$ $\mi{sessionStorage}', \mi{command}}$
      \breakalgohook{1} \textbf{such that} $((\mi{in}, \mi{nonces}), \mi{out}) \in R$}  \label{line:trigger-script}.
    \For{\textbf{each} $n \in d_{\nonces}(\an{\mi{in}, \mi{out}}) \cap N^p$} \label{line:for-each}
      \Append{$n$}{$\comp{s'}{nonces}$}
    \EndFor
    \Let{$\comp{s'}{cookies}\left[\comp{\comp{\compn{s'}{\ptr{d}}}{origin}}{host}\right]$\breakalgohook{1}}{$\langle\mathsf{CookieMerge}(\comp{s'}{cookies}\left[\comp{\comp{\compn{s'}{\ptr{d}}}{origin}}{host}\right]$, $\mi{cookies}')\rangle$} \label{line:cookiemerge}
    \Let{$\comp{s'}{localStorage}\left[\comp{\compn{s'}{\ptr{d}}}{origin}\right]$}{$\mi{localStorage}'$}
    \Let{$\comp{s'}{sessionStorage}\left[\an{\comp{\compn{s'}{\ptr{d}}}{origin}, \comp{\mi{tlw}}{nonce}}\right]$}{$\mi{sessionStorage}'$}
    \Let{$\comp{\compn{s'}{\ptr{d}}}{scriptstate}$}{$state'$}
    \Switch{$\mi{command}$}
      \Case{$\an{}$}
        \State \textbf{stop} $\{\}$, $s'$
      \EndCase
      \Case{$\an{\tHref, \mi{url},
          \mi{hrefwindow}}$}\footnote{See the definition of
        URLs in Appendix~\ref{sec:urls}.}
      \Let{$\ptr{w}'$,
        $s'$}{$\mathsf{GETNAVIGABLEWINDOW}$($\ptr{w}$,
        $\mi{hrefwindow}$, $s'$)} \Let{$\mi{req}$}{$\hreq{ nonce=n, 
          method=\mGet, host=\comp{\mi{url}}{host},
          path=\comp{\mi{url}}{path},
          headers=\an{},
          parameters=\comp{\mi{url}}{params}, body=\an{}
        }$}
      \Let{$s'$}{$\mathsf{CANCELNAV}(\comp{\compn{s'}{\ptr{w}'}}{nonce}, s')$}
      \State \textsf{SEND}($\comp{\compn{s'}{\ptr{w}'}}{nonce}$, $\mi{req}$, $\comp{\mi{url}}{protocol}$, $\bot$, $s'$) \label{line:send-href}
      \EndCase
      \Case{$\an{\tIframe, \mi{url}, \mi{window}}$}
        \Let{$\ptr{w}'$, $s'$}{$\mathsf{GETWINDOW}(\ptr{w}, \mi{window}, s')$}
        \Let{$\mi{req}$}{$\hreq{
            nonce=n,
            method=\mGet,
            host=\comp{\mi{url}}{host},
            path=\comp{\mi{url}}{path},
            headers=\an{},
            parameters=\comp{\mi{url}}{params},
            body=\an{}
          }$}
        \Let {$n$, $s'$}{\textsf{TAKENONCE}$(s')$}
        \Let{$w'$}{$\an{n, \an{}, \bot}$}
        \Let{$\comp{\comp{\compn{s'}{\ptr{w}'}}{activedocument}}{subwindows}$\breakalgohook{3}}{ $\comp{\comp{\compn{s'}{\ptr{w}'}}{activedocument}}{subwindows} \plusPairing w'$}
        \State \textsf{SEND}(n, $\mi{req}$, $\comp{\mi{url}}{protocol}$, $\bot$, $s'$) \label{line:send-iframe}
      \EndCase
      \Case{$\an{\tForm, \mi{url}, \mi{method}, \mi{data}, \mi{hrefwindow}}$}
        \If{$\mi{method} \not\in \{\mGet, \mPost\}$} \footnote{The working draft for HTML5 allowed for DELETE and PUT methods in HTML5 forms. However, these have since been removed. See \url{http://www.w3.org/TR/2010/WD-html5-diff-20101019/\#changes-2010-06-24}.}
          \State \textbf{stop} $\{\}$, $s'$
        \EndIf
        \Let{$\ptr{w}'$, $s'$}{$\mathsf{GETNAVIGABLEWINDOW}$($\ptr{w}$, $\mi{hrefwindow}$, $s'$)}
        \If{$\mi{method} = \mGet$}
          \Let{$\mi{body}$}{$\an{}$}
          \Let{$\mi{params}$}{$\mi{data}$}
          \Let{$\mi{origin}$}{$\bot$}
        \Else
          \Let{$\mi{body}$}{$\mi{data}$}
          \Let{$\mi{params}$}{$\comp{\mi{url}}{params}$}
          \Let{$\mi{origin}$}{$\comp{\compn{s'}{\ptr{d}}}{origin}$}
        \EndIf
        \Let{$\mi{req}$}{$\hreq{
            nonce=n,
            method=\mi{method},
            host=\comp{\mi{url}}{host},
            path=\comp{\mi{url}}{path},
            headers=\an{},
            parameters=\mi{params},
            xbody=\mi{body}
          }$}
        \Let{$s'$}{$\mathsf{CANCELNAV}(\comp{\compn{s'}{\ptr{w}'}}{nonce}, s')$}
        \State \textsf{SEND}($\comp{\compn{s'}{\ptr{w}'}}{nonce}$, $\mi{req}$, $\comp{\mi{url}}{protocol}$, $\mi{origin}$, $s'$) \label{line:send-form}
      \EndCase
      \Case{$\an{\tSetScript, \mi{window}, \mi{script}}$}
        \Let{$\ptr{w}'$, $s'$}{$\mathsf{GETWINDOW}(\ptr{w}, \mi{window}, s')$}
        \Let{$\comp{\comp{\compn{s'}{\ptr{w}'}}{activedocument}}{script}$}{$\mi{script}$}
        \State \textbf{stop} $\{\}$, $s'$
      \EndCase
      \Case{$\an{\tSetScriptState, \mi{window}, \mi{scriptstate}}$}
        \Let{$\ptr{w}'$, $s'$}{$\mathsf{GETWINDOW}(\ptr{w}, \mi{window}, s')$}
        \Let{$\comp{\comp{\compn{s'}{\ptr{w}'}}{activedocument}}{scriptstate}$}{$\mi{scriptstate}$}
        \State \textbf{stop} $\{\}$, $s'$
      \EndCase
      \Case{$\an{\tXMLHTTPRequest, \mi{url}, \mi{method}, \mi{data}, \mi{xhrreference}}$}
        \If{$\mi{method} \in \{\mConnect, \mTrace, \mTrack\}$} 
          \State \textbf{stop} $\{\}$, $s'$
        \EndIf
        \If{$\comp{\mi{url}}{host} \not\equiv \comp{\comp{\compn{s'}{\ptr{d}}}{origin}}{host}$ \breakalgohook{3} $\vee$ $\comp{\mi{url}}{protocol} \not\equiv \comp{\comp{\compn{s'}{\ptr{d}}}{origin}}{protocol}$} 
          \State \textbf{stop} $\{\}$, $s'$
        \EndIf
        \If{$\mi{method} \in \{\mGet, \mHead\}$}
          \Let{$\mi{data}$}{$\an{}$}
          \Let{$\mi{origin}$}{$\bot$}
        \Else
          \Let{$\mi{origin}$}{$\comp{\compn{s'}{\ptr{d}}}{origin}$}
        \EndIf
        \Let{$\mi{req}$}{$\hreq{
            nonce=n,
            method=\mi{method},
            host=\comp{\mi{url}}{host},
            path=\comp{\mi{url}}{path},
            headers={},
            parameters=\comp{\mi{url}}{params},
            xbody=\mi{data}
          }$}
        \State \textsf{SEND}($\an{\comp{\compn{s'}{\ptr{d}}}{nonce}, \mi{xhrreference}}$, $\mi{req}$, $\comp{\mi{url}}{protocol}$, $\mi{origin}$, $s'$)\label{line:send-xhr}
      \EndCase
      \Case{$\an{\tBack, \mi{window}}$} \footnote{Note that navigating a window using the back/forward buttons does not trigger a reload of the affected documents. While real world browser may chose to refresh a document in this case, we assume that the complete state of a previously viewed document is restored.}
        \Let{$\ptr{w}'$, $s'$}{$\mathsf{GETNAVIGABLEWINDOW}$($\ptr{w}$, $\mi{window}$, $s'$)}
        \If{$\exists\, \ptr{j} \in \mathbb{N}, \ptr{j} >
          1$ \textbf{such that} $\comp{\compn{\comp{\compn{s'}{\ptr{w'}}}{documents}}{\ptr{j}}}{active} \equiv \True$}
          \Let{$\comp{\compn{\comp{\compn{s'}{\ptr{w'}}}{documents}}{\ptr{j}}}{active}$}{$\bot$}
          \Let{$\comp{\compn{\comp{\compn{s'}{\ptr{w'}}}{documents}}{(\ptr{j}-1)}}{active}$}{$\True$}
          \Let{$s'$}{$\mathsf{CANCELNAV}(\comp{\compn{s'}{\ptr{w}'}}{nonce}, s')$}
        \EndIf
        \State \textbf{stop} $\{\}$, $s'$
      \EndCase
      \Case{$\an{\tForward, \mi{window}}$}
        \Let{$\ptr{w}'$, $s'$}{$\mathsf{GETNAVIGABLEWINDOW}$($\ptr{w}$, $\mi{window}$, $s'$)}
        \If{$\exists\, \ptr{j} \in \mathbb{N} $ \textbf{such that} $\comp{\compn{\comp{\compn{s'}{\ptr{w'}}}{documents}}{\ptr{j}}}{active} \equiv \True$ \breakalgohook{3} $\wedge$  $\compn{\comp{\compn{s'}{\ptr{w'}}}{documents}}{(\ptr{j}+1)} \in \mathsf{Documents}$}
          \Let{$\comp{\compn{\comp{\compn{s'}{\ptr{w'}}}{documents}}{\ptr{j}}}{active}$}{$\bot$}
          \Let{$\comp{\compn{\comp{\compn{s'}{\ptr{w'}}}{documents}}{(\ptr{j}+1)}}{active}$}{$\True$}
          \Let{$s'$}{$\mathsf{CANCELNAV}(\comp{\compn{s'}{\ptr{w}'}}{nonce}, s')$}
        \EndIf
        \State \textbf{stop} $\{\}$, $s'$
      \EndCase
      \Case{$\an{\tClose, \mi{window}}$}
        \Let{$\ptr{w}'$, $s'$}{$\mathsf{GETNAVIGABLEWINDOW}$($\ptr{w}$, $\mi{window}$, $s'$)}
        \State \textbf{remove} $\compn{s'}{\ptr{w'}}$ from the sequence containing it 
        \State \textbf{stop} $\{\}$, $s'$
      \EndCase

      \Case{$\an{\tPostMessage, \mi{window}, \mi{message}, \mi{origin}}$}
        \LetND{$\ptr{w}'$}{$\mathsf{Subwindows}(s')$ \textbf{such that} $\comp{\compn{s'}{\ptr{w}'}}{nonce} \equiv \mi{window}$}
        \If{$\exists \ptr{j} \in \mathbb{N}$ \textbf{such that} $\comp{\compn{\comp{\compn{s'}{\ptr{w'}}}{documents}}{\ptr{j}}}{active} \equiv \True$ \breakalgohook{3} $\wedge  (\mi{origin} \not\equiv \bot \implies \comp{\compn{\comp{\compn{s'}{\ptr{w'}}}{documents}}{\ptr{j}}}{origin} \equiv \mi{origin})$}    \label{line:append-pm-to-scriptinput-condition}
        \Let{$\comp{\compn{\comp{\compn{s'}{\ptr{w'}}}{documents}}{\ptr{j}}}{scriptinput}$\breakalgohook{4}}{ $\comp{\compn{\comp{\compn{s'}{\ptr{w'}}}{documents}}{\ptr{j}}}{scriptinput}$ \breakalgohook{4} $\plusPairing$
         $\an{\tPostMessage, \comp{\compn{s'}{\ptr{w}}}{nonce}, \comp{\compn{s'}{\ptr{d}}}{origin}, \mi{message}}$} \label{line:append-pm-to-scriptinput}
        \EndIf
      \EndCase
    \EndSwitch
  \EndFunction
\end{algorithmic} \setlength{\parindent}{1em}

The function $\mathsf{PROCESSRESPONSE}$ is responsible for
processing an HTTP response ($\mi{response}$) that was
received as the response to a request ($\mi{request}$) that
was sent earlier. In $\mi{reference}$, either a window or a
document reference is given (see explanation for
Algorithm~\ref{alg:send} above). Again, $\mi{protocol}$ is
either $\http$ or $\https$.

The function first saves any cookies that were contained in
the response to the browser state, then checks whether a
redirection is requested (Location header). If that is not
the case, the function creates a new document (for normal
requests) or delivers the contents of the response to the
respective receiver (for \xhr responses).
\captionof{algorithm}{\label{alg:processresponse} Process an HTTP response.}
\begin{algorithmic}[1]
\Function{$\mathsf{PROCESSRESPONSE}$}{$\mi{response}$, $\mi{reference}$, $\mi{request}$, $\mi{protocol}$, $s'$}
  \Let{$n, s'$}{\textsf{TAKENONCE}$(s')$} 
  \If{$\mathtt{Set{\mhyphen}Cookie} \in
    \comp{\mi{response}}{headers}$}
    \For{\textbf{each} $c \inPairing \comp{\mi{response}}{headers}\left[\mathtt{Set{\mhyphen}Cookie}\right]$, $c \in \mathsf{Cookies}$}
      \Let{$\comp{s'}{cookies}\left[\comp{\comp{\mi{request}}{url}}{host}\right]$\breakalgohook{3}}{$\mathsf{AddCookie}(\comp{s'}{cookies}\left[\comp{\comp{\mi{request}}{url}}{host}\right], c)$} \label{line:set-cookie}
    \EndFor
  \EndIf  
  \If{$\mathtt{Strict{\mhyphen}Transport{\mhyphen}Security} \in \comp{\mi{response}}{headers}$ $\wedge$ $\mi{protocol} \equiv \https$}
    \Append{$\comp{\mi{request}}{host}$}{$\comp{s'}{sts}$}
  \EndIf
  \If{$\mathtt{Location} \in \comp{\mi{response}}{headers} \wedge \comp{\mi{response}}{status} \in \{303, 307\}$} \label{line:location-header} \footnote{The RFC for HTTPbis (currently in draft status), which obsoletes RFC 2616, does not specify whether a POST/DELETE/etc. request that was answered with a status code of 301 or 302 should be rewritten to a GET request or not (``for historic reasons'' that are detailed in Section~7.4.). 
As the specification is clear for the status codes 303 and 307 (and most browsers actually follow the specification in this regard), we focus on modeling these.}
    \Let{$\mi{url}$}{$\comp{\mi{response}}{headers}\left[\mathtt{Location}\right]$}
    \Let{$\mi{method}'$}{$\comp{\mi{request}}{method}$} \footnote{While the standard demands that users confirm redirections of non-safe-methods (e.g., POST), we assume that users generally confirm these redirections.}
    \Let{$\mi{body}'$}{$\comp{\mi{request}}{body}$} \footnote{If, for example, a GET request is redirected and the original request contained a body, this body is preserved, as HTTP allows for payloads in messages with all HTTP methods, except for the TRACE method (a detail which we omit). 
Browsers will usually not send body payloads for methods that do not specify semantics for such data in the first place.}
    \If{$\str{Origin} \in \comp{request}{headers}$}
      \Let{$\mi{origin}$}{$\an{\comp{request}{headers}[\str{Origin}], \an{\comp{request}{host}, \mi{protocol}}}$}
    \Else
      \Let{$\mi{origin}$}{$\bot$}
    \EndIf
    \If{$\comp{\mi{response}}{status} \equiv 303 \wedge \comp{\mi{request}}{method} \not\in \{\mGet, \mHead\}$}
      \Let {$\mi{method}'$}{$\mGet$}
      \Let{$\mi{body}'$}{$\an{}$}
    \EndIf
    \If{$\nexists\, \ptr{w} \in \mathsf{Subwindows}(s')$ \textbf{such that} $\comp{\compn{s'}{\ptr{w}}}{nonce} \equiv \mi{reference}$} \breakalgo{3}\Comment{Do not redirect XHRs.}
      \State \textbf{stop} $\{\}$, $s$
    \EndIf
    \Let{$\mi{req}$}{$\hreq{
            nonce=n,
            method=\mi{method'},
            host=\comp{\mi{url}}{host},
            path=\comp{\mi{url}}{path},
            headers=\an{},
            parameters=\comp{\mi{url}}{params},
            xbody=\mi{body}'
          }$}
    \State \textsf{SEND}($\mi{reference}$, $\mi{req}$, $\comp{\mi{url}}{protocol}$, $\mi{origin}$, $s'$)\label{line:send-redirect}
  \EndIf

  \If{$\exists\, \ptr{w} \in \mathsf{Subwindows}(s')$ \textbf{such that} $\comp{\compn{s'}{\ptr{w}}}{nonce} \equiv \mi{reference}$} \breakalgo{3}\Comment{normal response}
    \Let{$\mi{script}$}{$\proj{1}{\comp{\mi{response}}{body}}$}
    \Let{$\mi{scriptstate}$}{$\proj{2}{\comp{\mi{response}}{body}}$}
    \Let{$d$}{$\an{n, \an{\comp{\mi{request}}{host}, \comp{\mi{request}}{protocol}}, \mi{script}, \mi{scriptstate}, \an{}, \an{}, \True}$} \label{line:take-script} \label{line:set-origin-of-document}
    \If{$\comp{\compn{s'}{\ptr{w}}}{documents} \equiv \an{}$}
      \Let{$\comp{\compn{s'}{\ptr{w}}}{documents}$}{$\an{d}$}
    \Else
      \LetND{$\ptr{i}$}{$\mathbb{N}$ \textbf{such that} $\comp{\compn{\comp{\compn{s'}{\ptr{w}}}{documents}}{\ptr{i}}}{active} \equiv \True$}
      \Let{$\comp{\compn{\comp{\compn{s'}{\ptr{w}}}{documents}}{\ptr{i}}}{active}$}{$\bot$}
      \State \textbf{remove} $\compn{\comp{\compn{s'}{\ptr{w}}}{documents}}{(\ptr{i}+1)}$ and all following documents \breakalgohook{3} from $\comp{\compn{s'}{\ptr{w}}}{documents}$
      \Append{$d$}{$\comp{\compn{s'}{\ptr{w}}}{documents}$}
    \EndIf
    \State \textbf{stop} $\{\}$, $s'$
  \ElsIf{$\exists\, \ptr{w} \in \mathsf{Subwindows}(s')$, $\ptr{d}$ \textbf{such that} $\comp{\compn{s'}{\ptr{d}}}{nonce} \equiv \proj{1}{\mi{reference}} $ \breakalgohook{1}  $\wedge$  $\compn{s'}{\ptr{d}} = \comp{\compn{s'}{\ptr{w}}}{activedocument}$} \label{line:process-xhr-response} \Comment{process XHR response}
    \Append{\breakalgo{3}$\an{\http, \comp{\mi{response}}{body}, \proj{2}{\mi{reference}}}$}{$\comp{\compn{s'}{\ptr{d}}}{scriptinput}$}
  \EndIf
\EndFunction
\end{algorithmic} \setlength{\parindent}{1em}

\subsubsection{Main Algorithm.}\label{app:mainalgorithmwebbrowserprocess}
This is the main algorithm of the browser relation. It was
already presented informally in
Section~\ref{sec:web-browsers} and follows the structure
presented there. It receives the message $m$ as input, as
well as $a$, $f$ and $s$ as above.

\captionof{algorithm}{\label{alg:browsermain} Main Algorithm}
\begin{algorithmic}[1]
\Statex[-1] \textbf{Input:} $(a{:}f{:}m),s$
  \Let{$s'$}{$s$}

  \If{$\comp{s}{isCorrupted} \equiv \fullcorrupt$} 
    \Let{$\comp{s'}{pendingRequests}$}{$\an{m, \comp{s}{pendingRequests}}$} \breakalgo{1} \Comment{Collect incoming messages}
    \LetND{$m'$}{$d_{N^p}(s')$}
    \LetND{$a'$}{$\addresses$}
    \State \textbf{stop} $\{(a'{:}a{:}m')\}$, $s'$
  \ElsIf{$\comp{s}{isCorrupted} \equiv \closecorrupt$}
    \Let{$\comp{s'}{pendingRequests}$}{$\an{m, \comp{s}{pendingRequests}}$} \breakalgo{1}\Comment{Collect incoming messages}
    \Let{$N^\text{clean}$}{$N^p \setminus \{n | n \inPairing \comp{s}{nonces}\}$} \label{line:key-not-used-anymore}
    \LetND{$m'$}{$d_{N^\text{clean}}(s')$}
    \LetND{$a'$}{$\addresses$}
    \Let{$\comp{s'}{nonces}$}{$\comp{s}{nonces}$}
    \State \textbf{stop} $\{(a'{:}a{:}m')\}$, $s'$
  \EndIf
  \Let {$n$, $s'$}{\textsf{TAKENONCE}$(s')$}
  \If{$m \equiv \trigger$} \Comment{A special trigger message. }
    \LetND{$\mi{switch}$}{$\{1,2\}$}
    \If{$\mi{switch} \equiv 1$} \Comment{Run some script.}
      \LetNDST{$\ptr{w}$}{$\mathsf{Subwindows}(s')$}{$\comp{\compn{s'}{\ptr{w}}}{documents} \neq \an{}$\breakalgohook{2}}{\textbf{stop} $\{\}$, $s'$}
      \Let{$\ptr{d}$}{$\ptr{w} \plusPairing \str{activedocument}$}
      \State \textsf{RUNSCRIPT}($\ptr{w}$, $\ptr{d}$, $s'$)
    \ElsIf{$\mi{switch} \equiv 2$} \Comment{Create some new request.}
      \Let{$w'$}{$\an{n, \an{}, \bot}$}
      \Append{$w'$}{$\comp{s'}{windows}$}
      \LetND{$\mi{protocol}$}{$\{\http, \https\}$}
      \LetND{$\mi{host}$}{$\dns$}
      \LetND{$\mi{path}$}{$\mathbb{S}$}
      \LetND{$\mi{parameters}$}{$\dict{\mathbb{S}}{\mathbb{S}}$}
      \Let {$n'$, $s'$}{\textsf{TAKENONCE}$(s')$}
      \Let{$\mi{req}$}{$\hreq{
          nonce=n',
          method=\mGet,
          host=\mi{host},
          path=\mi{path},
          headers=\an{},
          parameters=\mi{parameters},
          body=\an{}
        }$}
      \State \textsf{SEND}($n$, $\mi{req}$, $\mi{protocol}$, $\bot$, $s'$)\label{line:send-random}
    \EndIf
  \ElsIf{$m \equiv \fullcorrupt$} \Comment{Request to corrupt browser}
    \Let{$\comp{s'}{isCorrupted}$}{$\fullcorrupt$}
    \State \textbf{stop} $\{\}$, $s'$
  \ElsIf{$m \equiv \closecorrupt$} \Comment{Close the browser}
    \Let{$\comp{s'}{secrets}$}{$\an{}$}  
    \Let{$\comp{s'}{windows}$}{$\an{}$}
    \Let{$\comp{s'}{pendingDNS}$}{$\an{}$}
    \Let{$\comp{s'}{pendingRequests}$}{$\an{}$}
    \Let{$\comp{s'}{sessionStorage}$}{$\an{}$}
    \State \textbf{let} $\comp{s'}{cookies} \subsetPairing \cookies$ \textbf{such that} \breakalgohook{1} $(c \inPairing \comp{s'}{cookies}) {\iff} (c \inPairing \comp{s}{cookies} \wedge \comp{\comp{c}{content}}{session} \equiv \bot$)
    \Let{$\comp{s'}{isCorrupted}$}{$\closecorrupt$}
    \State \textbf{stop} $\{\}$, $s'$
  \ElsIf{$\exists\, \an{\mi{reference}, \mi{request}, \mi{key}, f}$
      $\inPairing \comp{s'}{pendingRequests}$ \breakalgohook{0}
      \textbf{such that} $\proj{1}{\decs{m}{\mi{key}}} \equiv \cHttpResp$ }
    \Comment{Encrypted HTTP response}
    \Let{$m'$}{$\decs{m}{\mi{key}}$}
    \If{$\comp{m'}{nonce} \not\equiv \comp{\mi{request}}{nonce}$}
      \State \textbf{stop} $\{\}$, $s$
    \EndIf
    \State \textbf{remove} $\an{\mi{reference}, \mi{request}, \mi{key}, f}$ \textbf{from} $\comp{s'}{pendingRequests}$
    \State \textsf{PROCESSRESPONSE}($m'$, $\mi{reference}$, $\mi{request}$, $\https$, $s'$)
  \ElsIf{$\proj{1}{m} \equiv \cHttpResp$ $\wedge$ $\exists\, \an{\mi{reference}, \mi{request}, \bot, f}$ $\inPairing \comp{s'}{pendingRequests}$ \breakalgohook{0}\textbf{such that} $\comp{m'}{nonce} \equiv \comp{\mi{request}}{key}$ }
    \State \textbf{remove} $\an{\mi{reference}, \mi{request}, \bot, f}$ \textbf{from} $\comp{s'}{pendingRequests}$
    \State \textsf{PROCESSRESPONSE}($m$, $\mi{reference}$, $\mi{request}$, $\http$, $s'$)
  \ElsIf{$m \in \dnsresponses$} \Comment{Successful DNS response}
      \If{$\comp{m}{nonce} \not\in \comp{s}{pendingDNS}$}
        \State \textbf{stop} $\{\}$, $s$
      \EndIf
      \Let{$\an{\mi{reference}, \mi{message}, \mi{protocol}}$}{$\comp{s}{pendingDNS}[\comp{m}{nonce}]$}
      \If{$\mi{protocol} \equiv \https$}
        \Let{$k, s'$}{\textsf{TAKENONCE}($s'$)} \label{line:takenonce-k}
        \AppendBreak{2}{$\langle\mi{reference}$, $\mi{message}$, $\mi{k}$, $\comp{m}{result}\rangle$}{$\comp{s'}{pendingRequests}$} \label{line:add-to-pendingrequests-https}
        \Let{$\mi{message}$}{$\enc{\an{\mi{message},\mi{k}}}{\comp{s'}{keyMapping}\left[\comp{\mi{message}}{host}\right]}$} \label{line:select-enc-key}
      \Else
        \AppendBreak{2}{$\langle\mi{reference}$, $\mi{message}$, $\bot$, $\comp{m}{result}\rangle$}{$\comp{s'}{pendingRequests}$} \label{line:add-to-pendingrequests}
      \EndIf
      \Let{$\comp{s'}{pendingDNS}$}{$\comp{s'}{pendingDNS} - \comp{m}{nonce}$}
      \State \textbf{stop} $\{(\comp{m}{result}{:}a{:}\mi{message})\}$, $s'$
  \Else
    \State \textbf{stop} $\{\}$, $s$
  \EndIf
\end{algorithmic} \setlength{\parindent}{1em}

%% file: appendix-general-properties.tex
\section{General Security Properties of the Web Model}\label{app:generalproperties}

We now formally state and prove the general application independent
security properties of the web which in Section
\ref{sec:generalproperties} have been sketched only.

Let $\mathpzc{Web} = (\bidsystem, \scriptset, \mathsf{script}, E_0)$
be a web system. In the following, we write $s_x = (S_x,E_x)$ for the
states of a web system.

\begin{definition}\label{def:emitting-pidp}
  In what follows, given an atomic process $p$ and a message $m$, we
  say that \emph{$p$ emits $m$} in a run $\rho=s_0,s_1,\ldots$ if
  there is a processing step  of the form
\[ s_{u-1} \xrightarrow[p \rightarrow E]{} s_{u}\] for some
$u \in \mathbb{N}$, a set of events $E$ and some addresses
$x$, $y$ with $(x{:}y{:}m) \in E$.
\end{definition}

\begin{definition}\label{def:contains}
  We say that a term $t$ \emph{is derivably contained in (a term) $t'$
    for (a set of DY processes) $P$ (in a processing step $s_i
    \rightarrow s_{i+1}$ of a run $\rho=s_0,s_1,\ldots$)} if $t$ is
  derivable from $t'$ with the knowledge available to $P$, i.e.,
  \begin{align*}
     t \in d_{\eta}(\{t'\}\cup \varsigma) \text{ with }  \eta := \bigcup_{p\in P}N^p \text{ and } \varsigma := \bigcup_{p\in P, j\leq i}S_j(p)\ .
  \end{align*}

\end{definition}

\begin{definition}\label{def:leak}
  We say that \emph{a set of processes $P$ leaks a term $t$ (in a
    processing step $s_i \rightarrow s_{i+1}$) to a set of processes
    $P'$} if there exists a message $m$ that is emitted (in $s_i
  \rightarrow s_{i+1}$) by some $p \in P$ and $t$ is derivably
  contained in $m$ for $P'$ in the processing step $s_i \rightarrow
  s_{i+1}$. If we omit $P'$, we define $P' := \bidsystem \setminus
  P$. If $P$ is a set with a single element, we omit the set notation.
\end{definition}

\begin{definition}\label{def:creating-pidp}
  We say that an DY process $p$ \emph{created} a message $m$ (at
  some point) in a run if $m$ is derivably contained in a message
  emitted by $p$ in some processing step and if there is no earlier
  processing step where $m$ is derivably contained in a message
  emitted by some DY process $p'$.
\end{definition}

\begin{definition}\label{def:accepting-pidp}
  We say that a browser $b$ \emph{accepted} a message (as a response
  to some request) if the browser decrypted the message (if it was an
  HTTPS message) and called the function $\mathsf{PROCESSRESPONSE}$,
  passing the message and the request (see
  Algorithm~\ref{alg:processresponse}).
\end{definition}

\begin{definition}\label{def:knowing-pidp}
  We say that an atomic DY process \emph{$p$ knows a term $t$} in some
  state $s=(S,E)$ of a run if it can derive the term from its
  knowledge, i.e., $t \in d_{N^p}(S(p))$.
\end{definition}

\begin{definition}\label{def:initiating-pidp}
  We say that a \emph{script initiated a request $r$} if a browser
  triggered the script (in Line~\ref{line:trigger-script} of
  Algorithm~\ref{alg:runscript}) and the first component of the
  $\mi{command}$ output of the script relation is either $\tHref$,
  $\tIframe$, $\tForm$, or $\tXMLHTTPRequest$ such that the browser
  issues the request $r$ in the same step as a result.
\end{definition}

For a run $\rho = s_0, s_1,\dots$ of any $\mathpzc{Web}$, we state the
following lemmas:

\begin{lemma}\label{lemma:k-does-not-leak-from-honest-browser-general} %
  If in the processing step $s_i \rightarrow s_{i+1}$ of a run $\rho$ of $\mathpzc{Web}$ an honest browser $b$  (I)
  emits an HTTPS request of the form

  \[ m = \ehreqWithVariable{\mi{req}}{k}{\pub(k')} \]
  (where $\mi{req}$ is an HTTP request, $k$ is a nonce (symmetric
  key), and $k'$ is the private key of some other DY process $u$), and (II) in the
  initial state $s_0$ the private key $k'$ is only known to $u$, and
  (III) $u$ never leaks $k'$, then all of the following
  statements are true:
  \begin{enumerate}[label=(\arabic*)]
  \item There is no state of $\mathpzc{Web}$ where any party except
    for $u$ knows $k'$, thus no one except for $u$ can
    decrypt $\mi{req}$.
    \label{prop:attacker-cannot-decrypt-general}
  \item If there is a processing step $s_j \rightarrow s_{j+1}$ where
    the browser $b$ leaks $k$ to $\bidsystem \setminus \{u, b\}$ there
    is a processing step $s_h \rightarrow s_{h+1}$ with $h < j$
    where $u$ leaks the symmetric key $k$ to $\bidsystem \setminus
    \{u,b\}$ or the browser is fully corrupted in
    $s_j$. \label{prop:k-doesnt-leak-general}
  \item The value of the host header in $\mi{req}$ is the domain that
    is assigned the public key $\pub(k')$ in the browsers' keymapping
    $s_0.\str{keymapping}$ (in its initial
    state). \label{prop:host-header-matches-general}
  \item If $b$ accepts a response (say, $m'$) to $m$ in a processing step $s_j
    \rightarrow s_{j+1}$ and $b$ is honest in $s_j$ and $u$ did not
    leak the symmetric key $k$ to $\bidsystem \setminus \{u,b\}$ prior
    to $s_j$, then $u$ created the HTTPS response $m'$ to the HTTPS
    request $m$, i.e., the nonce of the HTTP request $\mi{req}$ is not known to
    any atomic process $p$, except for the atomic process $b$ and
    $u$.\label{prop:only-owner-answers-general}
  \end{enumerate}
\end{lemma}

\begin{proof} 

  \textbf{\ref{prop:attacker-cannot-decrypt-general}} follows
  immediately from the condition. If $k'$ is initially only known to
  $u$ and $u$ never leaks $k'$, i.e., even with the knowledge of all
  nonces (except for those of $u$), $k'$ can never be derived from any network
  output of $u$, $k'$ cannot be known to any other party. Thus, nobody
  except for $u$ can derive $\mi{req}$ from $m$.

  \textbf{\ref{prop:k-doesnt-leak-general}} 
  We assume that $b$ leaks $k$ to $\bidsystem \setminus \{u,b\}$ in
  the processing step $s_j \rightarrow s_{j+1}$ without $u$ prior
  leaking the key $k$ to anyone except for $u$ and $b$ and that the
  browser is not fully corrupted in $s_j$, and lead this to a contradiction.

  The browser is honest in $s_i$. From the definition of the browser
  $b$, we see that the key $k$ is always chosen from a fresh set of
  nonces (Line~\ref{line:takenonce-k} of
  Algorithm~\ref{app:mainalgorithmwebbrowserprocess}) that are not
  used anywhere else. Further, the key is stored in the browser's
  state in $\mi{pendingRequests}$. The information from
  $\mi{pendingRequests}$ is not extracted or used anywhere else (in
  particular it is not accessible by scripts). If the browser becomes
  closecorrupted prior to $s_j$ (and after $s_i$), the key cannot be
  used anymore (compare Line~\ref{line:key-not-used-anymore} of
  Algorithm~\ref{alg:browsermain}). Hence, $b$ does not leak $k$ to
  any other party in $s_j$ (except for $u$ and $b$). This proves
  \ref{prop:k-doesnt-leak-general}.

  \textbf{\ref{prop:host-header-matches-general}} Per the
  definition of browsers (Algorithm~\ref{alg:browsermain}), a host
  header is always contained in HTTP requests by browsers. From
  Line~\ref{line:select-enc-key} of Algorithm~\ref{alg:browsermain} we
  can see that the encryption key for the request $\mi{req}$ was
  chosen using the host header of the message. It is chosen from the
  $\mi{keymapping}$ in the browser's state, which is never changed
  during $\rho$.  This proves
  \ref{prop:host-header-matches-general}.

  \textbf{\ref{prop:only-owner-answers-general}} An HTTPS response
  $m'$ that is accepted by $b$ as a response to $m$ has to be encrypted
  with $k$.  The nonce $k$ is stored by the browser in the
  $\mi{pendingRequests}$ state information. The browser only stores
  freshly chosen nonces there (i.e., the nonces are not used twice, or
  for other purposes than sending one specific request). The
  information cannot be altered afterwards (only deleted) and cannot
  be read except when the browser checks incoming messages. The nonce
  $k$ is only known to $u$ (which did not leak it to any other party
  prior to $s_j$) and $b$ (which did not leak it either, as $u$ did
  not leak it and $b$ is honest, see
  \ref{prop:k-doesnt-leak-general}). The browser $b$ cannot send
  responses. This proves \ref{prop:only-owner-answers-general}.  \qed
\end{proof}

\begin{corollary}\label{cor:k-does-not-leak-from-honest-browser-general}
  In the situation of
  Lemma~\ref{lemma:k-does-not-leak-from-honest-browser-general}, as
  long as $u$ does not leak the symmetric key $k$ to $\bidsystem
  \setminus \{u,b\}$ and the browser does not become fully corrupted,
  $k$ is not known to any DY process $p \not\in\{b,u\}$ (i.e.,
  $\nexists\,s' = (S',E') \in \rho$: $k \in d_{N^p}(S'(p))$).
\end{corollary}

\begin{lemma} \label{lemma:https-document-origin-general} If for some
  $s_i \in \rho$ an honest browser $b$ has a document $d$ in its state
  $S_i(b).\str{windows}$ with the origin $\an{\mi{dom}, \https}$ where
  $\mi{dom} \in \mathsf{Domain}$, and
  $S_i(b).\str{keyMapping}[\mi{dom}] \equiv \pub(k)$ with
  $k\in\nonces$ being a private key, and there is only one DY process
  $p$ that knows the private key $k$ in all $s_j$, $j \leq i$, then
  $b$ extracted (in Line~\ref{line:take-script} in
  Algorithm~\ref{alg:processresponse}) the script in that document
  from an HTTPS response that was created by $p$.
\end{lemma}

\begin{proof}
  The origin of the document $d$ is set only once: In
  Line~\ref{line:set-origin-of-document} of
  Algorithm~\ref{alg:processresponse}. The values (domain and
  protocol) used there stem from the information about the request
  (say, $\mi{req}$) that led to loading of $d$. These values have been
  stored in $\mi{pendingRequests}$ between the request and the
  response actions. The contents of $\mi{pendingRequests}$ are indexed
  by freshly chosen nonces and can never be altered or overwritten
  (only deleted when the response to a request arrives). The
  information about the request $\mi{req}$ was added to
  $\mi{pendingRequests}$ in
  Line~\ref{line:add-to-pendingrequests-https} (or
  Line~\ref{line:add-to-pendingrequests} which we can exclude as we
  will see later) of Algorithm~\ref{alg:browsermain}. In particular,
  the request was an HTTPS request iff a (symmetric) key was added to
  the information in $\mi{pendingRequests}$. When receiving the
  response to $\mi{req}$, it is checked against that information and
  accepted only if it is encrypted with the proper key and contains
  the same nonce as the request (say, $n$). Only then the protocol
  part of the origin of the newly created document becomes
  $\https$. The domain part of the origin (in our case $\mi{dom}$) is
  taken directly from the $\mi{pendingRequests}$ and is thus
  guaranteed to be unaltered.

  From Line~\ref{line:select-enc-key} of
  Algorithm~\ref{alg:browsermain} we can see that the encryption key
  for the request $\mi{req}$ was actually chosen using the host header
  of the message which will finally be the value of the origin of the
  document $d$. Since $b$ therefore selects the public key
  $S_i(b).\str{keyMapping}[\mi{dom}] =
  S_0(b).\str{keyMapping}[\mi{dom}]\equiv \pub(k)$ for $p$ (the key
  mapping cannot be altered during a run), we can see that $\mi{req}$
  was encrypted using a public key that matches a private key which is
  only (if at all) known to $p$. With
  Lemma~\ref{lemma:k-does-not-leak-from-honest-browser-general} we see
  that the symmetric encryption key for the response, $k$, is only
  known to $b$ and the respective web server. The same holds for the
  nonce $n$ that was chosen by the browser and included in the
  request. Thus, no other party than $p$ can encrypt a response that
  is accepted by the browser $b$ and which finally defines the script
  of the newly created document. \qed
\end{proof}

\begin{lemma} \label{lemma:https-script-origin-general} If in a
  processing step $s_i \rightarrow s_{i+1}$
  of a run $\rho$ of $\mathpzc{Web}$ an honest browser $b$ issues an
  HTTP(S) request with the Origin header value $\an{\mi{dom}, \https}$
  where and $S_i(b).\str{keyMapping}[\mi{dom}] \equiv \pub(k)$ with
  $k\in\nonces$ being a private key, and there is only one DY process
  $p$ that knows the private key $k$ in all $s_j$, $j \leq i$, then
  that request was initiated by a script that $b$ extracted (in
  Line~\ref{line:take-script} in Algorithm~\ref{alg:processresponse})
  from an HTTPS response that was created by $p$.
\end{lemma}

\begin{proof} First, we can see that the request was
  initiated by a script: As it contains an origin header,
  it must have been a POST request (see the browser
  definition in Appendix~\ref{sec:descr-web-brows}). POST
  requests can only be initiated in
  Lines~\ref{line:send-form}, \ref{line:send-xhr} of
  Algorithm~\ref{alg:runscript} and
  Line~\ref{line:send-redirect} of
  Algorithm~\ref{alg:processresponse}. In the latter
  instance (Location header redirect), the request contains
  at least two different origins, therefore it is
  impossible to create a request with exactly the origin
  $\an{\mi{dom}, \https}$ using a redirect.  In the other
  two cases (FORM and XMLHTTPRequest), the request was
  initiated by a script.

  The Origin header of the request is defined by the origin
  of the script's document. With
  Lemma~\ref{lemma:https-document-origin-general} we see that the
  content of the document, in particular the script, was
  indeed provided by $p$.  \qed
\end{proof}

%% file: appendix-browserid-lowlevel.tex
\section{Step-By-Step Description of BrowserID (Primary IdP)}\label{app:browserid-lowlevel}

We now present additional details of the implementation of
BrowserID\@. While the basic steps have been shown in
Section~\ref{sec:javascript-descr}, we will now again refer
to Figure~\ref{fig:browserid-lowlevel-ld} and provide a
step-by-step description. As above, for brevity of
presentation, we focus on the main login flow without the
CIF, and we leave out steps for fetching additional
resources (like JavaScript files) and some less relevant
postMessages and \xhrs. Also, we assume that a typical IdP
implementation like the example implementation provided by
Mozilla is used.

We emphasize, however, that our formal model of BrowserID
with primary IdPs
(cf.~Appendix~\ref{sec:analysisbrowserid-pidp}) closely
follows the full BrowserID implementation (see also
Figure~\ref{fig:browserid-lowlevel-pidp-detailed-1} on
Pages~\pageref{fig:browserid-lowlevel-pidp-detailed-1}
and~\pageref{fig:browserid-lowlevel-pidp-detailed-2}, which
is an extended version of
Figure~\ref{fig:browserid-lowlevel-ld}).

\subsection{LPO Sessions} 
Before we describe the login flow step-by-step, we first
introduce LPO sessions.

LPO establishes a session with the browser by setting a
cookie \texttt{browserid\_state}
(Step~\refprotostep{ld-ctx-1} in
Figure~\ref{fig:browserid-lowlevel-ld}) on the client-side.
LPO considers such a session authenticated after having
received a valid CAP (Step~\refprotostep{ld-lpo-auth} in
Figure~\ref{fig:browserid-lowlevel-ld}). In future runs,
the user is presented a list of her email addresses (which
is fetched from LPO) in order to choose one address. Then,
she is asked if she trusts the computer she is using and is
given the option to be logged in for one month or ``for
this session only'' (\emph{ephemeral} session). In order to
use any of the email addresses, the user is required to
authenticate to the IdP responsible for that address to get
an UC issued. If the localStorage (under the origin LPO)
already contains a valid UC, then, however, authentication
at the IdP is not necessary.

\subsection{Step-By-Step Description}

We (again) assume that the user uses a ``fresh'' browser,
i.e., the user has not been logged in before. The user has
already opened a document of some RP (RP-Doc) in her
browser. RP-Doc includes a JavaScript file, which provides
the BrowserID API. The user is now about to click on a
login button in order to start a BrowserID login.

\subsubsection{Phase~\refprotophase{ld-start-1}.} After the
user has clicked on the login button, RP-Doc opens a new
browser window, the \emph{login dialog}
(LD)~\refprotostep{ld-open}. The document of LD is loaded
from LPO~\refprotostep{ld-init-1}. Now, LD sends a
\emph{ready} postMessage~\refprotostep{ld-rpdoc-ready-1} to
its opener, which is RP-Doc. RP-Doc then responds by
sending a \emph{request}
postMessage~\refprotostep{rpdoc-ld-request-1}. This
postMessage may contain additional information like a name
or a logo of RP-Doc.  LD then fetches the so-called
\emph{session context} from LPO
using~\refprotostep{ld-ctx-1}. The session context contains
information about whether the user is already logged in at
LPO, which, by our assumption, is not the case at this
point. The session context also contains an \xsrf
protection token which will be sent in all subsequent POST
requests to LPO\@. Also, an $\str{httpOnly}$ cookie called
\texttt{browserid\_state} is set, which contains an LPO
session identifier.  Now, the user is prompted to enter her
email address (\emph{login email address}), which she wants
to use to log in at RP~\refprotostep{ld-user-email}. LD
sends the login email address to LPO via an
\xhr~\refprotostep{ld-addrinfo-1}, in order to get
information about the IdP the email address belongs to. The
information from this so-called \emph{support document} may
be cached at LPO for further use. LPO extracts the domain
part of the login email address and fetches an information
document~\refprotostep{lpo-idp-wk-1} from a fixed path
(\url{/.well-known/browserid}) at the IdP. This document
contains the public key of IdP, and two paths, the
provisioning path and the authentication path at IdP. These
paths will be used later in the login process by LD. LPO
converts these paths into URLs and sends them in its
response~\refprotostep{ld-addrinfo-1-resp} to the
requesting \xhr~\refprotostep{ld-addrinfo-1}.

\subsubsection{Phase~\refprotophase{ld-prov-1}.} As there is
no record about the login email address in the localStorage under
the origin of LPO, the LD now tries to get a UC for this
identity. For that to happen, the LD creates a new iframe,
the \emph{provisioning iframe}
(PIF)~\refprotostep{ld-pif-open-1}. The PIF's document is
loaded~\refprotostep{pif-init-1} from the provisioning URL
LD has just received before
in~\refprotostep{ld-addrinfo-1-resp}. The PIF now interacts
with the LD via
postMessages~\refprotostep{pif-ld-pms-1}. As the user is
currently not logged in, the PIF tells the LD that the user
is not authenticated yet. This also indicates to the LD
that the PIF has finished operation. The LD then closes the
PIF~\refprotostep{ld-pif-close-1}.

\subsubsection{Phase~\refprotophase{ld-auth}.} Now, the LD
saves the login email address in the localStorage indexed
by a fresh nonce. This nonce is stored in the
sessionStorage to retrieve the email address later from the
localStorage again. Next, the LD navigates itself to the
authentication URL it has received
in~\refprotostep{ld-addrinfo-1-resp}. The loaded document
now interacts with the user and the
IdP~\refprotostep{idp-ld-auth} in order to establish some
authenticated session depending on the actual IdP
implementation, which is out of scope of the BrowserID
standard. For example, during this authentication
procedure, the IdP may issue some session cookie.

\subsubsection{Phase~\refprotophase{ld-start-2}.} After the
authentication to the IdP has been completed, the
authentication document navigates the LD to the LD URL
again. The LD's document is fetched again from LPO and the
login process starts over. The following steps are similar
to Phase~\refprotophase{ld-start-1}: The ready and request
postMessages are exchanged and the session context is
fetched. As the user has not been authenticated to LPO yet,
the session context still contains the same information as
above in~\refprotostep{ld-ctx-1}. Now, the user is not
prompted to enter her email address again. The email
address is fetched from the localStorage under the index of
the nonce stored in the sessionStorage. Now, the address
information is requested again from LPO.

\subsubsection{Phase~\refprotophase{ld-prov-2}.} As there
still is no UC belonging to the login email address in the
localStorage, the PIF is created again. As the user now has
established an authenticated session with the IdP, the PIF
asks the LD to generate a fresh key pair. After the LD has
generated the key pair~\refprotostep{gen-key-pair}, it
stores the key pair in the localStorage (under the origin
of LPO) and sends the public key to the PIF as a
postMessage~\refprotostep{pubkey-ld-pif}. The following
steps \refprotostep{req-uc}--\refprotostep{send-uc} are not
specified in the BrowserID protocol. Typically, the PIF
would send the public key to IdP (via an
\xhr)~\refprotostep{req-uc}. The IdP would create the
UC~\refprotostep{certify-uc} and send it back to the
PIF~\refprotostep{send-uc}.  The PIF then sends the UC to
the LD~\refprotostep{recv-uc}, which stores it in the
localStorage. Now, the LD closes the PIF.

\subsubsection{Phase~\refprotophase{ld-lpo-auth}.} The LD is
now able to create a CAP, as it has access to a UC and the
corresponding private key in its localStorage. First, LD
creates an IA for LPO~\refprotostep{ld-gen-cap-lpo}. The IA
and the UC is then combined to a CAP, which is then sent to
LPO in an \xhr POST message~\refprotostep{ld-lpo-auth}. LPO
is now able to verify this CAP with the public key of IdP,
which LPO has already fetched and cached before
in~\refprotostep{lpo-idp-wk-1}. If the CAP is valid, LPO
considers its session with the user's browser to be
authenticated for the email address the UC in the CAP is
issued for.

\subsubsection{Phase~\refprotophase{ld-cap}.} Now,
in~\refprotostep{ld-lpo-list-emails}, the LD fetches a list
of email addresses, which LPO considers to be owned by the
user. If the login email address would not appear in this
list, LD would abort the login process. After this, the LD
fetches the address information about the login email
address again in~\refprotostep{ld-addrinfo-3}. Using this
information, LD validates if the UC is signed by the
correct party (primary/secondary IdP).  Now, LD generates
an IA for the sender's origin of the request
postMessage~\refprotostep{rpdoc-ld-request-1} (which was
repeated in Phase~\refprotophase{ld-start-2}) using the
private key from the localStorage~\refprotostep{ld-gen-cap} (the IA is generated for
the login email address). Also, it is recorded in the
localStorage that the user is now logged in at RP with this
email address. The LD then combines the IA with the UC
stored in the localStorage to the CAP, which is then sent
to RP-Doc in the \emph{response}
postMessage~\refprotostep{ld-rpdoc-cap}.

This concludes the login process that runs in
LD. Afterwards, RP-Doc closes LD~\refprotostep{ld-close}.

%% file: appendix-browserid-pidp.tex
\section{Model of BrowserID with Primary IdPs}
\label{sec:analysisbrowserid-pidp}

We now present the full details of our formal model of BrowserID with
primary IdPs and the fixes discussed in
Section~\ref{sec:attackbrowserid} applied.  We consider ephemeral
sessions (the default), which are supposed to last until the browser
is closed.

We model the BrowserID system as a web system (in the sense of
Section~\ref{sec:webmodel}). We call a web system
$\bidwebsystem=(\bidsystem, \scriptset, \mathsf{script}, E_0)$ a
\emph{BrowserID web system} if it is of the form described in what
follows.

\subsection{Outline}\label{app:outlinebrowserIDmodel}
The system $\bidsystem=\mathsf{Hon}\cup \mathsf{Web} \cup
\mathsf{Net}$ consists of the (network) attacker process
$\fAP{attacker}$, the web server for $\fAP{LPO}$, a finite set
$\fAP{B}$ of web browsers, a finite set $\fAP{RP}$ of web servers for
the relying parties, and a finite set $\fAP{IDP}$ of web servers for
the identity providers, with $\mathsf{Hon} := \fAP{B} \cup \fAP{RP}
\cup \fAP{IDP} \cup \{\fAP{LPO}\}$, $\mathsf{Web} := \emptyset$, and $\mathsf{Net} :=
\{\fAP{attacker}\}$. DNS servers are assumed to be dishonest, and
hence, are subsumed by $\fAP{attacker}$. More details on the processes
in $\bidsystem$ are provided below.
Figure~\ref{fig:scripts-in-w} shows the set of scripts $\scriptset$
and their respective string representations that are defined by the
mapping $\mathsf{script}$.
The set $E_0$ contains only the trigger events as specified in
Section~\ref{sec:websystem}.

\begin{figure}[htb]
  \centering
  \begin{tabular}{|@{\hspace{1ex}}l@{\hspace{1ex}}|@{\hspace{1ex}}l@{\hspace{1ex}}|}\hline 
    \hfill $s \in \scriptset$\hfill  &\hfill  $\mathsf{script}(s)$\hfill  \\\hline\hline
    $\Rasp$ & $\str{att\_script}$  \\\hline
    $\mi{script\_rp\_index}$ & $\str{script\_rp\_index}$  \\\hline
    $\mi{script\_lpo\_cif}$ &  $\str{script\_lpo\_cif}$  \\\hline

    $\mi{script\_lpo\_ld}$ & $\str{script\_lpo\_ld}$ \\\hline
    $\mi{script\_idp\_pif}$ & $\str{script\_idp\_pif}$ \\\hline
    $\mi{script\_idp\_ad}$ & $\str{script\_idp\_ad}$ \\\hline
  \end{tabular}
  
  \caption{List of scripts in $\scriptset$ and their respective string
    representations.}
  \label{fig:scripts-in-w}
\end{figure}

This outlines $\bidwebsystem$. We will now define the DY processes in
$\bidwebsystem$ and their addresses, domain names, and secrets in more
detail.

\subsection{Addresses and Domain Names}
The set $\addresses$ contains for $\fAP{LPO}$, $\fAP{attacker}$, every
relying party in $\fAP{RP}$, every identity provider in $\fAP{IDP}$,
and every browser in $\fAP{B}$ one address each. By $\mapAddresstoAP$
we denote the corresponding assignment from a process to its address.
The set $\dns$ contains one domain for $\fAP{LPO}$, one for every
relying party in $\fAP{RP}$, a finite set of domains for every
identity provider in $\fAP{IDP}$, and a finite set of domains for
$\fAP{attacker}$. Browsers (in $\fAP{B})$ do not have a domain.

By $\mapAddresstoAP$ and $\mapDomain$ we denote the assignments from
atomic processes to sets of $\addresses$ and $\dns$, respectively. If
$\mapDomain$ or $\mapAddresstoAP$ returns a set with only one element,
we often write $\mapDomain(x)$ or $\mapAddresstoAP(x)$ to refer to the
element.

\subsection{Keys and Secrets}
The set $\nonces$ of nonces is partitioned into four sets, an infinite
set $N^\bidsystem$, an infinite set $K_\text{SSL}$, an infinite set
$K_\text{sign}$, and a finite set $\RPSecrets$. We thus have
\begin{align*}
\def\hereMaxHeightPhantom{\vphantom{K_{\text{p}}^\bidsystem}}
\nonces = 
\underbrace{N^\bidsystem\hereMaxHeightPhantom}_{\text{infinite}} 
\dot\cup \underbrace{K_{\text{SSL}}\hereMaxHeightPhantom}_{\text{finite}} 
\dot\cup \underbrace{K_{\text{sign}}\hereMaxHeightPhantom}_{\text{finite}} 
\dot\cup \underbrace{\RPSecrets\hereMaxHeightPhantom}_{\text{finite}}\ .
\end{align*}
The set $N^\bidsystem$ contains the nonces that are available for each
DY process in $\bidsystem$. It is partitioned into infinite sets of
nonces, one set $N^p\subseteq N^\bidsystem$ for every $p\in
\bidsystem$.

The set $K_\text{SSL}$ contains the keys that will be used for SSL
encryption. Let $\mapSSLKey\colon \dns \to K_\text{SSL}$ be an injective
mapping that assigns a (different) private key to every domain.

The set $K_\text{sign}$ contains the keys that will be used by IdPs
for signing UCs. Let $\mapSignKey\colon \fAP{IdPs} \to K_\text{sign}$
be an injective mapping that assigns a (different) private key to every identity
provider.

The set $\RPSecrets\subseteq \nonces$ is the
set of passwords (secrets) the browsers share with the identity
providers. 

\subsection{Identities}\label{app:browserid-pidp-identities}
Indentites are email addresses, which consist of a user name and a
domain part. For our model, this is defined as follows:
\begin{definition}
  An \emph{identity} (email address) $i$ is a term of the form
  $\an{\mi{name},\mi{domain}}$ with $\mi{name}\in \mathbb{S}$ and
  $\mi{domain} \in \dns$.

  Let $\IDs$ be the finite set of identities. By $\IDs^y$ we denote
  the set $\{ \an{\mi{name}, \mi{domain}} \in \IDs\,|\, \mi{domain}
  \in \mapDomain(y) \}$.

  We say that an ID is \emph{governed} by the DY process to which the
  domain of the ID belongs. Formally, we define the mapping $\mapGovernor:
  \IDs \to \bidsystem$, $\an{\mi{name}, \mi{domain}} \mapsto
  \mapDomain^{-1}(\mi{domain})$.
\end{definition}%
The governor of an ID will usually be an IdP, but could also be the
attacker. Note that we omit delegation of authority over domains.

We further define UCs, IAs and CAPs formally:
\begin{definition}
  A (valid) \emph{user certificate} (UC) $\mi{uc}$ for a user $u$ with
  email address $\mi{id} = \an{\mi{name}, d}$ and public key (verification key)
  $\pub(k_u)$, where $d \in \mapDomain(y)$ is a domain of the
  governor $y$ of $\mi{id}$
  and $k_u$ is the private key (signing key) of $u$, is a message of
  the form $\mi{uc}=\sig{\an{\an{\mi{name}, d},
      \pub(k_u)}}{\mapSignKey(y)}$.

  An (valid) \emph{identity assertion} (IA) $ia$ for an origin
  $\mi{o}$ (e.g., $\an{\str{example.com},\str{S}}$) signed with the
  key $k_u$ is a message of the form $ia=\sig{o}{k_u}$.

  A \emph{certificate assertion pair} (CAP) is of the form
  $\an{\mi{uc},\mi{ia}}$, with $\mi{uc}$ and $\mi{ia}$ as
  above.\footnote{Note that the time stamps are omitted both from the
    UC and the IA\@. This models that both certificates are valid
    indefinitely. In reality, they are valid for a certain period
    of time, as indicated by the time stamps. So our modeling is a
    safe overapproximation.}
\end{definition}

Each browser $b\in \fAP{B}$ owns a set of secrets ($\in
\RPSecrets$). Each secret is assigned a set $S$ of IDs for a specific
IdP $y$ such that $S \subseteq \IDs^y$. Browsers have disjoint secrets
and secrets have disjoint sets of IDs.  The IdPs of the secrets of a
browser are disjoint. An \emph{ID $i$ is owned by a browser $b$} if
the identity associated with $i$ \emph{belongs to} $b$:

Let $\mapPLItoOwner: \RPSecrets \to \fAP{B}$ denote the mapping that
assigns to each secret the browser that owns this secret. Let
$\mapIDtoPLI: \IDs \to \RPSecrets$ denote the mapping that assigns to
each identity the associated secret. Now, we define the mapping
$\mapIDtoOwner: \IDs \to \fAP{B}$, $i \mapsto
\mapPLItoOwner(\mapIDtoPLI(i))$, which assigns to each identity the
browser that owns this identity (we say that the identity belongs to
the browser).

\subsection{Corruption}
RPs and IdPs can become corrupted: If they receive the message
$\corrupt$, they start collecting all incoming messages in their state
and (upon triggering) send out all messages that are derivable from
their state and collected input messages, just like the attacker
process. We say that an RP or IdP is \emph{honest} if the according
part of their state ($s.\str{corrupt}$) is $\bot$, and that they are
corrupted otherwise.

We are now ready to define the processes in $\websystem$ as well as
the scripts in $\scriptset$ in more detail. 

\subsection{Processes in $\bidsystem$ (Overview)}

We first provide an overview of the processes in $\bidsystem$. All
processes in $\websystem$ contain in their initial states all public
keys and the private keys of their respective domains (if any). We
define $I^p=\{\mapAddresstoAP(p)\}$ for all $p\in \mathsf{Hon}$.

\subsubsection{Attacker.}  The $\fAP{attacker}$ process is a
network attacker (see Section~\ref{sec:websystem}), who
uses all addresses for sending and listening. All parties
use the attacker as a DNS server. See
Appendix~\ref{app:attacker-pidp} for details.

\subsubsection{Browsers.}  Each $b \in \fAP{B}$ is a web
browser as defined in Section~\ref{sec:web-browsers}. The
initial state contains all secrets owned by $b$, stored
under the origin of the respective IdP. See
Appendix~\ref{sec:browsers-pidp} for details.

\subsubsection{LPO.} LPO is a web server that serves important scripts
($\str{script\_lpo\_cif}$ and $\str{script\_lpo\_ld}$) and manages
user sessions.  See Appendix~\ref{sec:lpo-pidp} for details.

\subsubsection{IdPs.} Each IdP is a web server. IdPs are modeled
following the example implementation provided by Mozilla. As outlined
in Section~\ref{sec:browserid}, users can authenticate to the IdP with
their credentials. IdP tracks the state of the users with
sessions. Authenticated users can receive signed UCs from the
IdP. When receiving a special message ($\corrupt$) IdPs can become
corrupted. Similar to the definition of corruption for the browser,
IdPs then start sending out all messages that are derivable from their
state. See Appendix~\ref{app:idps} for details.

\subsubsection{Relying Parties.}  A relying party $r \in \fAP{RP}$ is a
web server. The definition of $R^r$ follows the description in
Section~\ref{sec:browserid} and the security considerations in
\cite{mozilla/persona/mdn} (Cross-site Request Forgery protection,
e.g., by checking origin headers, and HTTPS only with STS enabled). RP
answers any $\mGet$ request with the script $\str{script\_rp\_index}$
(see below). When receiving an HTTPS $\mPost$ message, RP checks
(among others) if the message contains a valid CAP\@. For this
purpose, all signing keys of the identity providers (see below) are
contained in the initial state of all RPs. If successful, RP responds
with an \emph{RP service token for ID $i$} of the form $\an{n, i}$,
where $i \in \IDs$ is the ID for which the CAP was issued and $n$ is a
freshly chosen nonce. The RP $r$ keeps a list of such tokens in its
state. Intuitively, a client having such a token can use the service
of $r$ for ID $i$. See Appendix~\ref{app:relying-parties-pidp} for
details. Just like IdPs, RPs can become corrupted.

\subsection{Attacker}\label{app:attacker-pidp} As mentioned, the attacker
$\fAP{attacker}$ is modeled to be a network attacker as specified in
Section~\ref{sec:websystem}. We allow it to listen to/spoof all
available IP addresses, and hence, define $I^\fAP{attacker} =
\addresses$. His initial state is $s_0^\fAP{attacker} =
\an{\mi{attdoms}, \mi{sslkeys}, \mi{signkeys}}$, where $\mi{attdoms}$
is a sequence of all domains along with the corresponding private keys
owned by the attacker, $\mi{sslkeys}$ is a sequence of all domains and
the corresponding public keys, and $\mi{signkeys}$ is a sequence
containing all public signing keys for all IdPs. All other parties use
the attacker as a DNS server.

\subsection{Browsers}\label{sec:browsers-pidp} 

Each $b \in \fAP{B}$ is a web browser as defined in
Section~\ref{sec:web-browsers}, with $I^b := \{\mapAddresstoAP(b)\}$
being its address. 

To define the inital state, first let $\mi{ID}^b :=
\mapIDtoOwner^{-1}(b)$ be
 the set of all IDs of $b$, $\mi{ID}^{b,d} :=
\{i \mid \exists\, x:\ i = \an{x, d} \in \mi{ID}^b\}$ be the set of
IDs of $b$ for a domain $d$, and $\mi{SecretDomains}^b := \{d \mid
\mi{ID}^{b,d} \neq \emptyset \}$ be the set of all domains that $b$
owns identities for.

Then, the initial state $s_0^b$ is defined as follows: the key mapping
maps every domain to its public (ssl) key, according to the mapping
$\mapSSLKey$; the DNS address is $\mapAddresstoAP(\fAP{attacker})$;
the list of secrets contains an entry $\an{\an{d,\https}, s}$ for each
$d \in \mi{SecretDomains}^b$ and $s = \mapIDtoPLI(i)$ for some $i \in
\mi{ID}^{b,d}$ ($s$ is the same for all $i$); $\mi{ids}$ is
$\an{\mi{ID}^b}$; $\mi{sts}$ is empty.

\subsection{LPO} \label{sec:lpo-pidp}

$\LPO$ is a an atomic DY process $(I^\fAP{LPO}, Z^\fAP{LPO},
R^\fAP{LPO}, s^\fAP{LPO}_0, N^\fAP{LPO})$ with the IP address
$I^\fAP{LPO} = \{\mapAddresstoAP(\fAP{LPO})\}$.  The initial state
$s^\LPO_0$ of $\LPO$ contains the private key of its domain, and the
signing keys of all IdPs ($\LPO$ does not need the public ssl keys of
other parties, which is why we omit them from $\LPO$'s initial
state.). The definition of $R^\LPO$ follows the description of $\LPO$
in Appendix~\ref{app:browserid-lowlevel}.

HTTP responses by $\LPO$ can contain strings representing scripts,
namely the script $\str{script\_lpo\_cif}$ run in the CIF and the
script $\str{script\_lpo\_ld}$ run in the LD. These scripts are
defined in Appendix~\ref{app:browserid-scripts-pidp}.

\subsubsection{Client sessions at $\LPO$.}  Any party can establish a
\emph{session} at $\LPO$. Such a session can either be authenticated
or unauthenticated. Roughly speaking, a session becomes authenticated
if a client has provided a valid CAP (for the origin of LPO) to $\LPO$
during the session. LPO manages groups of IDs, i.e., lists of email
addresses. If a user authenticates a session using any ID in the
group, she is authenticated for all IDs in the group. An authenticated
session can (non-deterministically) \emph{expire}, i.e. the
authenticated session can get unauthenticated or it is removed
completely. Such an expiration is used to model a user logout or a
session expiration caused by a timeout.

More specifically, a session is identified by a nonce, which is issued
by $\LPO$. Each session is associated with some xsrfToken, which is
also a nonce issued by $\LPO$. $\LPO$ stores all information about
established sessions in its state as a dictionary indexed by the
session identifier. In this dictionary, for every session $\LPO$
stores a pair containing the xsrfToken and, in authenticated sessions,
the sequence of all IDs associated with the secret provided in the
session, or, in unauthenticated sessions, the empty sequence $\an{}$
of IDs. On the receiver side (typically a browser) $\LPO$ places, by
appropriate headers in its HTTPS responses, a cookie named
$\str{browserid\_state}$ whose value is the session identifier (a
nonce). This cookie is flagged to be a session, httpOnly, and secure
cookie.

Before we provide a detailed formal specification of
$\LPO$, we first provide an informal description.

\subsubsection{HTTPSRequests to $\LPO$.}  $\LPO$ answers only to certain
requests (listed below). In reality, all such requests have to be over
HTTPS, and all responses send by $\LPO$ contain the $\str{Strict
  \mhyphen Transport \mhyphen Security}$ header. We overapproximate
safely here in omitting these two requirements from the model.

\begin{description}

\item[\texttt{GET /cif}.] $\LPO$ replies to this request by
  providing the script $\str{script\_lpo\_cif}$.

\item[\texttt{GET /ld}.] $\LPO$ replies to this request by
  providing the script $\str{script\_lpo\_ld}$.

\item[\texttt{GET /ctx}.] This requests the session context
  information from $\LPO$. The response body is of the form
  $\an{\mi{loggedIn}, \mi{xsrfToken}}$, where $\mi{loggedIn}$ is
  $\True$ or $\bot$, depending on whether the user is logged in at
  $\LPO$ or not, and $\mi{xsrfToken}$ is the token that the client is
  supposed to include into the auth request (see below).
\item[\texttt{POST /auth}.] With this request, a client can log into
  $\LPO$. The client has to provide a sequence of a CAP and an XSRF
  token. The CAP must be valid and issued for the origin of LPO.

\end{description}

We define $\LPO$ formally as an atomic DY process
$(I^\LPO, Z^\LPO, R^\LPO, s^\LPO_0,
N^\LPO)$. As already mentioned, we define $I^\LPO
= \{\mapAddresstoAP(\LPO)\}$.

In order to define the set $Z^\LPO$ of states of $\LPO$,
we first define the terms describing the session context of
a session. 

\begin{definition}
  A term of the form $\an{\mi{ids},\mi{xsrfToken}}$ with $\mi{ids}
  \subsetPairing \IDs$ and $\mi{xsrfToken} \in \nonces$ is called an
  \emph{LPO session context}. We denote the set of all LPO session
  contexts by $\LPOSessionCTXs$.
\end{definition}

Now, we define the set $Z^\LPO$ of states of LPO as
well as the initial state $s^\LPO_0$ of LPO.

\begin{definition}
  A \emph{state $s\in Z^\LPO$ of LPO} is a term of the form
  $\langle\mi{nonces}$, $\mi{sslkey}$, $\mi{signkeys}$,
  $\mi{sessions}\rangle$ where $\mi{nonces} \subsetPairing \nonces$
  (used nonces), $\mi{sslkey}=\mapSSLKey(\mapDomain(\LPO))$,
  $\mi{signkeys}$ is a mapping of domain names to public signing keys of the
  form $\mi{signkeys}=\an{\left\{\an{d, \pub(\mapSignKey(y))} \mid y \in
      \fAP{IdPs},\ d \in \mapDomain(y)\right\}}$, and $\mi{sessions} \in
  \dict{\nonces}{\LPOSessionCTXs}$.\footnote{As mentioned before, the
    state of LPO does not need to contain public keys.}

    The \emph{initial state $s^\LPO_0$ of LPO} is a state of LPO
    with $s^\LPO_0.\str{nonces} = \an{}$ and $s^\LPO_0.\str{sessions} =
    \an{}$.
\end{definition}

\begin{example} Let $k$ be a private signing key for some identity
  provider which owns the domain $\str{example.com}$. A possible state
  $s$ of LPO may look like this:
\begin{align*}
  s &=
  \an{\an{n_1,\ldots,n_m},\mapSSLKey(\mapDomain(\fAP{LPO})),[\str{example.com}:\pub(k)],\mi{sessions}} 
\end{align*}
with 
\begin{align*}
  \mi{sessions} &= \an{ \an{\str{sessionid_1},\an{\an{\mi{id}'_1, \ldots, \mi{id}'_l},\str{xsrfToken}}} , \ldots }
\end{align*}
\end{example}

We now specify the relation $R^\LPO \subseteq (\events \times Z^\LPO) \times
  (2^\events \times Z^\LPO)$ of LPO. Just like
in Appendix~\ref{sec:descr-web-brows}, we describe this
relation by a non-deterministic algorithm.

\captionof{algorithm}{\label{alg:lpo-pidp} Relation of LPO
  $R^\LPO$ }
\begin{algorithmic}[1]
\Statex[-1] \textbf{Input:} $(a{:}f{:}m),s$
  \Let{$s'$}{$s$}
  \Let{$\mi{sts}$}{$\an{\cSTS, \True}$}
  \If{$m \equiv \str{TRIGGER}$}
    \If{$s'.\str{sessions} \equiv \an{}$}
      \Stop{\DefStop}
    \EndIf
    \LetND{$\mi{sessionid}$}{$\{\mi{id} \mid \mi{id} \inPairing s'.\str{sessions}\}$}
    \LetND{$\mi{choice}$}{$\{\str{logout}, \str{expire}\}$}
    \If{$\mi{choice} \equiv \str{logout}$}
      \Let{$s'.\str{sessions}[\mi{sessionid}].\str{ids}$}{$\an{}$}
    \Else
      \Remove{$s'.\str{sessions}$}{$\mi{sessionid}$}
    \EndIf
    \Stop{\DefStop}
  \EndIf
  \LetST{$m_{\text{dec}}$, $k$}{$\an{m_{\text{dec}}, k} \equiv \dec{m}{s.\str{sslkey}}$\breakalgohook{0}}{\textbf{stop} \DefStop}
  \LetST{$n$, $\mi{method}$, $\mi{path}$, $\mi{params}$, $\mi{headers}$, $\mi{body}$}{\breakalgohook{0}$\an{\cHttpReq, n, \mi{method}, \mapDomain(\fAP{LPO}), \mi{path}, \mi{params}, \mi{headers}, \mi{body}} \equiv m_{\text{dec}}$\breakalgohook{0}}{\textbf{stop} \DefStop}
  \If{$\mi{method} \equiv \mGet \wedge \mi{path} \equiv \str{/cif}$} \Comment{Deliver CIF script}
    \Let{$\mi{scriptinit}$}{$\an{\str{init},\bot,\bot,\bot,\bot,\bot,\bot,\an{},\bot,\bot}$}
    \Let{$m'$}{$\encs{\an{\cHttpResp, n, 200, \an{\mi{sts}}, \an{\str{script\_lpo\_cif}, \mi{scriptinit}}}, k}$}
    \Stop{\StopWithMPrime}
  \ElsIf{$\mi{method} \equiv \mGet \wedge \mi{path} \equiv \str{/ld}$} \Comment{Deliver LD script}
    \Let{$\mi{scriptinit}$}{$\an{\str{init},\bot,\bot,\bot,\bot,\bot,\an{},\bot,\bot,\bot}$}
    \Let{$m'$}{$\encs{\an{\cHttpResp, n, 200, \an{\mi{sts}}, \an{\str{script\_lpo\_ld}, \mi{scriptinit}}}, k}$} 
    \Stop{\StopWithMPrime}
  \ElsIf{$\mi{method} \equiv \mGet \wedge \mi{path} \equiv \str{/ctx}$} \Comment{Deliver context info.}
    \Let{$\mi{sessionid}$}{$\mi{headers}[\str{Cookie}][\str{browserid\_state}]$}
    \If{$\mi{sessionid} \not\inPairing s.\str{sessions}$}
     \Let{$\mi{sessionid}$, $s'$}{$\mathsf{TAKENONCE}(s')$}
     \Let{$\mi{xsrfToken}$, $s'$}{$\mathsf{TAKENONCE}(s')$}
     \Append{$\an{\mi{sessionid}, \an{\an{}, \mi{xsrfToken}}}$}{$s'.\str{sessions}$}
    \EndIf
    \Let{$\mi{context}$}{$\an{\bot, s'.\str{sessions}[\mi{sessionid}].\str{xsrfToken}}$}
    \If{$s'.\str{session}[\mi{sessionid}].\str{ids} \not\equiv \an{}$}
      \Let{$\mi{context}.1$}{$\True$}
    \EndIf
    \Let{$\mi{setCookie}$}{$\an{\cSetCookie, \an{\an{\str{browserid\_state}, \mi{sessionid}, \True, \True, \True}}}$}
    \Let{$\mi{headers}$}{$\an{\mi{sts}, \mi{setCookie}}$}
    \Let{$m'$}{$\encs{\an{\cHttpResp, n, 200, \mi{headers}, \mi{context}}}{k}$}
    \Stop{\StopWithMPrime}
  \ElsIf{$\mi{method} \equiv \mPost \wedge \mi{path} \equiv \str{/auth}$}
    \LetST{$\mi{uc}$, $\mi{ia}$, $\mi{xsrfToken}$}{$\an{\an{\mi{uc}, \mi{ia}}, \mi{xsrfToken}} \equiv \mi{body}$\breakalgohook{1}}{\textbf{stop} \DefStop}
    \Let{$\mi{sessionid}$}{$\mi{headers}[\str{Cookie}][\str{browserid\_state}]$}
    \If{$s'.\str{sessions}[\mi{sessionid}].\str{xsrfToken} \not\equiv \mi{xsrfToken}$}
      \Stop{\DefStop}
    \EndIf
    \LetST{$\mi{name}$, $\mi{domain}$, $\mi{userpubkey}$}{\breakalgohook{1}$\an{\an{\mi{name}, \mi{domain}}, \mi{userpubkey}} \equiv \unsig{\mi{uc}}$\breakalgohook{1}}{\textbf{stop} \DefStop}
    \Let{$\mi{id}$}{$\an{\mi{name}, \mi{domain}}$}
    \Let{$\mi{origin}$}{$\unsig{\mi{ia}}$}
    \If{$\checksig{\mi{uc}}{s.\str{signkeys}[\mi{domain}]} \not\equiv \True \vee \checksig{\mi{ia}}{\mi{userpubkey}} \not\equiv \True$ \breakalgohook{1} $\vee$ $ \mi{origin} \not\equiv \an{s.\str{domain}, \https}$}
      \Stop{\DefStop}
    \EndIf
    \If{$s'.\str{sessions}[\mi{sessionid}].\str{ids} \equiv \an{}$}
      \If{$\not\exists\, n \in \mathbb{N}$ \textbf{such that} $\mi{id} \inPairing s'.\str{idgroups}.n$}
        \Append{$\an{\mi{id}}$}{$s'.\str{idgroups}$}
      \EndIf
      \LetND{$n$}{$\mathbb{N}$ \textbf{such that} $\mi{id} \inPairing s'.\str{idgroups}.n$}
    \Else
      \LetST{$n \gets \mathbb{N}$}{$s'.\str{idgroups}.n \equiv s'.\str{sessions}[\mi{sessionid}].\str{ids}$\breakalgohook{2}}{\textbf{stop} \DefStop}
      \If{$\mi{id} \not\inPairing s'.\str{idgroups}.n$}
        \Append{$\an{\mi{name}, \mi{domain}}$}{$s'.\str{idgroups}.n$}
      \EndIf
    \EndIf
    \Let{$s'.\str{sessions}[\mi{sessionid}].\str{ids}$}{$s'.\str{idgroups}.n$}
    \Let{$m'$}{$\encs{\an{\cHttpResp, n, 200, \an{\mi{sts}}, \True}}{k}$}
    \Stop{\StopWithMPrime}
  \EndIf
\Stop{\DefStop}
\end{algorithmic} \setlength{\parindent}{1em}
\subsection{Relying Parties} \label{app:relying-parties-pidp} 

A relying party $r \in \fAP{RP}$ is a web server modeled as an atomic
DY process $(I^r, Z^r, R^r, s^r_0, N^r)$ with the address $I^r :=
\{\mapAddresstoAP(r)\}$. Its initial state $s^r_0$ contains its
domain, the private key associated with its domain, the DNS server
address, and the signing keys of all IdPs.\footnote{We add the IdP
  verification keys to the initial status (instead of having RPs
  retrieve them dynamically from the IdP) in order to reduce the
  overall complexity.} The full state additionally contains the set of
service tokens the RP has issued. The definition of $R^r$ again
follows the description in Appendix~\ref{app:browserid-lowlevel}. RP
accepts only HTTPS requests.

In a typical flow with one client, $r$ will first receive an HTTP GET
request. In this case, it returns the script $\str{script\_rp\_index}$
(see Appendix~\ref{app:browserid-scripts-pidp} below) and sets the \texttt{Strict-Trans\-port-Security} header.

Afterwards, it will receive an HTTPS POST request. Provided that the
message contains a CAP,
$r$ checks that the UC and IA are valid and matching, and that the IA
contains the Origin of $r$ (with HTTPS). If the check is successful,
$r$ creates a new \emph{RP service token for the identity $i$},
$\an{n,i}$, and sends it to the browser. The RP keeps a list of such
tokens in its state. Intuitively, a client in possession of such a
token can use the service of $r$ for ID $i$ (e.g., access data of $i$
at $r$).

We now provide the formal definition of $r$ as an atomic DY process
$(I^r, Z^r, R^r, s^r_0, N^r)$. As mentioned, we define $I^r =
\{\mapAddresstoAP(r)\}$. Next, we define the set $Z^r$ of states of
$r$ and the initial state $s^r_0$ of $r$.

\begin{definition}\label{def:relying-parties}
  A \emph{state $s\in Z^r$ of an RP $r$} is a term of the form
  $\langle\mi{nonces}$, $\mi{domain}$, $\mi{sslkey}$, $\mi{signkeys}$,
  $\mi{serviceTokens}$, $\mi{corrupt}\rangle$ where $\mi{nonces} \subsetPairing
  \nonces$ (used nonces), $\mi{domain} = \mapDomain(r)$, $\mi{sslkey}=\mapSSLKey(\mapDomain(r))$,
  $\mi{signkeys}=\an{\left\{\an{d, \pub(\mapSignKey(y))} \mid y \in
      \fAP{IdPs},\ d \in \mapDomain(y)\right\}}$ (same as for $\fAP{LPO}$),
  $\mi{serviceTokens}\in\dict{\nonces}{\mathbb{S}}$,
  $\mi{corrupt}\in\terms$.

  The \emph{initial state $s^r_0$ of $r$} is a state of $r$ with
  $s^r_0.\str{nonces} = s^r_0.\str{serviceTokens} = \an{}$  and
  $s^r_0.\str{corrupt} = \bot$.
\end{definition}

We now specify the relation $R^r \subseteq (\events \times Z^r ) \times
  (2^\events \times Z^r)$ of $r$. Just like
in Appendix~\ref{sec:descr-web-brows}, we describe this
relation by a non-deterministic algorithm.  We note that we
use the function \textsf{TAKENONCE} introduced in
Section~\ref{app:proceduresbrowser} for this purpose.

\captionof{algorithm}{\label{alg:rp-pidp} Relation of a Relying
  Party $R^r$}
\begin{algorithmic}[1]
\Statex[-1] \textbf{Input:} $(a{:}f{:}m),s$
  \Let{$s'$}{$s$}
  \If{$s'.\str{corrupt} \not\equiv \bot \vee m \equiv \corrupt$}
    \Let{$s'.\str{corrupt}$}{$\an{\an{a, f, m}, s'.\str{corrupt}}$}
    \LetND{$m'$}{$d_{N^p}(s')$}
    \LetND{$a'$}{$\addresses$}
    \State \textbf{stop} $\{(a'{:}a{:}m')\}$, $s'$
  \EndIf
  \Let{$\mi{sts}$}{$\an{\cSTS, \True}$}
  \LetST{$m_{\text{dec}}$, $k$}{$\an{m_{\text{dec}}, k} \equiv \dec{m}{s.\str{sslkey}}$\breakalgohook{0}}{\textbf{stop} \DefStop}
  \LetST{$n$, $\mi{method}$, $\mi{path}$, $\mi{params}$, $\mi{headers}$, $\mi{body}$}{\breakalgohook{0}$\an{\cHttpReq, n, \mi{method}, s.\str{domain}, \mi{path}, \mi{params}, \mi{headers}, \mi{body}} \equiv m_{\text{dec}}$\breakalgohook{0}}{\textbf{stop} \DefStop}
  \If{$\mi{method} \equiv \mGet$} \Comment{Deliver CIF script}
    \Let{$\mi{scriptinit}$}{$\an{\str{init},\bot,\bot,\bot,\an{},\an{},\bot}$}
    \Let{$m'$}{$\encs{\an{\cHttpResp, n, 200, \an{\mi{sts}}, \an{\str{script\_rp\_index}, \mi{scriptinit}}}}{k}$}
    \Stop{\StopWithMPrime}
  \ElsIf{$\mi{method} \equiv \mPost \wedge \mi{headers} \equiv \an{\str{Origin}, \an{s.\str{domain},\https}}$} \label{line:rp-checksig-pidp}
    \LetST{$\mi{uc}$, $\mi{ia}$}{$\an{\mi{uc}, \mi{ia}} \equiv \mi{body}$}{\textbf{stop} \DefStop}
    \LetST{$\mi{name}$, $\mi{domain}$, $\mi{userpubkey}$}{\breakalgohook{1}$\an{\an{\mi{name}, \mi{domain}}, \mi{userpubkey}} \equiv \unsig{\mi{uc}}$\breakalgohook{1}}{\textbf{stop} \DefStop}
    \Let{$\mi{id}$}{$\an{\mi{name}, \mi{domain}}$}
    \Let{$\mi{origin}$}{$\unsig{\mi{ia}}$}
    \If{$\checksig{\mi{uc}}{s.\str{signkeys}[\mi{domain}]} \not\equiv \True \vee \checksig{\mi{ia}}{\mi{userpubkey}} \not\equiv \True \vee$\breakalgohook{1}$ \mi{origin} \not\equiv \an{s.\str{domain}, \https}$}
      \Stop{\DefStop}
    \EndIf
    \Let{$n_{\text{token}}$, $s'$}{$\mathsf{TAKENONCE}(s')$}\label{line:rp-create-nonce-pidp}
    \Append{$\an{n_{\text{token}},\mi{id}}$}{$s'.\str{serviceTokens}$}
    \Let{$m'$}{$\encs{\an{\cHttpResp, n, 200, \an{\mi{sts}}, \an{n_{\text{token}}, \mi{id}}}}{k}$}
    \Stop{\StopWithMPrime}
  \EndIf
  
\end{algorithmic} \setlength{\parindent}{1em}

\subsection{Identity Providers}  \label{app:idps}

An identity provider $i \in \mathsf{IdPs}$ is a web server modeled as
an atomic process $(I^i, Z^i, R^i, s_0^i, N^i)$ with the address $I^i
:= \{\mapAddresstoAP(i)\}$. Its initial state $s^i_0$ contains
a list of domains and (private) SSL keys (see below), a
list of users and identites (see below), and a private key for signing
UCs. Besides this, the full state of $i$ further contains a list of
used nonces, and information about active sessions.

Sessions are structured as a dictionary: For each session identifier
(session ID) the dictionary contains the list of identities for which
the session is authenticated.

IdPs, in our model, only accept SSL connections. Thus, after receiving
a request, an IdP first decodes the message. It then checks whether a
valid session ID is contained in the cookie that was sent with the
request. If there is no such ID, a new session with a freshly chosen
session ID is created. IdP saves this ID into its list of active
sessions, along with the initial session data (an empty list of
authenticated identities). A Set-Cookie header is added to IdPs
response to the browser in order to add the session cookie to the
client's cookie store.

The IdP then checks the method and the path of the request and acts as
follows:

If the method is GET, IdP serves, depending on the path, the
provisioning iframe ($\str{script\_idp\_pif}$) or the authentication
dialog ($\str{script\_idp\_ad}$) defined in Appendix~\ref{app:browserid-scripts-pidp}.

If the method is POST, the IdP can either authenticate the user or
sign a UC. In the first case, IdP extracts the identity of the user
(an email address) and the user's secret from the request. If the
secret and the identity are found in the user database, the session is
considered to be logged in for all identities associated with this
secret.  In the second case (signing UC), the IdP extracts the user's
identity and the public key of the user from the request. If the
session is considered to be logged in for this identity, the IdP
creates a UC and signs it with its signing key before sending it to
the user.

\subsubsection{Formal description.} In the following, we will first
define the (initial) state of $i$ formally and afterwards present the
definition of the relation $R^i$.

To define the initial state, we will need a term that represents the
``user database'' of the IdP $i$. We will call this term
$\mi{userset}^i$. This database defines, which secret is valid for
which set of identities.  It is encoded as a mapping of secrets to
lists of identities for which these secrets are valid. For example, if
the secret $\mi{secret}_1$ is valid for the identites $\mi{id}_1$ and
$\mi{id}_2$ and the secret $\mi{secret}_2$ is valid for the identities
$\mi{id}_3$ and $\mi{id}_4$, the $\mi{userset}^i$ looks as follows:
\begin{align*}
\mi{userset}^i = [\mi{secret}_1{:}\an{\mi{id}_1, \mi{id}_2}, \mi{secret}_2{:}\an{\mi{id}_3, \mi{id}_4}]
\end{align*}

To define $\mi{userset}^i$ (for the identity provider $i$), we first
define the set $\RPSecrets^i = \bigcup_{j \in \IDs^i}\mapIDtoPLI(j)$,
the function $\mathsf{IDsofSecret}: \RPSecrets \to \IDs$, $s \mapsto
\{j\,|\, j \in \IDs,\ \mapIDtoPLI(j) = s\}$, and finally
$\mi{userset}^i = \an{\{\an{s, \an{ \mathsf{IDsofSecret}(s) }}\, |\, s \in
\RPSecrets^i\}}$.

We also need a term that represents a dictionary that maps domains to
(private) SSL keys of the IdP $i$. We define $\mi{sslkeys}^i =
\an{\left\{\an{d, \mapSSLKey(d)} \mid d \in \mapDomain(i)\right\}}$.

\begin{definition}
  A \emph{state $s\in Z^i$ of an IdP $i$} is a term of the form
  $\langle\mi{nonces}$, $\mi{sslkeys}$, $\mi{users}$, $\mi{signkey}$,
  $\mi{sessions}$, $\mi{corrupt}\rangle$ where $\mi{nonces}
  \subsetPairing \nonces$ (used nonces), $\mi{sslkeys} =
  \mi{sslkeys}^i $, $\mi{users} = \mi{userset}^i$, $\mi{signkey} \in
  \nonces$ (the key used by the IdP $i$ to sign UCs),
  $\mi{sessions}\in\dict{\nonces}{\terms}$, $\mi{corrupt} \in \terms$.

  The \emph{initial state $s^i_0$ of $i$} is the state $\an{\an{},
    \mi{sslkeys}^i, \mi{userset}^i, \mapSignKey(i), \an{},
    \bot}$.
\end{definition}

The relation $R^i$ that defines the behavior of the IdP $i$ is defined as follows:

\captionof{algorithm}{\label{alg:idp-pidp} Relation of IdP $R^i$}
\begin{algorithmic}[1]
\Statex[-1] \textbf{Input:} $(a{:}f{:}m),s$
  \Let{$s'$}{$s$}
  \If{$s'.\str{corrupt} \not\equiv \bot \vee m \equiv \corrupt$}
    \Let{$s'.\str{corrupt}$}{$\an{\an{a, f, m}, s'.\str{corrupt}}$}
    \LetND{$m'$}{$d_{N^p}(s')$}\label{line:usage-of-signkey-corrupt-pidp}
    \LetND{$a'$}{$\addresses$}
    \State \textbf{stop} $\{(a'{:}a{:}m')\}$, $s'$
  \EndIf
  \Let{$\mi{sts}$}{$\an{\cSTS, \True}$}
  \LetST{$m_{\text{dec}}$, $k$, $k'$, $\mi{inDomain}$}{\breakalgohook{0}$\an{m_{\text{dec}}, k} \equiv \dec{m}{k'} \wedge \an{inDomain,k'} \in s.\str{sslkeys}$\breakalgohook{0}}{\textbf{stop} \DefStop}
  \LetST{$n$, $\mi{method}$, $\mi{path}$, $\mi{params}$, $\mi{headers}$, $\mi{body}$}{\breakalgohook{0}$\an{\cHttpReq, n, \mi{method}, \mi{inDomain}, \mi{path}, \mi{params}, \mi{headers}, \mi{body}} \equiv m_{\text{dec}}$\breakalgohook{0}}{\textbf{stop} \DefStop}
  
  \If{$\mi{method} \equiv \mPost$}
    \If{$\mi{path} \not\equiv \str{/certreq}$}  \Comment{User logs in.}
      \LetST{$\mi{id}$, $\mi{secret}$}{$\an{\mi{id}, \mi{secret}} \equiv \mi{body}$\breakalgohook{2}}{\textbf{stop} \DefStop}
      \If{$\mi{headers} \not\equiv \an{\str{Origin}, \an{\mi{inDomain}, \https}}$}
        \Stop{\DefStop}
      \EndIf
      
      \Let{$\mi{ids}$}{$s.\str{users}[\mi{secret}]$}
      \If{$\mi{ids} \equiv \an{} \vee \mi{id} \equiv \an{} \vee \mi{id} \not\inPairing \mi{ids}$} \Comment{Check id/secret pair.}
        \Stop{\DefStop}
      \EndIf
      \Let{$\mi{sessionid}$, $s'$}{$\mathsf{TAKENONCE}(s')$}
      \Let{$s'.\str{sessions}[\mi{sessionid}]$}{$\mi{ids}$}\label{line:populate-sessions-pidp}
      \Let{$\mi{setCookie}$}{$\an{\cSetCookie, \an{\an{\str{sessionid}, \mi{sessionid}, \True, \True, \True}}}$}
      \Let{$m'$}{$\encs{\an{\cHttpResp, n, 200, \an{\mi{sts,setCookie}}, \True}}{k}$}
      \Stop{\StopWithMPrime}\label{line:send-auth-response-pidp}
    \Else \Comment{User wants a certificate.}
      \LetST{$\mi{id}$, $\mi{pubkey}$}{$\an{\mi{id}, \mi{pubkey}} \equiv \mi{body}$\breakalgohook{2}}{\textbf{stop} \DefStop}      
      \Let{$\mi{sessionid}$}{$\mi{headers}[\str{Cookie}][\str{sessionid}]$}
      \If{$\mi{id} \not\inPairing s'.\str{sessions}[\mi{sessionid}]$} \Comment{Check if user is logged in.}
        \Stop{\DefStop}
      \EndIf
      \Let{$\mi{uc}$}{$\sig{\an{\mi{id}, \mi{pubkey}}}{s.\str{signkey}}$} \label{line:usage-of-signkey-pidp}
      \Let{$m'$}{$\encs{\an{\cHttpResp, n, 200, \an{\mi{sts}}, \mi{uc}}}{k}$}
      \Stop{\StopWithMPrime}
    \EndIf
  \Else
    \If{$\mi{path} \equiv \str{/pif}$}
      
      \Let{$m'$}{$\mathsf{enc}_\mathsf{s}(\langle\cHttpResp, n, 200, \an{\mi{sts}}, \langle\str{script\_idp\_pif}$,\breakalgohook{2}$\an{\str{init},\an{},\an{},\an{},\bot,\bot,\bot}\rangle\rangle, k)$}
    \Else
      \Let{$m'$}{$\encs{\an{\cHttpResp, n, 200, \an{\mi{sts}}, \an{\str{script\_idp\_ad, \an{}}}}}{k}$}
    \EndIf
    \Stop{\StopWithMPrime}
  \EndIf
  \Stop{\DefStop}
\end{algorithmic} \setlength{\parindent}{1em}

\subsection{BrowserID Scripts} \label{app:browserid-scripts-pidp} As
already mentioned in Section~\ref{app:outlinebrowserIDmodel}, the set
$\scriptset$ of the web system
$\bidwebsystem_\text{primary}=(\bidsystem, \scriptset,
\mathsf{script}, E_0)$ consists of the scripts $\Rasp$,
$\mi{script\_rp\_index}$, $\mi{script\_lpo\_cif}$,
$\mi{script\_lpo\_ld}$, $\mi{script\_idp\_pif}$, and
$\mi{script\_idp\_ad}$ with their string representations being
$\str{att\_script}$, $\str{script\_rp\_index}$,
$\str{script\_lpo\_cif}$, $\str{script\_lpo\_ld}$,
$\str{script\_idp\_pif}$, and $\str{script\_idp\_ad}$ (defined by
$\mathsf{script}$). 

The script $\Rasp$ is the attacker script (see
Section~\ref{sec:websystem}). The formal model of the other scripts
follows the description in Appendix~\ref{app:browserid-lowlevel}. The
script $\mi{script\_rp\_index}$ defines the script of the RP index
page. In reality, this page has its own script(s) and includes a
script from LPO. In our model, we combine both scripts into
$\mi{script\_rp\_index}$.  In particular, this script is responsible
for creating the CIF and the LD \iframes/subwindows, whose contents are loaded from LPO.

In what follows, the scripts $\mi{script\_rp\_index}$,
$\mi{script\_lpo\_cif}$, and $\mi{script\_lpo\_ld}$ are
defined formally. First, we introduce some notation and
helper functions.

\subsubsection{Notations and Helper Functions.}
In the formal description of the scripts we use an abbreviation for
URLs at LPO. We write $\mathsf{URL}^\mathsf{LPO}_\mi{path}$ to
describe the following URL term: $\an{\tUrl, \https, \domLPO,
  \mi{path}, \an{}}$.  Also, we call $\mathsf{origin}_\LPO$ the origin
of LPO which describes the following origin term:
$\an{\domLPO,\https}$.

In order to simplify the description of the scripts, several helper functions are used.

\paragraph{CHOOSEINPUT.}  As explained in
Section~\ref{sec:web-browsers}, the state of a document
contains a term, say, $\mi{scriptinputs}$, which records the
input this document has obtained so far (via \xhrs and
\pms). If the script of the document is activated, it will
typically need to pick one input message from
$\mi{scriptinputs}$ and record which input it has already
processed. For this purpose, the function
$\mathsf{CHOOSEINPUT}(s',\mi{scriptinputs})$ is used, where
$s'$ denotes the scripts current state. It saves the
indexes of already handled messages in the scriptstate $s'$
and chooses a yet unhandled input message from
$\mi{scriptinputs}$. The index of this message is then
saved in the scriptstate (which is returned to the script).

\captionof{algorithm}{\label{alg:chooseinput} Choose an unhandled input message for a script}
\begin{algorithmic}[1]
  \Function{$\mathsf{CHOOSEINPUT}$}{$s',\mi{scriptinputs}$}
  \LetST{$\mi{iid}$}{$\mi{iid} \in
    \{1,\cdots,|\mi{scriptinputs}|\} \wedge \mi{iid} \not\inPairing
    s'.\str{handledInputs}$\breakalgohook{1}}{\Return$(\bot,s')$}
  \Let{$\mi{input}$}{$\proj{\mi{iid}}{\mi{scriptinputs}}$}
  \Let{$s'.\str{handledInputs}$}{$s'.\str{handledInputs} \plusPairing
    \mi{iid}$}

  \State \Return$(\mi{input}, s')$
  \EndFunction
\end{algorithmic} \setlength{\parindent}{1em}

\paragraph{PARENTWINDOW.} To determine the nonce
referencing the parent window in the browser, the function
$\mathsf{PARENTWINDOW}(\mi{tree}, \mi{docnonce})$ is
used. It takes the term $\mi{tree}$, which is the (partly
cleaned) tree of browser windows the script is able to see
and the document nonce $\mi{docnonce}$, which is the nonce
referencing the current document the script is running in,
as input. It outputs the nonce referencing the window which
directly contains in its subwindows the window of the
document referenced by $\mi{docnonce}$. If there is no such
window (which is the case if the script runs in a document
of a top-level window), $\mathsf{PARENTWINDOW}$ returns
$\bot$.

\paragraph{SUBWINDOWS.} This function takes a term
$\mi{tree}$ and a document nonce $\mi{docnonce}$ as input
just as the function above. If $\mi{docnonce}$ is not a
reference to a document contained in $\mi{tree}$, then
$\mathsf{SUBWINDOWS}(\mi{tree},\mi{docnonce})$ returns
$\an{}$. Otherwise, let $\an{\mi{docnonce}$, $\mi{origin}$,
  $\mi{script}$, $\mi{scriptstate}$, $\mi{scriptinput}$,
  $\mi{subwindows}$, $\mi{active}}$ denote the subterm of
$\mi{tree}$ corresponding to the document referred to by
$\mi{docnonce}$. Then,
$\mathsf{SUBWINDOWS}(\mi{tree},\mi{docnonce})$ returns
$\mi{subwindows}$.

\paragraph{AUXWINDOW.} This function takes a term
$\mi{tree}$ and a document nonce $\mi{docnonce}$ as input
as above. From all window terms in $\mi{tree}$ that have
the window containing the document identified by
$\mi{docnonce}$ as their opener, it selects one
non-deterministically and returns its nonce. If there is no such
window, it returns the nonce of the window containing
$\mi{docnonce}$.

\paragraph{OPENERWINDOW.} This function takes a
term $\mi{tree}$ and a document nonce $\mi{docnonce}$ as
input as above. It returns the window nonce of the opener
window of the window that contains the document identified
by $\mi{docnonce}$. Recall that the nonce identifying the
opener of each window is stored inside the window term. If
no document with nonce $\mi{docnonce}$ is found in the tree
$\mi{tree}$, $\notdef$ is returned.

\paragraph{GETWINDOW.} This function takes a term
$\mi{tree}$ and a document nonce $\mi{docnonce}$ as input
as above. It returns the nonce of the window containing $\mi{docnonce}$.

\paragraph{GETORIGIN.} The function
$\mathsf{GETORIGIN}(\mi{tree},\mi{docnonce})$ extracts the origin of a
document. It searches for the document with the identifier
$\mi{docnonce}$ in the (cleaned) tree $\mi{tree}$ of the
browser's windows and documents. It returns the origin $o$
of the document. If no document with nonce $\mi{docnonce}$
is found in the tree $\mi{tree}$, $\notdef$ is returned.

\subsubsection{Web storage under LPO's origin.}\label{app:webst-under-lpos}
The web storage under the origin of LPO used by the scripts
$\mi{script\_lpo\_cif}$ and $\mi{script\_lpo\_ld}$ (see below) is organized as
follows: 

The \ls is a dictionary.  There are two types of entries in this
dictionary: Under the key $\str{siteInfo}$, a dictionary is stored
which has origins as keys and IDs as values. An entry in this
dictionary indicates that the user is logged in at the referenced
origin with a certain ID.  The second type of entry has a nonce as a
key. The value is an email address (ID). This models the email address
a user entered in the LD before being navigated away to the AD. The
nonce is also stored in the sessionStorage (see below).

\begin{example} \label{ex:lpo-localstorage-pidp}
  \[
  \begin{array}{rll}
    \langle\hspace{-0.61ex}\vphantom 
        & \langle\str{siteInfo},\langle\hspace{-0.61ex}\vphantom
                                &\an{ \an{\str{domain_{RP1}},\str{S}},\mi{id_1} },\\
                                &&\an{ \an{\str{domain_{RP2}},\str{S}},\mi{id_1} },\\
                                &&\an{ \an{\str{domain_{RP3}},\str{S}},\mi{id_2} }
                              \rangle
           \rangle,\\
       & \an{n_1, \mi{id}_1}, \\
       & \an{n_2, \mi{id}_3} \rangle
  \end{array} 
  \]

  This example shows a \ls under the origin of $\LPO$, indicating that
  the user is logged in at $\str{domain_{RP1}}$ and
  $\str{domain_{RP2}}$ with $\mi{id_1}$ and at $\str{domain_{RP3}}$
  with $\mi{id_2}$ (using HTTPS). Further, the nonces $n_1$ and $n_2$
  each refer to an email address which the user entered in the LD.
\end{example}

The sessionStorage is also a dictionary. It may only contain one key,
$\str{idpnonce}$. Its value is a nonce (like $n_1$ or $n_2$ in the
example above) which references the latest email address entry in the
localStorage (see above).

\subsubsection{login.persona.org Communication Iframe Script (script\_lpo\_cif).}\label{app:scriptLPOcif-pidp}
As defined in Section~\ref{sec:websystem}, a script is a
relation that takes as input a term and a set of nonces it
may use. It outputs a new term. As specified in
Section~\ref{sec:web-browsers} (Triggering the Script of a
Document (\textbf{\hlExp{$m = \trigger$},
  \hlExp{$\mi{action} = 1$}})) and formally specified in
Algorithm~\ref{alg:runscript}, the input term is provided
by the browser. It contains the current internal state of
the script (which we call \emph{scriptstate} in what
follows) and additional information containing all browser
state information the script has access to, such as the
input the script has obtained so far via \xhrs and \pms,
information about windows, etc. The browser expects the
output term to have a specific form, as also specified in
Section~\ref{sec:web-browsers} and
Algorithm~\ref{alg:runscript}. The output term contains,
among other information, the new internal scriptstate.

As for $\mi{script\_lpo\_cif}$, this script models the
script run in the CIF, as sketched in
Appendix~\ref{app:browserid-lowlevel}.

We first describe the structure of the internal scriptstate
of the script $\mi{script\_lpo\_cif}$.

\begin{definition} \label{def:scriptstatelpocif-pidp} A \emph{scriptstate
  $s$ of $\mi{script\_lpo\_cif}$} is a term of the form $\langle q$,
  $\mi{requestOrigin}$, $\mi{loggedInUser}$, $\mi{pause}$,
  $\mi{context}$, $\mi{key}$, $\mi{uc}$, $\mi{handledInputs}$,
  $\mi{refXHRctx}$, $\mi{PIFindex} \rangle$ where $q \in \mathbb{S}$,
  $\mi{requestOrigin} \in \origins \cup \{\bot\}$, $\mi{loggedInUser}
  \in \IDs \cup \{\an{},\bot\}$, $\mi{pause} \in \{\True,\bot\}$,
  $\mi{context} \in \terms$, $\mi{key} \in \nonces \cup \{\bot\}$,
  $\mi{uc} \in \terms$, $\mi{handledInputs} \subsetPairing
  \mathbb{N}$, $\mi{refXHRctx} \in \nonces \cup \{\bot\}$,
  $\mi{PIFindex} \in \mathbb{N} \cup \{\bot\}$.  The \emph{initial
    scriptstate $\mi{initState_{cif}}$} of $\mi{script\_lpo\_cif}$ is
  the state
  $\an{\str{init},\bot,\bot,\bot,\bot,\bot,\bot,\an{},\bot,\bot}$.
\end{definition}

Before we provide the formal specification of the relation
that defines the behavior of $\mi{script\_lpo\_cif}$, we
present an informal description. The behavior mainly
depends on the state $q$ the script is in.

\begin{description}
\item[$q=\str{init}$] This is the initial state. Its only
  transition handles no input and outputs a \pm
  \texttt{cifready} to its parent window and transitions to
  $\str{default}$.

\item[$q=\str{default}$] This is the state to which
  $\str{script\_lpo\_cif}$ always returns to. This state handles all
  \pms the CIF expects to receive from its parent window. If the \pm
  received was sent from the parent window of the CIF, it behaves as
  follows, depending on the first element of the received postMessage:

  \begin{description}

  \item[\pm \texttt{loaded}] The script records the sender's origin of the
    received \pm as the remote origin in the scriptstate if the
    scriptstate did not contain any information about the remote
    origin yet. Also, an ID, which represents the assumption of the
    sender on who it believes to be logged in, is saved in the
    scriptstate. If the $\mi{pause}$ flag in the scriptstate is
    $\True$ it transitions to the state $\str{default}$. Otherwise, it
    is checked, if the current context in the scriptstate is
    $\bot$. If the check is true, the script transitions to the state
    $\str{fetchContext}$, or to the state $\str{checkAndEmit}$
    otherwise.

  \item[\pm \texttt{dlgRun}] The script sets the $\mi{pause}$ flag in the
    scriptstate to $\True$ and transitions to
    $\str{default}$.

  \item[\pm \texttt{dlgCmplt}] The script sets the $\mi{pause}$ flag in the
    scriptstate to $\bot$.  It then transitions to the state
    $\str{fetchContext}$.

  \item[\pm \texttt{loggedInUser}] This message has to contain an
    ID. This ID is saved in the scriptstate and then the
    script transitions to $\str{default}$.

  \item[\pm \texttt{logout}] The script removes the entry for the RP (recorded in
    the scriptstate) from the \ls and then transitions to the state
    $\str{sendLogout}$. If no remote origin is set in the script's
    state, it is now set to the sender's origin of the received \pm.
  \end{description}

\item[$q=\str{fetchContext}$] In this state, the script sends an \xhr to LPO with a
  $\mGet$ request to the path \texttt{/ctx} and then
  transitions to the state $\str{receiveContext}$.

\item[$q=\str{receiveContext}$] In this state, the script expects an \xhr response as
  input containing the session context. This context is
  saved as the current context in the scriptstate. The
  script transitions to $\str{checkAndEmit}$.

\item[$q=\str{checkAndEmit}$] This state lets the script create the provisioning
  iframe and transition to $\str{startPIF}$ iff (1) some email address
  is marked as logged in at RP in the \ls, (2) if an email address is
  recorded in the current scriptstate, this email address differs from
  the one recorded in the \ls, and (3) the user is marked as logged in
  in the current context. Otherwise, if the email address recorded in
  the current scriptstate is $\an{}$, the script transitions to
  $\str{default}$, else it transitions to $\str{sendLogout}$.

\item[$q=\str{startPIF}$] In this state, the script waits for a \pm
  from the PIF containing a \texttt{ping} message. If such a message
  is received and the sender's window and origin match the PIF, the
  script sends a \texttt{pong} message back to the PIF and transitions
  to the state $\str{runPIF}$.

\item[$q=\str{runPIF}$] This is the state in which
  $\str{script\_lpo\_cif}$ interacts with the PIF. This state handles
  all \pms the CIF expects to receive from the latest PIF (as recorded
  in $\str{PIFindex}$ in its state). If the \pm received was sent from
  the PIF's window and the PIF's origin, it behaves as follows,
  depending on the first element of the postMessage:

  \begin{description}
  \item[\pm \texttt{beginProvisioning}] The script responds with a \pm
    to the PIF containing the email address of the identity which is
    to authenticate to the relying party (as recorded in the CIF's
    state).
  \item[\pm \texttt{genKeyPair}] The script creates a fresh key pair (i.e. the
    CIF chooses a fresh nonce) and sends the public key contained in
    a \pm to the PIF.
  \item[\pm \texttt{registerCertificate}] The script stores the UC received in
    this \pm in the CIF's state and transitions to the state
    $\str{createCAPforRP}$.
  \item[\pm \texttt{raiseProvisioningFailure}] This message indicates that no one
    is logged in. This is recorded in the CIF's state accordingly. The
    script transitions to the state $\str{sendLogout}$ in which the
    CIF's parent window will be notified that no one is logged in.
  \end{description}

\item[$q=\str{createCAPforRP}$] In this state, the script creates an
  IA for the request origin (as recorded in the script's state),
  combines the IA with the UC to a CAP, and sends the CAP in a
  postMessage to its parent restricting the receiver to the request
  origin.

\item[$q=\str{sendLogout}$] In this state, the script sends a \texttt{logout} \pm to
  the parent document and then  transitions to
  the $\str{default}$ state.

\end{description}

We now specify the relation $\mi{script\_lpo\_cif}
\subseteq (\terms \times 2^{\nonces})\times \terms$ of the
CIF's scripting process formally. Just like in
Appendix~\ref{sec:descr-web-brows}, we describe this
relation by a non-deterministic algorithm. 

Just like all scripts, as explained in
Section~\ref{sec:web-browsers} (see also
Algorithm~\ref{alg:runscript} for the formal
specification), the input term this script obtains from the
browser contains the cleaned tree of the browser's windows
and documents $\mi{tree}$, the nonce of the current
document $\mi{docnonce}$, its own scriptstate
$\mi{scriptstate}$ (as defined in
Definition~\ref{def:scriptstatelpocif-pidp}), a sequence of all
inputs $\mi{scriptinput}$ (also containing already handled
inputs), a dictionary $\mi{cookies}$ of all accessible
cookies of the document's domain, the \ls
$\mi{localStorage}$ belonging to the document's origin, the
secrets $\mi{secret}$ of the document's origin, and a set
$\mi{nonces}$ of fresh nonces as input. The script returns
a new scriptstate $s'$, a new set of cookies
$\mi{cookies'}$, a new \ls $\mi{localStorage'}$, and a term
$\mi{command}$ denoting a command to the browser.

\captionof{algorithm}{\label{alg:scriptlpocif-pidp} Relation of $\mi{script\_lpo\_cif}$ }
\begin{algorithmic}[1]
\Statex[-1] \textbf{Input:} $\langle\mi{tree}$, $\mi{docnonce}$, $\mi{scriptstate}$, $\mi{scriptinputs}$, $\mi{cookies}$, $\mi{localStorage}$, $\mi{sessionStorage}$,\breakalgohook{-1}$\mi{ids}$, $\mi{secret}\rangle$, $\mi{nonces}$ 
\Let{$s'$}{$\mi{scriptstate}$}
\Let{$\mi{cookies}'$}{$\mi{cookies}$}
\Let{$\mi{localStorage}'$}{$\mi{localStorage}$}
\Let{$\mi{sessionStorage}'$}{$\mi{sessionStorage}$}

\Switch{$s'.\str{q}$}
 \Case{$\str{init}$}
  \Let{$\mi{command}$}{$\langle\tPostMessage$, $\textsf{PARENTWINDOW}(\mi{tree}, \mi{docnonce})$,\breakalgohook{2}$\an{\str{cifready},\an{}}$, $\bot\rangle$}
  \Let{$s'.\str{q}$}{$\str{default}$}
  \State \textbf{stop} $\an{s',\mi{cookies}',\mi{localStorage}',\mi{sessionStorage}',\mi{command}}$
 \EndCase

 \Case{$\str{default}$}
  \Let{$\mi{input},s'$}{\textsf{CHOOSEINPUT}($s',\mi{scriptinputs}$)}
  \If{$\proj{1}{\mi{input}} \equiv \tPostMessage$}
   \Let{$\mi{senderWindow}$}{$\proj{2}{\mi{input}}$}
   \Let{$\mi{senderOrigin}$}{$\proj{3}{\mi{input}}$}
   \Let{$m$}{$\proj{4}{\mi{input}}$}%

   \If{$\mi{senderWindow} \equiv \textsf{PARENTWINDOW}(\mi{tree}, \mi{docnonce})$}
   \Switch{$m$}
    \Case{$\an{\str{loaded},\mi{id}}$}
      \If{$s'.\str{requestOrigin} \equiv \bot$}
        \Let{$s'.\str{requestOrigin}$}{$\mi{senderOrigin}$} \label{line:set-parent-origin-pidp}
      \EndIf
      \Let{$s'.\str{loggedInUser}$}{$\mi{id}$}
      \If{$s'.\str{pause} \equiv \True$}
       \State \textbf{stop} $\an{s',\mi{cookies}',\mi{localStorage}',\mi{sessionStorage}',\an{}}$
      \ElsIf{$s'.\str{context} \equiv \bot$}
       \Let{$s'.\str{q}$}{$\str{fetchContext}$}
       \State \textbf{stop} $\an{s',\mi{cookies}',\mi{localStorage}',\mi{sessionStorage}',\an{}}$
      \Else
       \Let{$s'.\str{q}$}{$\str{checkAndEmit}$}
       \State \textbf{stop} $\an{s',\mi{cookies}',\mi{localStorage}',\mi{sessionStorage}',\an{}}$
      \EndIf
    \EndCase

    \Case{$\an{\str{dlgRun},\an{}}$}
      \Let{$s'.\str{pause}$}{$\True$}
      \State \textbf{stop} $\an{s',\mi{cookies}',\mi{localStorage}',\mi{sessionStorage}',\an{}}$ 
    \EndCase

    \Case{$\an{\str{dlgCmplt},\an{}}$}
      \Let{$s'.\str{pause}$}{$\bot$}
      \Let{$s'.\str{q}$}{$\str{fetchContext}$}
      \State \textbf{stop} $\an{s',\mi{cookies}',\mi{localStorage}',\mi{sessionStorage}',\an{}}$ 
    \EndCase

    \Case{$\an{\str{loggedInUser},\mi{id}}$}
      \Let{$s'.\str{loggedInUser}$}{$\mi{id}$}
      \State \textbf{stop} $\an{s',\mi{cookies}',\mi{localStorage}',\mi{sessionStorage}',\an{}}$ 
   \EndCase

    \Case{$\an{\str{logout},\an{}}$}
      \If{$s'.\str{requestOrigin} \equiv \bot$}
        \Let{$s'.\str{requestOrigin}$}{$\mi{senderOrigin}$}
      \EndIf
      \Let{$s'.\str{loggedInUser}$}{$\bot$}
      \State \textbf{remove} the element with key $s'.\str{requestOrigin}$\breakalgohook{6}from the dictionary $\mi{localStorage}'[\str{siteInfo}]$ \label{line:remove-siteinfo-cif-pidp}
      \Let{$s'.\str{q}$}{$\str{sendLogout}$}
    \EndCase
   \EndSwitch
   \EndIf
  \EndIf
 \EndCase

 \Case{$\str{fetchContext}$}
  \LetND{$s'.\str{refXHRctx}$}{$\mi{nonces}$}
  \Let{$\mi{command}$}{$\an{\tXMLHTTPRequest,\textsf{URL}^\LPO_\str{/ctx},\mGet,\an{},s'.\str{refXHRctx}}$} \label{line:ctx-over-https-pidp}
  \Let{$s'.\str{q}$}{$\str{receiveContext}$}
  \State \textbf{stop} $\an{s',\mi{cookies}',\mi{localStorage}',\mi{sessionStorage}',\mi{command}}$
 \EndCase

 \Case{$\str{receiveContext}$}
  \Let{$\mi{input},s'$}{\textsf{CHOOSEINPUT}($s',\mi{scriptinputs}$)}
  \If{$(\proj{1}{\mi{input}} \equiv \tXMLHTTPRequest) \wedge (\proj{3}{\mi{input}} \equiv s'.\str{refXHRctx})$}
   \Let{$s'.\str{context}$}{$\proj{2}{\mi{input}}$} \label{line:store-context-pidp}
   \Let{$s'.\str{q}$}{$\str{checkAndEmit}$}
   \State \textbf{stop} $\an{s',\mi{cookies}',\mi{localStorage}',\mi{sessionStorage}',\an{}}$
  \EndIf

 \EndCase

 \Case{$\str{checkAndEmit}$}
  \Let{$s'.\str{email}$}{$\mi{localStorage}'[\str{siteInfo}][s'.\str{requestOrigin}]$} \label{line:set-semail-cif-pidp}
  \If{$(s'.\str{email} \not\equiv \an{}) $\breakalgohook{2}$\wedge (s'.\str{loggedInUser} \notin \{\an{},\bot\} \Rightarrow (s'.\str{loggedInUser} \not\equiv s'.\str{email})) $\breakalgohook{2}$\wedge (\proj{1}{s'.\str{context}} \equiv \True)$}
   \Let{$s'.\str{q}$}{$\str{startPIF}$}
   \Let{$\mi{url}$}{$\an{\cUrl, \https, \pi_2(s'.\str{email}), \str{/pif}}$}
   \Let{$s'.\str{PIFindex}$}{$|\mi{subwindows}|+1$} \breakalgo{1}\Comment{Index of the next subwindow to be created.} 
   \Let{$\mi{command}$}{$\an{\tIframe,\mi{url},\wSelf}$}
   \State \textbf{stop} $\an{s',\mi{cookies}',\mi{localStorage}',\mi{sessionStorage}',\mi{command}}$
  \ElsIf{$s'.\str{loggedInUser} \equiv \an{}$}
   \Let{$s'.\str{q}$}{$\str{default}$}
   \State \textbf{stop} $\an{s',\mi{cookies}',\mi{localStorage}',\mi{sessionStorage}',\an{}}$
  \Else
   \Let{$s'.\str{q}$}{$\str{sendLogout}$}
   \State \textbf{stop} $\an{s',\mi{cookies}',\mi{localStorage}',\mi{sessionStorage}',\an{}}$
  \EndIf
 \EndCase

 \Case{$\str{startPIF}$}
  \Let{$\mi{idpOrigin}$}{$\an{\pi_2(s'.\str{email}), \https}$}
  \Let{$\mi{input},s'$}{\textsf{CHOOSEINPUT}($s',\mi{scriptinputs}$)}
  \Let{$\mi{pifNonce}$}{$\proj{s'.\str{PIFindex}}{\mi{subwindows}}.\str{nonce}$}
  \If{$\proj{1}{\mi{input}} \equiv \tPostMessage$}
   \Let{$\mi{senderWindow}$}{$\proj{2}{\mi{input}}$}
   \Let{$\mi{senderOrigin}$}{$\proj{3}{\mi{input}}$}
   \Let{$m$}{$\proj{4}{\mi{input}}$}

   \If{$m \equiv \str{ping} \wedge \mi{senderWindow} \equiv \mi{pifNonce} $\breakalgohook{3}$\wedge \mi{senderOrigin} \equiv \mi{idpOrigin}$}
     \Let{$\mi{command}$}{$\an{\tPostMessage,\mi{pifNonce},\str{pong},\mi{idpOrigin}}$}
     \Let{$s'.\str{q}$}{$\str{runPIF}$}
     \State \textbf{stop} $\an{s',\mi{cookies}',\mi{localStorage}',\mi{sessionStorage}',\mi{command}}$
   \EndIf
  \EndIf
 \EndCase

 \Case{$\str{runPIF}$}
  \Let{$\mi{idpOrigin}$}{$\an{\pi_2(s'.\str{email}), \https}$}
  \Let{$\mi{input},s'$}{\textsf{CHOOSEINPUT}($s',\mi{scriptinputs}$)}
  \Let{$\mi{pifNonce}$}{$\proj{s'.\str{PIFindex}}{\mi{subwindows}}.\str{nonce}$}
  \If{$\proj{1}{\mi{input}} \equiv \tPostMessage$}
   \Let{$\mi{senderWindow}$}{$\proj{2}{\mi{input}}$}
   \Let{$\mi{senderOrigin}$}{$\proj{3}{\mi{input}}$}
   \Let{$m$}{$\proj{4}{\mi{input}}$}

   \If{$\mi{senderWindow} \equiv \mi{pifNonce} \wedge \mi{senderOrigin} \equiv \mi{idpOrigin}$}
     \Switch{$\proj{1}{m}$}
       \Case{$\str{beginProvisioning}$}
         \Let{$\mi{jschannel\_nonce}$}{$\proj{2}{m}$}
         \Let{$\mi{command}$}{$\langle\tPostMessage$, $\mi{pifNonce}$,\breakalgohook{6}$\an{jschannel\_nonce, s'.\str{email}}$, $\mi{idpOrigin}\rangle$}
         \State \textbf{stop} $\an{s',\mi{cookies}',\mi{localStorage}',\mi{sessionStorage}',\mi{command}}$ \label{line:send-pm-with-nonce-cif-email-pidp}
       \EndCase
       \Case{$\str{genKeyPair}$}
         \Let{$\mi{jschannel\_nonce}$}{$\proj{2}{m}$}
         \LetND{$s'.\str{key}$}{$\mi{nonces}$}
         \Let{$\mi{command}$}{$\langle\tPostMessage$, $\mi{pifNonce}$,\breakalgohook{6} $\an{jschannel\_nonce, \pub(s'.\str{key})}$, $\mi{idpOrigin}\rangle$}
         \State \textbf{stop} $\an{s',\mi{cookies}',\mi{localStorage}',\mi{sessionStorage}',\mi{command}}$ \label{line:send-pm-with-nonce-cif-key-pidp}
       \EndCase
       \Case{$\str{registerCertificate}$}
         \If{$\proj{1}{\unsig{\proj{2}{m}}} \equiv s'.\str{email} \wedge s'.\str{email} \not\equiv \an{}$} \breakalgo{1}\Comment{This check is our fix against identity injection.}
           \Let{$s'.\str{uc}$}{$\proj{2}{m}$} \label{line:set-uc-cif-pidp}
           \Let{$s'.\str{q}$}{$\str{createCAPforRP}$}
         \EndIf
         \State \textbf{stop} $\an{s',\mi{cookies}',\mi{localStorage}',\mi{sessionStorage}',\mi{command}}$
       \EndCase
       \Case{$\str{raiseProvisioningFailure}$}
         \Let{$s'.\str{loggedInUser}$}{$\bot$}
         \Let{$s'.\str{q}$}{$\str{sendLogout}$}
         \State \textbf{stop} $\an{s',\mi{cookies}',\mi{localStorage}',\mi{sessionStorage}',\mi{command}}$
       \EndCase
     \EndSwitch
   \EndIf
  \EndIf
 \EndCase

 \Case{$\str{createCAPforRP}$}
   \Let{$\mi{ia}$}{$\sig{s'.\str{requestOrigin}}{s'.\str{key}}$}
   \Let{$\mi{cap}$}{$\an{s'.\str{uc},\mi{ia}}$} \label{line:use-uc-cif-pidp}
   \Let{$\mi{command}$}{$\langle\tPostMessage$, $\textsf{PARENTWINDOW}(\mi{tree}, \mi{docnonce})$, \breakalgohook{2}$\an{\str{response},\mi{cap}}$, $s'.\str{requestOrigin}\rangle$} %
   \Let{$s'.\str{q}$}{$\str{null}$}
   \State \textbf{stop} $\an{s',\mi{cookies}',\mi{localStorage}',\mi{sessionStorage}',\mi{command}}$ %
 \EndCase

 \Case{$\str{sendLogout}$}
  \Let{$\mi{command}$}{$\langle\tPostMessage$, $\textsf{PARENTWINDOW}(\mi{tree}, \mi{docnonce})$, \breakalgohook{2}$\an{\str{logout},\an{}}$, $\bot\rangle$}
  \Let{$s'.\str{q}$}{$\str{default}$}
  \State \textbf{stop} $\an{s',\mi{cookies}',\mi{localStorage}',\mi{sessionStorage}',\mi{command}}$
 \EndCase

\EndSwitch

\State \textbf{stop} $\an{\mi{scriptstate},\mi{cookies},\mi{localStorage},\mi{sessionStorage},\an{}}$

\end{algorithmic} \setlength{\parindent}{1em}

\subsubsection{login.persona.org Login Dialog Script (script\_lpo\_ld).}\label{app:scriptLPOld-pidp}
This script models the LD contents. Its formal specification, presented next,
follows the one presented above for $\mi{script\_lpo\_cif}$.

\begin{definition}\label{def:scriptstatelpold-pidp}
  A \emph{scriptstate $s$ of $\mi{script\_lpo\_ld}$} is a term of the
  form $\langle q$, $\mi{requestOrigin}$, $\mi{context}$,
  $\mi{email}$, $\mi{key}$, $\mi{uc}$, $\mi{handledInputs}$,
  $\mi{refXHRctx}$, $\mi{refXHRLPOauth}$, $\mi{PIFindex} \rangle$ with
  $q \in \mathbb{S}$, $\mi{requestOrigin} \in \origins \cup \{\bot\}$,
  $\mi{context} \in \terms$, $\mi{email} \in \IDs \cup \{\bot\}$,
  $\mi{key} \in \nonces \cup \{\bot\}$, $\mi{uc} \in \terms$,
  $\mi{handledInputs} \subsetPairing \mathbb{N}$,
  $\mi{refXHRctx},\mi{refXHRLPOauth} \in \nonces \cup \{\bot\}$,
  $\mi{PIFindex} \in \mathbb{N} \cup \{\bot\}$.
  The \emph{initial scriptstate $\mi{initState_{ld}}$} is the state
  $\an{\str{init},\bot,\bot,\bot,\bot,\bot,\an{},\bot,\bot,\bot}$.
\end{definition}

Before we provide the formal specification of the relation
that defines the behavior of $\mi{script\_lpo\_ld}$, we
present an informal description. The behavior mainly
depends on the state $q$ the script is in.

\begin{description}
\item[$q\equiv\str{init}$] This is the initial state. Its only
  transition takes no input and outputs a \pm \texttt{ldready} to its
  parent window and transitions to $\str{start}$.

\item[$q\equiv\str{start}$] In this state, the script expects a
  \texttt{request} \pm. The sender's origin of this \pm is recorded as
  the requesting origin in the scriptstate. An \xhr is sent to $\LPO$
  with a $\mGet$ request to the path \texttt{/ctx} and then the script
  transitions to the state $\str{receiveContext}$.

\item[$q\equiv\str{receiveContext}$] In this state, the script expects
  an \xhr response as input containing the session context. This
  context is saved as the current context in the scriptstate. The
  script checks if an \texttt{idpNonce} is recorded in the
  sessionStorage. The presence of this nonce indicates that there was
  a run of $\str{script\_lpo\_ld}$ in the same window
  previously. Indexed by this nonce, there can be an email address
  (identity) recorded in the localStorage which is then copied to the
  script's state. Otherwise an email address is non-deterministically
  choosen (and copied to the script's state) out of the email
  addresses owned by the browser.

  The script now always issues the command to create an iframe, the
  PIF. The URL for the PIF is determined by the domain of the email
  address now recorded in the state. The script then transitions to
  the state $\str{startPIF}$.

\item[$q=\str{startPIF}$] In this state, the script waits for a \pm
  from the PIF containing a \texttt{ping} message. If such a message
  is received and the sender's window and origin match the PIF, the
  script sends a \texttt{pong} message back to the PIF and transitions
  to the state $\str{runPIF}$.

\item[$q=\str{runPIF}$] This is the state in which
  $\str{script\_lpo\_ld}$ interacts with the PIF. This state handles
  all \pms the LD expects to receive from the latest PIF (as recorded
  in $\str{PIFindex}$ in its state). If the \pm received was sent from
  the PIF's window and the PIF's origin, it behaves as follows,
  depending on the first element of the received postMessage:

  \begin{description}
  \item[\pm \texttt{beginProvisioning}] The script responds with a \pm to the PIF containing
    the email address of the identity which is to authenticate to the
    relying party (as recorded in the LD's state).
  \item[\pm \texttt{genKeyPair}] The script creates a fresh key pair (i.e. the
    LD chooses a fresh nonce) and sends the public key contained in
    an \pm to the PIF.
  \item[\pm \texttt{registerCertificate}] The script stores the UC received in
    this \pm in the LD's state. If the context contained in the
    script's state indicates that the browser is authenticated to LPO,
    the script transitions to the state
    $\str{createCAPforRP}$. Otherwise, the script transitions to the
    state $\str{createCAPforLPO}$.
  \item[\pm \texttt{raiseProvisioningFailure}] This message indicates
    that no one is logged in. The script now chooses a fresh nonce,
    the so-called \emph{idpNonce}, which is stored in the
    sessionStorage. In the localStorage, this nonce is used as a key
    under which the email address is stored, the LD is currently
    trying to get an UC for. The script navigates the window it is
    running to the authentication path at the identity provider
    responsible for the email address.
  \end{description}

\item[$q=\str{createCAPforLD}$] In this state, the script creates an IA for LPO, combines it
  with the UC (stored in the script's state) to a CAP and sends the
  CAP to LPO in an \xhr. The nonce identifying the \xhr is stored as
  $\str{refXHRLPOauth}$ in the script's state.

\item[$q=\str{receiveLPOauthresponse}$] In this state, the script expects the response to the
  \xhr identified by the nonce $\str{refXHRLPOauth}$. If the response
  indicates a successful authentication at LPO, the context recorded
  in the script's state is changed accordingly and the script
  transitions to the state $\str{createCAPforRP}$.

\item[$q=\str{createCAPforRP}$] In this state, the script creates an IA for the request origin
  (as recorded in the script's state), combines the IA with the UC to
  a CAP, and sends the CAP in a postMessage to its parent restricting
  the receiver to the request origin. The script records in the
  localStorage that the email address it is currently using is logged
  in at the request origin. The script then transitions to the state
  $\str{null}$.

\item[$q\equiv\str{null}$] In this state, the script does nothing.

\end{description}

We now formally specify the relation $\mi{script\_lpo\_ld}
\subseteq (\terms \times 2^{\nonces})\times \terms$ of the
LD's scripting process. Just like in
Appendix~\ref{sec:descr-web-brows}, we describe this
relation by a non-deterministic algorithm. Like all
scripts, the input term given to this script is determined
by the browser and the browser expects a term of a specific
form (see Algorithm~\ref{alg:runscript})

\captionof{algorithm}{\label{alg:scriptlpold-pidp} Relation of $\mi{script\_lpo\_ld}$ }
\begin{algorithmic}[1]
\Statex[-1] \textbf{Input:} $\langle\mi{tree}$, $\mi{docnonce}$, $\mi{scriptstate}$, $\mi{scriptinputs}$, $\mi{cookies}$, $\mi{localStorage}$, $\mi{sessionStorage}$,\breakalgohook{-1}$\mi{ids}$, $\mi{secret}\rangle$, $\mi{nonces}$ 
\Let{$s'$}{$\mi{scriptstate}$}
\Let{$\mi{cookies}'$}{$\mi{cookies}$}
\Let{$\mi{localStorage}'$}{$\mi{localStorage}$}
\Let{$\mi{sessionStorage}'$}{$\mi{sessionStorage}$}

\Switch{$s'.\str{q}$}
 \Case{$\str{init}$}
  \Let{$\mi{command}$}{$\langle\tPostMessage$, $\textsf{OPENERWINDOW}(\mi{tree}, \mi{docnonce})$, \breakalgohook{2}
$\an{\str{ldready},\an{}}$, $\bot\rangle$}
  \Let{$s'.\str{q}$}{$\str{start}$}
  \State \textbf{stop} $\an{s',\mi{cookies}',\mi{localStorage}',\mi{sessionStorage}',\mi{command}}$
 \EndCase

 \Case{$\str{start}$}
  \Let{$\mi{input},s'$}{\textsf{CHOOSEINPUT}($s',\mi{scriptinputs}$)}
  \If{$\proj{1}{\mi{input}} \equiv \tPostMessage$}
   \Let{$\mi{senderWindow}$}{$\proj{2}{\mi{input}}$}
   \Let{$\mi{senderOrigin}$}{$\proj{3}{\mi{input}}$}
   \Let{$m$}{$\proj{4}{\mi{input}}$}

   \If{$m \equiv \an{\str{request},\an{}}$}
    \Let{$s'.\str{requestOrigin}$}{$\mi{senderOrigin}$}
    \LetND{$s'.\str{refXHRctx}$}{$\mi{nonces}$}
    \Let{$\mi{command}$}{$\an{\tXMLHTTPRequest,\textsf{URL}^\LPO_\str{/ctx},\mGet,\an{},s'.\str{refXHRctx}}$} \label{line:ctx-over-https-2-pidp}
    \Let{$s'.\str{q}$}{$\str{receiveContext}$}
    \State \textbf{stop} $\an{s',\mi{cookies}',\mi{localStorage}',\mi{sessionStorage}',\mi{command}}$
   \EndIf
  \EndIf
 \EndCase

 \Case{$\str{receiveContext}$}
  \Let{$\mi{input},s'$}{\textsf{CHOOSEINPUT}($s',\mi{scriptinputs}$)}
  \If{$(\proj{1}{\mi{input}} \equiv \tXMLHTTPRequest) \wedge (\proj{3}{\mi{input}} \equiv s'.\str{refXHRctx})$}
   \Let{$s'.\str{context}$}{$\proj{2}{\mi{input}}$} \label{line:store-context-2-pidp}
   \Let{$s'.\str{q}$}{$\str{startPIF}$}
   \Let{$\mi{idpnonce}$}{$\mi{sessionStorage}[\str{idpnonce}]$} \label{line:get-idpnonce-ld-pidp}
   \If{$\mi{idpnonce} \equiv \an{} \vee \mi{localStorage}[\mi{idpnonce}] \equiv \an{}$} \label{line:check-idpnonce-ld-pidp}
     \LetND{$s'.\str{email}$}{$\mi{ids}$} \label{line:set-semail-ld-pidp-1}
   \Else
     \Let{$s'.\str{email}$}{$\mi{localStorage}[\mi{idpnonce}]$} \label{line:set-semail-ld-pidp-2}
     \Let{$\mi{sessionStorage}[\str{idpnonce}]$}{$\an{}$}
   \EndIf
   \Let{$\mi{url}$}{$\an{\cUrl, \https, \pi_2(s'.\str{email}), \str{/pif}}$}
   \Let{$s'.\str{PIFindex}$}{$|\mi{subwindows}|+1$} \breakalgo{1}\Comment{Index of the next subwindow to be created.}
   \Let{$\mi{command}$}{$\an{\tIframe,\mi{url},\wSelf}$}
  \State \textbf{stop} $\an{s',\mi{cookies}',\mi{localStorage}',\mi{sessionStorage}',\mi{command}}$
  \EndIf
 \EndCase

 \Case{$\str{startPIF}$}
  \Let{$\mi{idpOrigin}$}{$\an{\pi_2(s'.\str{email}), \https}$}
  \Let{$\mi{input},s'$}{\textsf{CHOOSEINPUT}($s',\mi{scriptinputs}$)}
  \Let{$\mi{pifNonce}$}{$\proj{s'.\str{PIFindex}}{\mi{subwindows}}.\str{nonce}$}
  \If{$\proj{1}{\mi{input}} \equiv \tPostMessage$}
   \Let{$\mi{senderWindow}$}{$\proj{2}{\mi{input}}$}
   \Let{$\mi{senderOrigin}$}{$\proj{3}{\mi{input}}$}
   \Let{$m$}{$\proj{4}{\mi{input}}$}

   \If{$m \equiv \str{ping} \wedge \mi{senderWindow} \equiv \mi{pifNonce} $\breakalgohook{3}$\wedge \mi{senderOrigin} \equiv \mi{idpOrigin}$}
     \Let{$\mi{command}$}{$\an{\tPostMessage,\mi{pifNonce},\str{pong},\mi{idpOrigin}}$}
     \Let{$s'.\str{q}$}{$\str{runPIF}$}
     \State \textbf{stop} $\an{s',\mi{cookies}',\mi{localStorage}',\mi{sessionStorage}',\mi{command}}$
   \EndIf
  \EndIf
 \EndCase

 \Case{$\str{runPIF}$}
  \Let{$\mi{idpOrigin}$}{$\an{\pi_2(s'.\str{email}), \https}$}
  \Let{$\mi{input},s'$}{\textsf{CHOOSEINPUT}($s',\mi{scriptinputs}$)}
  \Let{$\mi{pifNonce}$}{$\proj{s'.\str{PIFindex}}{\mi{subwindows}}.\str{nonce}$}
  \If{$\proj{1}{\mi{input}} \equiv \tPostMessage$}
   \Let{$\mi{senderWindow}$}{$\proj{2}{\mi{input}}$}
   \Let{$\mi{senderOrigin}$}{$\proj{3}{\mi{input}}$}
   \Let{$m$}{$\proj{4}{\mi{input}}$}

   \If{$\mi{senderWindow} \equiv \mi{pifNonce} \wedge \mi{senderOrigin} \equiv \mi{idpOrigin}$}
     \Switch{$\proj{1}{m}$}
       \Case{$\str{beginProvisioning}$}
         \Let{$\mi{jschannel\_nonce}$}{$\proj{2}{m}$}
         \Let{$\mi{command}$}{$\langle\tPostMessage$, $\mi{pifNonce}$,\breakalgohook{6}$\an{jschannel\_nonce, s'.\str{email}}$, $\mi{idpOrigin}\rangle$}
         \State \textbf{stop} $\an{s',\mi{cookies}',\mi{localStorage}',\mi{sessionStorage}',\mi{command}}$
       \EndCase
       \Case{$\str{genKeyPair}$}
         \Let{$\mi{jschannel\_nonce}$}{$\proj{2}{m}$}
         \LetND{$s'.\str{key}$}{$\mi{nonces}$} 
         \Let{$\mi{command}$}{$\langle\tPostMessage$, $\mi{pifNonce}$,\breakalgohook{6}$\an{jschannel\_nonce, \pub(s'.\str{key})}$, $\mi{idpOrigin}\rangle$}
         \State \textbf{stop} $\an{s',\mi{cookies}',\mi{localStorage}',\mi{sessionStorage}',\mi{command}}$
       \EndCase
       \Case{$\str{registerCertificate}$}
         \If{$\proj{1}{\unsig{\proj{2}{m}}} \equiv s'.\str{email} \wedge s'.\str{email} \not\equiv \an{}$}\breakalgo{1}\Comment{This check is our fix against identity injection.}
           \Let{$s'.\str{uc}$}{$\proj{2}{m}$} \label{line:set-uc-ld-pidp}
           \Let{$\mi{loggedIn}$}{$\proj{1}{s'.\str{context}}$}
           \If{$\mi{loggedIn} \equiv \True$}
             \Let{$s'.\str{q}$}{$\str{createCAPforRP}$}
           \EndIf
           \Let{$s'.\str{q}$}{$\str{createCAPforLPO}$}
         \EndIf
         \State \textbf{stop} $\an{s',\mi{cookies}',\mi{localStorage}',\mi{sessionStorage}',\mi{command}}$
         
       \EndCase
       \Case{$\str{raiseProvisioningFailure}$}
         \LetND{$\mi{idpnonce}$}{$\mi{nonces}$} \label{line:chose-idpnonce-ld-pidp}
         \Let{$\mi{localStorage}'[\mi{idpnonoce}]$}{$s'.\str{email}$} \label{write-localstorage-idpnonce-ld-pidp}
         \Let{$\mi{sessionStorage}'[\str{idpnonce}]$}{$\mi{idpnonce}$} \label{line:write-sessionstorage-idpnonce-ld-pidp}
         \Let{$\mi{command}$}{$\an{\tHref,\an{\cUrl, \https, \proj{2}{s'.\str{email}}, \an{}} ,\wSelf}$} %
         \State \textbf{stop} $\an{s',\mi{cookies}',\mi{localStorage}',\mi{sessionStorage}',\mi{command}}$
       \EndCase
     \EndSwitch
   \EndIf
  \EndIf
 \EndCase

 \Case{$\str{createCAPforLPO}$}
   \Let{$\mi{ia}$}{$\sig{\an{\mapDomain(\fAP{LPO}), \https}}{s'.\str{key}}$}
   \Let{$\mi{cap}$}{$\an{s'.\str{uc},\mi{ia}}$}
   \Let{$\mi{body}$}{$\an{\mi{cap}, \proj{2}{s'.\str{context}}}$}
   \LetND{$s'.\str{refXHRLPOauth}$}{$\mi{nonces}$}
  \Let{$\mi{command}$}{$\an{\tXMLHTTPRequest,\textsf{URL}^\LPO_\str{/auth},\mPost,\mi{body}, s'.\str{refXHRLPOauth}}$}
  \Let{$s'.\str{q}$}{$\str{receiveLPOauthresponse}$}
  \State \textbf{stop} $\an{s',\mi{cookies}',\mi{localStorage}',\mi{sessionStorage}',\mi{command}}$
 \EndCase

 \Case{$\str{receiveLPOauthresponse}$}
  \Let{$\mi{input},s'$}{\textsf{CHOOSEINPUT}($s',\mi{scriptinputs}$)}
  \If{$(\proj{1}{\mi{input}} \equiv \tXMLHTTPRequest) \wedge (\proj{3}{\mi{input}} \equiv s'.\str{refXHRLPOauth})$ \breakalgohook{2}$ \wedge \proj{2}{\mi{input}} \equiv \True$}
    \Let{$\proj{1}{s'.\str{context}}$}{$\True$}
    \Let{$s'.\str{q}$}{$\str{createCAPforRP}$}    
    \State \textbf{stop} $\an{s',\mi{cookies}',\mi{localStorage}',\mi{sessionStorage}',\mi{command}}$
  \EndIf
 \EndCase

 \Case{$\str{createCAPforRP}$}
   \Let{$\mi{ia}$}{$\sig{s'.\str{requestOrigin}}{s'.\str{key}}$}
   \Let{$\mi{cap}$}{$\an{s'.\str{uc},\mi{ia}}$}  \label{line:use-uc-ld-pidp}
   \Let{$\mi{command}$}{$\langle\tPostMessage$, $\textsf{OPENERWINDOW}(\mi{tree}, \mi{docnonce})$,\breakalgohook{2}$\an{\str{response},\mi{cap}}$, $s'.\str{requestOrigin}\rangle$} %
   \Let{$s'.\str{q}$}{$\str{null}$}
   \Let{$\mi{localStorage}'[\str{siteInfo}][s'.\str{requestOrigin}]$}{$s'.\str{email}$} \label{write-localstorage-siteinfo-ld-pidp}
   \State \textbf{stop} $\an{s',\mi{cookies}',\mi{localStorage}',\mi{sessionStorage}',\mi{command}}$ %
 \EndCase

\EndSwitch

\State \textbf{stop} $\an{\mi{scriptstate},\mi{cookies},\mi{localStorage},\mi{sessionStorage},\an{}}$
\end{algorithmic} \setlength{\parindent}{1em}

\subsubsection{Relying Party Web Page Script (script\_rp\_index).}\label{app:scriptrpindex-pidp}
This script models the default web page at a RP. The user usually
triggers the login process on this page. Its formal specification,
presented next, follows the one presented for the other scripts above.

\begin{definition}\label{def:scriptstaterpindex-pidp}
  A \emph{scriptstate $s$ of $\mi{script\_rp\_index}$} is a term
  of the form $\langle q$, $\mi{CIFindex}$, $\mi{LDindex}$,
  $\mi{dialogRunning}$, $\mi{cap}$, $\mi{handledInputs}$,
  $\mi{refXHRcap} \rangle$ with $q \in \mathbb{S}$,
  $\mi{CIFindex} \in\mathbb{N} \cup \{\bot\}$,
  $\mi{dialogRunning} \in \{\True,\bot\}$, $\mi{cap} \in
  \terms$, $\mi{handledInputs} \subsetPairing \mathbb{N}$,
  $\mi{refXHRcap} \in \nonces \cup \{\bot\}$. 
  We call $s$ the \emph{initial scriptstate of
  $\mi{script\_rp\_index}$} iff $s \equiv
  \an{\str{init},\bot,\bot,\bot,\an{},\an{},\bot}$.
\end{definition}

Before we provide the formal specification of the relation
that defines the behavior of $\mi{script\_rp\_index}$, we
present an informal description. The behavior mainly
depends on the state $q$ the script is in.

\begin{description}

\item[$q\equiv\str{init}$] This is the initial state. The script
  creates the CIF \iframe and then transitions to
  $\str{receiveCIFReady}$.

\item[$q\equiv\str{receiveCIFReady}$] In this state, the script
  expects a \texttt{cifready} \pm from the CIF \iframe with the sender origin of
  $\LPO$. The script chooses some ID, $\an{}$, or $\bot$ and sends this in a
  \texttt{loaded} \pm to the CIF \iframe with receiver's origin set to
  the origin of $\LPO$.\footnote{From the point of view of the real
    scripts running at RP either some ID is considered to be logged in
    (e.g. from some former ``session''), or no one is considered
    to be logged in ($\an{}$), or the script $\mi{script\_rp\_index}$ does
    not know if it should consider anyone to be logged in
    ($\bot$). This is overapproximated here by allowing
    $\mi{script\_rp\_index}$ to choose non-deterministically between
    these cases.}  It then transitions to the state $\str{default}$.

\item[$q\equiv\str{default}$] In this state, the script chooses
  non-deterministically between (1) opening the LD subwindow and then
  transitioning to the same state or (2) handling one of the following
  \pms (identified by their first element):

  \begin{description}

  \item[\pm \texttt{login}] This message has to be sent from the CIF
    with origin of $\LPO$. Handling this \pm stores the CAP (contained
    in the \pm) in the scriptstate and then transitions to the
    $\str{sendCAP}$ state.
  \item[\pm \texttt{logout}] This message has to be sent from the CIF
    with origin of $\LPO$. Handling this \pm has no effect and results
    in the same state.
  \item[\pm \texttt{ldready}] This message can only be handled after
    the LD has been opened and before a \texttt{response} \pm has been
    received. The \texttt{ldready} \pm has to be sent from the origin
    of $\LPO$. The script sends a \texttt{request} \pm to the LD and
    stays in the $\str{default}$ state.
  \item[\pm \texttt{response}] This message can only be handled after
    the LD has been opened and before another \texttt{response} \pm
    has been received. The \texttt{ldready} \pm has to be sent from
    the origin of $\LPO$. Handling this \pm stores the CAP (contained
    in the \pm) in the scriptstate, closes the LD, and then
    transitions to the $\str{dlgClosed}$ state.
  \end{description}

\item[$q\equiv\str{dlgClosed}$] In this state, the script sends a
  \texttt{loggedInUser} \pm to the CIF and transitions to
  the $\str{loggedInUser}$ state.
\item[$q\equiv\str{loggedInUser}$] In this state, the script sends a \texttt{dlgCmplt} \pm
  to the CIF and transitions to the $\str{sendCAP}$ state.
\item[$q\equiv\str{sendCAP}$] In this state, the script sends the CAP to RP as a $\mPost$
  \xhr and then transitions to the
  $\str{receiveServiceToken}$ state.
\item[$q\equiv\str{receiveServiceToken}$] In this state, the script receives $\an{n,i}$
  from RP, but does not do anything with it. The script then
  transitions to the $\str{default}$ state.

\end{description}

We now formally specify the relation $\mi{script\_rp\_index}
\subseteq (\terms \times 2^{\nonces})\times \terms$ of the
RP-Doc's scripting process. Just like in
Appendix~\ref{sec:descr-web-brows}, we describe this
relation by a non-deterministic algorithm. Like all
scripts, the input term given to this script is determined
by the browser and the browser expects a term of a specific
form (see Algorithm~\ref{alg:runscript}). Following
Algorithm~\ref{alg:scriptrpindex-pidp}, we provide some more
explanation. 

\captionof{algorithm}{\label{alg:scriptrpindex-pidp} Relation of $\mi{script\_rp\_index}$}
\begin{algorithmic}[1]
\Statex[-1] \textbf{Input:} $\langle\mi{tree}$, $\mi{docnonce}$, $\mi{scriptstate}$, $\mi{scriptinputs}$, $\mi{cookies}$, $\mi{localStorage}$, $\mi{sessionStorage}$,\breakalgohook{-1}$\mi{ids}$, $\mi{secret}\rangle$, $\mi{nonces}$ 
\Let{$s'$}{$\mi{scriptstate}$}
\Let{$\mi{cookies}'$}{$\mi{cookies}$}
\Let{$\mi{localStorage}'$}{$\mi{localStorage}$}
\Let{$\mi{sessionStorage}'$}{$\mi{sessionStorage}$}

\Switch{$s'.\str{q}$}
 \Case{$\str{init}$}
  \Let{$\mi{command}$}{$\an{\tIframe,\textsf{URL}^\LPO_\str{/cif},\mathsf{GETWINDOW}(\mi{tree},\mi{docnonce})}$}\label{line:cifindex-begin-pidp}
  \Let{$s'.\str{q}$}{$\str{receiveCIFReady}$}
  \Let{$\mi{subwindows}$}{$\mathsf{SUBWINDOWS}(\mi{tree},\mi{docnonce})$}
  \Let{$s'.\str{CIFindex}$}{$|\mi{subwindows}|+1$} \Comment{Index of the next subwindow to be created.}
  \State \textbf{stop} $\an{s',\mi{cookies}',\mi{localStorage}',\mi{sessionStorage}',\mi{command}}$\label{line:cifindex-end-pidp}
 \EndCase

 \Case{$\str{receiveCIFReady}$}
  \Let{$\mi{input},s'$}{\textsf{CHOOSEINPUT}($s',\mi{scriptinputs}$)}
  \If{$\proj{1}{\mi{input}} \equiv
    \tPostMessage$}
   \Let{$\mi{senderWindow}$}{$\proj{2}{\mi{input}}$}
   \Let{$\mi{senderOrigin}$}{$\proj{3}{\mi{input}}$}
   \Let{$m$}{$\proj{4}{\mi{input}}$}
   \Let{$\mi{subwindows}$}{$\mathsf{SUBWINDOWS}(\mi{tree},\mi{docnonce})$}
   \If{$(m \equiv \an{\str{cifready},\an{}})$\breakalgohook{3}$ \wedge(\mi{senderOrigin} \equiv \mathsf{origin_\LPO}) $\breakalgohook{3}$\wedge(\mi{senderWindow} \equiv \proj{s'.\str{CIFindex}}{\mi{subwindows}}.\str{nonce})$}
    \LetND{$\mi{id}$}{$\{\bot,\an{}\} \cup \IDs $}
    \Let{$\mi{command}$}{$\langle\tPostMessage$, $\proj{s'.\str{CIFindex}}{\mi{subwindows}}$,\breakalgohook{4}
$\an{\str{loaded},\mi{id}}$, $\mathsf{origin_\LPO}\rangle$}
    \Let{$s'.\str{q}$}{$\str{default}$}
    \State \textbf{stop} $\an{s',\mi{cookies}',\mi{localStorage}',\mi{sessionStorage}',\mi{command}}$
   \EndIf
  \EndIf
 \EndCase

 \Case{$\str{default}$} \label{line:state-default-pidp}
  \If{$s'.\str{dialogRunning} \equiv \bot$}
   \LetND{$\mi{choice}$}{$\{\str{openLD},\str{handlePM}\}$}
  \Else
   \Let{$\mi{choice}$}{$\str{handlePM}$}
  \EndIf
  \If{$\mi{choice} \equiv \str{openLD}$} 
   \Let{$s'.\str{dialogRunning}$}{$\True$} \label{line:ldindex-begin-pidp}
   \Let{$\mi{command}$}{$\an{\tHref,\textsf{URL}^\LPO_\str{/ld},\wBlank}$}
   \Let{$s'.\str{q}$}{$\str{default}$}
   \State \textbf{stop}
  $\an{s',\mi{cookies}',\mi{localStorage}',\mi{sessionStorage}',\mi{command}}$ \label{line:ldindex-end-pidp}

  \Else
   \Let{$\mi{input},s'$}{\textsf{CHOOSEINPUT}($s',\mi{scriptinputs}$)}
   \If{$\proj{1}{\mi{input}} \equiv \tPostMessage$}
    \Let{$\mi{senderWindow}$}{$\proj{2}{\mi{input}}$}
    \Let{$\mi{senderOrigin}$}{$\proj{3}{\mi{input}}$}
    \Let{$m$}{$\proj{4}{\mi{input}}$}
    \Let{$\mi{subwindows}$}{$\mathsf{SUBWINDOWS}(\mi{tree},\mi{docnonce})$}
    \If{$\mi{senderOrigin} \equiv \mathsf{origin_\LPO}$}
     \If{$\mi{senderWindow} \equiv \proj{s'.\str{CIFindex}}{\mi{subwindows}}.\str{nonce}$}
      \If{$\proj{1}{m} \equiv \str{login}$}
       \Let{$s'.\str{cap}$}{$\proj{2}{m}$} \label{line:set-cap-pidp}
       \Let{$s'.\str{q}$}{$\str{sendCAP}$}
       \State \textbf{stop} $\an{s',\mi{cookies}',\mi{localStorage}',\mi{sessionStorage}',\an{}}$
      \ElsIf{$\proj{1}{m} \equiv \str{logout}$}
       \Let{$s'.\str{q}$}{$\str{default}$}
       \State \textbf{stop} $\an{s',\mi{cookies}',\mi{localStorage}',\mi{sessionStorage}',\an{}}$
      \EndIf

     \ElsIf{$s'.\str{dialogRunning} \equiv \True$}
      \If{$\proj{1}{m} \equiv \str{ldready}$}                 
       \Let{$\mi{command}$}{$\langle\tPostMessage$,\breakalgohook{7}$\mathsf{AUXWINDOW}(\!\mi{tree}, \mi{docnonce}\!)$, $\an{\str{request},\!\an{}}$, $\mathsf{origin_\LPO}\rangle$}
       \Let{$s'.\str{q}$}{$\str{default}$}
       \State \textbf{stop}\! $\an{s',\mi{cookies}',\mi{localStorage}',\mi{sessionStorage}',\mi{command}}$
      \ElsIf{$\proj{1}{m} \equiv \str{response}$}
       \Let{$s'.\str{dialogRunning}$}{$\bot$}
       \Let{$s'.\str{cap}$}{$\proj{2}{m}$} \label{line:set-cap-2-pidp}
   \Let{$\mi{command}$}{$\an{\tClose,\mathsf{AUXWINDOW}(\mi{tree},\mi{docnonce})}$} 
       \Let{$s'.\str{q}$}{$\str{dlgClosed}$}
       \State \textbf{stop}\! $\an{s',\mi{cookies}',\mi{localStorage}',\mi{sessionStorage}',\mi{command}}$
      \EndIf
     \EndIf
    \EndIf
   \EndIf
  \EndIf
 \EndCase

 \Case{$\str{dlgClosed}$} 
   \Let{$\mi{subwindows}$}{$\mathsf{SUBWINDOWS}(\mi{tree},\mi{docnonce})$}
   \Let{$\mi{id}$}{$\proj{1}{\unsig{\proj{1}{s'.\str{cap}}}}$} \label{line:use-cap-1-pidp}\Comment{Extract ID from CAP.}
\Let{$\mi{command}$}{$\langle\tPostMessage$, $\proj{s'.\str{CIFindex}}{\mi{subwindows}}.\str{nonce}$,\breakalgohook{2}$\an{\str{loggedInUser}$, $\mi{id}}$, $\mathsf{origin_\LPO}\rangle$}
  \Let{$s'.\str{q}$}{$\str{loggedInUser}$}
  \State \textbf{stop} $\an{s',\mi{cookies}',\mi{localStorage}',\mi{sessionStorage}',\mi{command}}$
 \EndCase

 \Case{$\str{loggedInUser}$}
   \Let{$\mi{subwindows}$}{$\mathsf{SUBWINDOWS}(\mi{tree},\mi{docnonce})$}
\Let{$\mi{command}$}{\breakalgohook{2}$\an{\tPostMessage,\proj{s'.\str{CIFindex}}{\mi{subwindows}}.\str{nonce},\an{\str{dlgCmplt},\an{}},\mathsf{origin_\LPO}}$}
  \Let{$s'.\str{q}$}{$\str{sendCAP}$}
  \State \textbf{stop} $\an{s',\mi{cookies}',\mi{localStorage}',\mi{sessionStorage}',\mi{command}}$
 \EndCase

 \Case{$\str{sendCAP}$} \label{line:state-sendcap-pidp}
   \LetND{$s'.\str{refXHRcap}$}{$\mi{nonces}$}
   \LetST{$\mi{host}$, $\mi{protocol}$}{\breakalgohook{2}$\an{host, protocol} \equiv \mathsf{GETORIGIN}(\mi{tree}, \mi{docnonce})$\breakalgohook{2}}{\textbf{stop}\breakalgohook{2} $\langle\mi{scriptstate}$, $\mi{cookies}$, $\mi{localStorage}$, $\mi{sessionStorage}$, $\mi{command}\rangle$}
   \Let{$\mi{command}$}{$\langle\tXMLHTTPRequest$, $\an{\cUrl, \mi{protocol}, \mi{host}, \mathtt{/}, \an{}}$, $\mPost$, $s'.\str{cap}$,\breakalgohook{2}$s'.\str{refXHRcap}\rangle$} \label{line:use-cap-2-pidp} \Comment{Relay received CAP to RP.}
   \Let{$s'.\str{q}$}{$\str{receiveServiceToken}$}
   \State \textbf{stop} $\an{s',\mi{cookies}',\mi{localStorage}',\mi{sessionStorage}',\mi{command}}$ \label{line:send-cap-pidp}
 \EndCase

 \Case{$\str{receiveServiceToken}$}
  \Let{$\mi{input},s'$}{\textsf{CHOOSEINPUT}($s',\mi{scriptinputs}$)}
  \If{$(\proj{1}{\mi{input}} \equiv \tXMLHTTPRequest) \wedge (\proj{3}{\mi{input}} \equiv s'.\str{refXHRcap})$}
   \Let{$s'.\str{q}$}{$\str{default}$}\label{line:ignore-service-token-pidp}
   \State \textbf{stop} $\an{s',\mi{cookies}',\mi{localStorage}',\mi{sessionStorage}',\an{}}$
  \EndIf
 \EndCase

\EndSwitch

\State \textbf{stop} $\an{\mi{scriptstate},\mi{cookies},\mi{localStorage},\mi{sessionStorage},\an{}}$
\end{algorithmic} \setlength{\parindent}{1em}

In Lines~\ref{line:cifindex-begin-pidp}--\ref{line:cifindex-end-pidp}
and \ref{line:ldindex-begin-pidp}--\ref{line:ldindex-end-pidp} the
script asks the browser to create iframes. To obtain the
window reference for these iframes, the script first
determines the current number of subwindows and stores it
(incremented by 1) in the scriptstate ($\str{CIFindex}$ and
$\str{LDindex}$, respectively).  When the script is invoked
the next time, the iframe the script asked to be created
will have been added to the sequence of subwindows by the
browser directly following the previously existing
subwindows. The script can therefore access the iframe by
the indexes $\str{CIFindex}$ and $\str{LDindex}$,
respectively.

\subsubsection{Identity Provider Authentication Dialog Script (script\_idp\_ad).}
This script runs in the LD after $\mi{script\_lpo\_ld}$ has navigated
the LD window. The purpose of this script is to authenticate the
browser to the identity provider.

The script non-deterministically chooses if it sends authentication data to the IdP (i.e. its origin) via an \xhr, or if it navigates the window to an URL at LPO which servers $\mi{script\_lpo\_ld}$. Note that $\mi{script\_idp\_ad}$ does not read or change its scriptstate. Hence, we omit the definition of the scriptstate for this script.

\captionof{algorithm}{\label{alg:scriptidpad-pidp} Relation of $\mi{script\_idp\_ad}$ }
\begin{algorithmic}[1]
\Statex[-1] \textbf{Input:} $\langle\mi{tree}$, $\mi{docnonce}$, $\mi{scriptstate}$, $\mi{scriptinputs}$, $\mi{cookies}$, $\mi{localStorage}$, $\mi{sessionStorage}$,\breakalgohook{-1}$\mi{ids}$, $\mi{secret}\rangle$, $\mi{nonces}$ 
\LetND{$\mi{action}$}{$\{\str{authenticate},\str{navigate}\}$}
\If{$\mi{action} \equiv \str{authenticate}$}
  \LetND{$\mi{email}$}{$\mi{ids}$}
  \Let{$\mi{body}$}{$\an{\mi{email},\mi{secret}}$}
  \LetST{$\mi{host}$, $\mi{protocol}$}{\breakalgohook{1}$\an{host, protocol} \equiv \mathsf{GETORIGIN}(\mi{tree},\mi{docnonce})$\breakalgohook{1}\!\!}{\breakalgohook{1}\textbf{stop} $\an{\mi{scriptstate},\mi{cookies},\mi{localStorage},\mi{sessionStorage},\an{}}$}
  \Let{$\mi{command}$\!}{\!$\an{\tXMLHTTPRequest,\an{\cUrl, \mi{protocol}, \mi{host}, \mathtt{/auth}, \an{}},\mPost,\mi{body},\bot}$}
  \State \textbf{stop} $\an{\mi{scriptstate},\mi{cookies}',\mi{localStorage}',\mi{sessionStorage}',command}$
\Else
  \Let{$\mi{command}$}{$\an{\tHref,\an{\cUrl, \https, \mapDomain(\fAP{LPO}), \mathtt{/ld}, \an{}},\wSelf}$}
  \State \textbf{stop} $\an{\mi{scriptstate},\mi{cookies}',\mi{localStorage}',\mi{sessionStorage}',command}$
\EndIf
\end{algorithmic} \setlength{\parindent}{1em}

\subsubsection{Identity Provider Provisioning Iframe Script (script\_idp\_pif).}
This script acts as a proxy between the LD or CIF and the IdP server.

\begin{definition}\label{def:scriptstateidppif-pidp}
  A \emph{scriptstate $s$ of $\mi{script\_idp\_pif}$} is a term of the form
  $\langle q$, $\mi{emails}$, $\mi{pubkeys}$, $\mi{ucs}$,
  $\mi{provisioningnonces}$, $\mi{genkeypairnonces}$, $\mi{xhrnonces}$, $\mi{handledInputs}
  \rangle$ with $q \in \mathbb{S}$,
  $\mi{emails}$, $\mi{pubkeys}$, $\mi{ucs} \in\terms$,
  $\mi{provisioningnonces}$, $\mi{genkeypairnonces}$, $\mi{xhrnonces}
  \in \nonces \cup \{\bot\}$, $\mi{handledInputs} \subsetPairing \mathbb{N}$.
  We call $s$ the \emph{initial scriptstate of
  $\mi{script\_idp\_pif}$} iff $s \equiv
  \an{\str{init},\an{},\an{},\an{},\bot,\bot,\bot}$.
\end{definition}

Before we provide the formal specification of the relation
that defines the behavior of $\mi{script\_idp\_pif}$, we
present an informal description. The behavior mainly
depends on the state $q$ the script is in.

\begin{description}
\item[$q=\str{init}$] This is the initial state. Its only transition
  handles no input and outputs a \pm \texttt{ping} to its parent
  window, which has to have the origin of LPO, and transitions to
  $\str{waiting}$.

\item[$q=\str{waiting}$] In this state, the script expects a \pm
  containing either $\str{ping}$ or $\str{pong}$, which has to be sent
  by the parent window from the origin of LPO. If such a \pm has been
  received, the script transitions to $\str{default}$.

\item[$q=\str{default}$] In this state, the script chooses an action non-deterministically out of the following:
  \begin{description}
  \item[$\str{beginprovisioning}$] The script sends a \pm to the
    parent window, which has to have the origin of LPO, indicating
    that the provisioning process of a UC should start. A fresh nonce
    is chosen, stored in the script's state, and included in this
    \pm. The \pm requests the email address of the user from the
    receiver. The address is to be sent to the PIF in a \pm which is
    identified by the nonce in the request.

  \item[$\str{genkeypair}$] The script sends a \pm to the parent
    window, which has to have the origin of LPO, indicating that a new
    key pair should be generated. This \pm requests the public key of
    this fresh key pair. As above, a nonce is included to identify the
    response corresponding to the request.

  \item[$\str{registercert}$] The script sends a \pm containing a UC to the
    parent window, which has to have the origin of LPO. This \pm is
    only sent if the script has received a UC before.

  \item[$\str{raisefailure}$] The script sends a \pm to the parent window, which
    has to have the origin of LPO, indicating that the browser is
    currently not authenticated to the identity provider.

  \item[$\str{requestuc}$] The script sends an \xhr to the origin of
    the current document if the scriptstate contains at least one
    email address and one public key. The message contains a
    non-deterministically chosen email address and a public key (from
    the scriptstate). The nonce identifying this \xhr is
    non-deterministically chosen and stored in the scriptstate.

  \item[$\str{handleresponse}$] The script chooses non-deterministically a script
    input and distinguishes if this input is a \pm or an \xhr
    response.

    If the chosen input is a \pm, it is checked if the \pm was sent by
    the parent window and if this window has the origin of LPO. If
    this check is successful, it is checked if the message contains a
    nonce, which was previously been recorded in the script's state. If
    this nonce indicates that this message is a response to a
    $\str{beginProvisioning}$ \pm, the second part is assumed to
    contain an email address. This address is then recorded in the
    script's state.  If the nonce indicates that this message is a
    response to a $\str{genKeyPair}$ \pm, the second part is assumed
    to contain a public key. This public key is then recorded in the
    script's state.

    If the chosen input is an \xhr response, it is checked if the
    nonce identifying the \xhr is recorded in the script's state. If
    this is the case, the message is assumed to contain an UC. The
    content of the message is stored in the script's state.
  \end{description}

\end{description}

\captionof{algorithm}{\label{alg:scriptidppif-pidp} Relation of $\mi{script\_idp\_pif}$ }
\begin{algorithmic}[1]
\Statex[-1] \textbf{Input:} $\langle\mi{tree}$, $\mi{docnonce}$, $\mi{scriptstate}$, $\mi{scriptinputs}$, $\mi{cookies}$, $\mi{localStorage}$, $\mi{sessionStorage}$,\breakalgohook{-1}$\mi{ids}$, $\mi{secret}\rangle$, $\mi{nonces}$ 
\Let{$s'$}{$\mi{scriptstate}$}

\Switch{$s'.\str{q}$}

  \Case{$\str{init}$}
    \Let{$\mi{command}$}{$\langle\tPostMessage$, $\mathsf{PARENTWINDOW}(\mi{tree}, \mi{docnonce})$,\breakalgohook{2}$\an{\str{ping},\an{}}$, $\an{\mapDomain(\fAP{LPO})$, $\https}\rangle$}
    \Let{$s'.\str{q}$}{$\str{waiting}$}
    \State \textbf{stop} $\an{s', \mi{cookies}, \mi{localStorage}, \mi{sessionStorage}, \mi{command}}$
  \EndCase

  \Case{$\str{waiting}$}
    \Let{$\mi{input}, s'$}{$\mathsf{CHOOSEINPUT}(s',\mi{scriptinputs})$}
    \Let{$\mi{senderWindow}$}{$\proj{2}{\mi{input}}$}
    \Let{$\mi{senderOrigin}$}{$\proj{3}{\mi{input}}$}
    \Let{$m$}{$\proj{4}{input}$}
    \If{$\proj{1}{\mi{input}} \in \{\str{ping}, \str{pong}\} $\breakalgohook{2}$\wedge\mi{senderWindow} \equiv \mathsf{PARENTWINDOW}(\mi{tree}, \mi{docnonce}) $\breakalgohook{2}$\wedge\mi{senderOrigin} \equiv \an{\mapDomain(\fAP{LPO}),\https}$}
      \Let{$s'.\str{q}$}{$\str{default}$}
    \EndIf
    \State \textbf{stop} $\an{s', \mi{cookies}, \mi{localStorage}, \mi{sessionStorage}, \an{}}$
  \EndCase

  \Case{$\str{default}$}
    \LetND{$\mi{action}$}{$\{ \str{beginprovisioning}$, $\str{genkeypair}$, $\str{registercert}$,\breakalgohook{2}$\str{raisefailure}$, $\str{requestuc}$, $\str{handleresponse} \}$}

    \Switch{$\mi{action}$}

      \Case{$\str{beginprovisioning}$}
        \LetND{$\mi{jschannel\_nonce}$}{$\mi{nonces}$} 
        \Let{$\mi{command}$}{$\langle \tPostMessage$, $\mathsf{PARENTWINDOW}(\mi{tree}, \mi{docnonce})$,\breakalgohook{4}$\an{\str{beginProvisioning},\mi{jschannel\_nonce}}$, $\mapDomain(\fAP{LPO})\rangle$}
        \Let{$s'.\str{provisioningnonces}$}{\breakalgohook{4}$s'.\str{provisioningnonces} \plusPairing \mi{jschannel\_nonce}$}
        \State \textbf{stop} $\an{s', \mi{cookies}, \mi{localStorage}, \mi{sessionStorage}, \mi{command}}$
      \EndCase

      \Case{$\str{genkeypair}$}
        \LetND{$\mi{jschannel\_nonce}$}{$\mi{nonces}$}
        \Let{$\mi{command}$}{$\langle\tPostMessage$, $\mathsf{PARENTWINDOW}(\mi{tree}, \mi{docnonce})$,\breakalgohook{4}$\an{\str{genKeyPair},\mi{jschannel\_nonce}}$, $\mapDomain(\fAP{LPO})\rangle$}
        \Let{$s'.\str{genkeypairnonces}$}{\breakalgohook{4}$s'.\str{genkeypairnonces} \plusPairing \mi{jschannel\_nonce}$}
        \State \textbf{stop} $\an{s', \mi{cookies}, \mi{localStorage}, \mi{sessionStorage}, \mi{command}}$ \label{line:start-genkeypair-pidp}
      \EndCase
      
      \Case{$\str{registercert}$}
        \If{$s'.\str{ucs} \not\equiv \an{}$}
          \LetND{$\mi{uc}$}{$s'.\str{ucs}$}
          \Let{$\mi{command}$\!}{\!$\langle\tPostMessage$,$\mathsf{PARENTWINDOW}(\!\mi{tree},\!\mi{docnonce}\!)$,\breakalgohook{5}$\an{\str{registerCertificate},\mi{uc}}$, $\mapDomain(\fAP{LPO})\rangle$}
          \State \textbf{stop} $\an{s', \mi{cookies}, \mi{localStorage}, \mi{sessionStorage}, \mi{command}}$
        \EndIf
      \EndCase
      
      \Case{$\str{raisefailure}$}
        \Let{$\mi{command}$}{$\langle\tPostMessage$, $\mathsf{PARENTWINDOW}(\mi{tree}, \mi{docnonce})$,\breakalgohook{4}$\an{\str{raiseProvisioningFailure},\bot}$, $\mapDomain(\fAP{LPO})\rangle$}
        \State \textbf{stop} $\an{s', \mi{cookies}, \mi{localStorage}, \mi{sessionStorage}, \mi{command}}$
      \EndCase
      
      \Case{$\str{requestuc}$}
        \If{$s'.\str{emails} \not\equiv \an{} \wedge s'.\str{pubkeys} \not\equiv \an{}$}
          \LetND{$\mi{email}$}{$s'.\str{emails}$}
          \LetND{$\mi{pubkey}$}{$s'.\str{pubkeys}$}
          \Let{$\mi{body}$}{$\an{\mi{email},\mi{pubkey}}$}
          \LetND{$\mi{xhrnonce}$}{$\mi{nonces}$}
          \Append{$\mi{xhrnonce}$}{$s'.\str{xhrnonces}$}
          \LetST{$\mi{host}$,$\mi{protocol}$}{\breakalgohook{5}$\an{\mi{host},\mi{protocol}} \equiv \mathsf{GETORIGIN}(\mi{tree},\mi{docnonce})$\breakalgohook{5}\!\!}{\breakalgohook{5}\textbf{stop} $\an{s', \mi{cookies}, \mi{localStorage}, \mi{sessionStorage}, \an{}}$}
          \Let{$\mi{command}$}{$\langle \tXMLHTTPRequest$,\breakalgohook{5}$\an{\cUrl,\mi{protocol},\mi{host},\str{/certreq},\an{}}$,$\mPost$,$\mi{body}$,$\mi{xhrnonce} \rangle$}
          \State \textbf{stop} $\an{s', \mi{cookies}, \mi{localStorage}, \mi{sessionStorage}, \mi{command}}$ \label{line:send-req-uc-pidp}
        \EndIf
      \EndCase
      
      \Case{$\str{handleresponse}$}
        \Let{$\mi{input},s'$}{$\mathsf{CHOOSEINPUT}(s',\mi{scriptinputs})$}
        \If{$\proj{1}{\mi{input}} \equiv \tPostMessage$}
          \Let{$\mi{senderWindow}$}{$\proj{2}{\mi{input}}$}
          \Let{$\mi{senderOrigin}$}{$\proj{3}{\mi{input}}$}
          \Let{$m$}{$\proj{4}{input}$}
          \If{$\mi{senderWindow} \equiv \mathsf{PARENTWINDOW}(\mi{tree}, \mi{docnonce}) $\breakalgohook{5}$\wedge\mi{senderOrigin} \equiv \an{\mapDomain(\fAP{LPO}),\https}$}
            \If{$\proj{1}{m} \in s'.\str{provisioningnonces}$}
              \Append{$\proj{2}{m}$}{$s'.\str{emails}$}
            \ElsIf{$\proj{1}{m} \in s'.\str{genkeypairnonces}$}
              \Append{$\proj{2}{m}$}{$s'.\str{pubkeys}$} \label{line:populate-pubkeys-pidp}
            \EndIf
            \State \textbf{stop} $\an{s', \mi{cookies}, \mi{localStorage}, \mi{sessionStorage}, \an{}}$
          \EndIf
        \ElsIf{$\proj{1}{\mi{input}} \equiv \tXMLHTTPRequest $\breakalgohook{4}$\wedge\proj{3}{\mi{input}} \in s'.\str{xhrnonces}$}
          
          \Append{$\proj{2}{\mi{input}}$}{$s'.\str{ucs}$}
          \State \textbf{stop} $\an{s', \mi{cookies}, \mi{localStorage}, \mi{sessionStorage}, \an{}}$
        \EndIf
      \EndCase
      
    \EndSwitch
  \EndCase

\EndSwitch
\State \textbf{stop} $\an{\mi{scriptstate}, \mi{cookies}, \mi{localStorage}, \mi{sessionStorage}, \an{}}$
\end{algorithmic} \setlength{\parindent}{1em}

%% file: securityproperties-formal.tex
\section{Formal Security Properties}\label{app:form-secur-prop}

The security properties for BrowserID, informally introduced in
Section~\ref{sec:securitypropsBrowserID}, are formally defined as
follows. First note that every RP service token $\an{n,i}$ recorded in
RP was created by RP as the result of a unique HTTPS $\mPost$ request
$m$ with a valid CAP for ID $i$. We refer to $m$ as the \emph{request
  corresponding to $\an{n,i}$}.

\begin{definition}\label{def:security-property-pidp} Let $\bidwebsystem$ be a BrowserID web
  system. We say that \emph{$\bidwebsystem$ is secure} if
  for every run $\rho$ of $\bidwebsystem$, every state
  $(S_j, E_j)$ in $\rho$, every $r\in \fAP{RP}$ that is
  honest in $S_j$, every RP service token of the form
  $\an{n,i}$ recorded in $r$ in the state $S_j(r)$, the
  following two conditions are satisfied:

  \textbf{(A)} If $\an{n,i}$ is derivable from the
  attackers knowledge in $S_j$ (i.e., $\an{n,i} \in
  d_{N^\fAP{attacker}}(S_j(\fAP{attacker}))$), then it
  follows that the browser $b$ owning $i$ is fully corrupted in
  $S_j$ (i.e., the value of $\mi{isCorrupted}$ is
  $\fullcorrupt$) or $\mapGovernor(i)$ is not an honest IdP
  (in $S_j$).

  \textbf{(B)} If the request corresponding to $\an{n,i}$
  was sent by some $b\in \fAP{B}$ which is honest in $S_j$,
  then $b$ owns $i$.
\end{definition}

%% file: appendix-proof-pidp.tex
\section{Proof of Theorem~\ref{thm:secur-fixed-syst}}
\label{app:proofbrowserid-pidp}

In order to prove Theorem~\ref{thm:secur-fixed-syst}, we have to prove
Conditions A and B of Definition~\ref{def:security-property-pidp}. We
prove these conditions separately. First, we provide an overview of
the proofs.

\subsection{Overview}

For Condition \textbf{(A)}, we analyze the request to an
honest RP $r$ upon which $r$ returned a service token
$\an{n,i}$, where $i$ is an ID and $n$ a nonce. We show
that it must contain a valid CAP (for the identity
$i$). For this, it must in particular contain a valid UC
and a matching IA. We show that the UC must have been
created by the IdP that governs the identity $i$ (which is
honest by assumption). We can then show that only $b$ can
request a UC at the IdP for the identity $i$, and that $b$
does not leak the private key that corresponds to the
public key used for this UC, and that this key was chosen
from $b$'s set of fresh nonces. Thus, only $b$ can know the
key that is used in the creation of the UC in the CAP. We
show that neither the private key corresponding to the
public key in the UC, nor the IA can leak to the
attacker. Thus, the attacker cannot have sent the request
corresponding to $\an{n,i}$ to the RP $r$. Also, $\an{n,i}$
does not leak to the attacker. The attacker can therefore
not know $\an{n,i}$, which contradicts the assumption and
proves that Condition \textbf{(A)} is satisfied.

For Condition \textbf{(B)}, we focus on the request corresponding to
$\an{n,i}$ as well. We observe that if the request was sent
by $b$, the script that initiated the request was
$\mi{script\_rp\_index}$, which again got the CAP that is
finally used in the request from either
$\mi{script\_lpo\_cif}$ or $\mi{script\_lpo\_ld}$ (any
other sources, including the attacker script, can be ruled
out). In both of these scripts, the identity in the CAP is
checked against the list of identities of the browser
(here, the proposed patch comes into play). This ensures
that the request corresponding to $\an{n,i}$ contains a CAP
for an identity of the browser, which contradicts the
assumption that Condition \textbf{(B)} is not satisfied and thus
proves the theorem.

\subsection{Condition A}

We assume that Condition A is not satisfied and prove that this leads
to a contradiction. That is, we make the following
assumption: There is a run $\rho = s_0, s_1,\dots$ of $\bidwebsystem$,
a state $s_j = (S_j, E_j)$ in $\rho$, an $r \in \fAP{RP}$ that is
honest in $S_j$, an RP service token of the form $\an{n,i}$ recorded
in $r$ in the state $S_j(r)$ such that $\an{n,i} \in
d_{N^\fAP{attacker}}(S_j(\fAP{attacker}))$ and the browser $b$ owning $i$
is not fully corrupted in $S_j$ and $\mapGovernor(i)$ is an honest IdP
in $S_j$.

By definition of RPs, for $\an{n,i}$ there exists a
corresponding HTTPS request received by $r$, which we call
$\mi{req}_\text{cap}$, and a corresponding response
$\mi{resp}_\text{cap}$. The request must contain a valid CAP $c$ and
must have been sent by some atomic process $p$ to $r$. The response
must contain $\an{n,i}$ and it must be encrypted by some symmetric
encryption key $k$ sent in $\mi{req}_\text{cap}$.

In particular, it follows that the request and the response must be of
the following form, where $d_r \in \mathsf{dom}(r)$ is the domain of
$r$, $n_\text{cap}, k \in \nonces$ are some nonces, $\mi{path}$, $\mi{params} \in \terms$, $c$ is some valid
CAP, and $\mi{sts}$ is the Strict-Transport-Security header (as in the
definition of RP's relation):
\begin{align}
  \label{eq:proofreqcapA-pidp} \mi{req}_\text{cap} &=
  \mathsf{enc}_\mathsf{a}(\langle \hreq{ nonce=n_\text{cap}, method=\mPost,
      xhost=d_r, path=\mi{path}, parameters=\mi{params}, headers=[\str{Origin}:
      \an{d_r, \https}], xbody=c},\nonumber\\ &\hspace{4.5em}k\rangle, \pub(\mapKey(d_r)))\\
  \label{eq:proofrespcapA-pidp} \mi{resp}_\text{cap} &=
  \ehrespWithVariable{\hresp{ nonce=n_\text{cap}, status=200,
      headers=\an{\mi{sts}}, xbody=\an{n,i}}}{k}
\end{align}
Moreover, there must exist a processing step of the following form,
where $m \leq j$, $a_r \in \mapAddresstoAP(r)$, and $x$ is some
address:
\[ s_{m-1} \xrightarrow[r \rightarrow
\{(x{:}a_r{:}\mi{resp}_\text{cap})\}]{(a_r{:}x{:}\mi{req}_\text{cap})
  \rightarrow r} s_{m}\enspace . \]

From the assumption and the definition of RPs it follows that $c$ is
of the following form:
\begin{align*}
  c &= \an{\mi{uc}, \mi{ia}}\\
  &\equiv \an{\sig{\an{i, \pub(k_u)}}{k_\text{sign}}, \sig{\an{d_r,
        \https}}{k_u}}
\end{align*}
where $k_u$ and $k_\text{sign}$ are some private keys. When we write
$i = \an{i_\text{name}, i_\text{domain}}$, we have that:
\begin{align*}
  c \equiv \an{\sig{\an{\an{i_\text{name}, i_\text{domain}},
        \pub(k_u)}}{k_\text{sign}}, \sig{\an{d_r, \https}}{k_u}}\ .
\end{align*}

As $r$ accepts the CAP $c$, we know that $\pub(k_\text{sign}) \equiv
S_j(r).\str{signkeys}[i_\text{domain}]$. As the subterm
$\str{signkeys}$ of $r$'s state is never changed, we have
$S_j(r).\str{signkeys} = S_0(r).\str{signkeys}$. With the definition
of the initial state of $r$ (See
Definition~\ref{def:relying-parties}), we have that
$\pub(k_\text{sign}) \equiv S_j(r).\str{signkeys}[i_\text{domain}]
\equiv \pub(\mapSignKey(\mapDomain^{-1}(i_\text{domain})))$.

The private key $\mapSignKey(\mapDomain^{-1}(i_\text{domain}))$ is
initially only known to the DY process $\mi{idp} :=
\mapDomain^{-1}(i_\text{domain}) = \mapGovernor(i)$. From the
assumption we know that $\mi{idp}$ is an honest IdP (and not the
attacker, a corrupted IdP, or some other DY process). As we can see in
Algorithm~\ref{alg:idp-pidp} (that defines the behavior of IdPs), the
$\mi{signkey}$ can only be used in
Line~\ref{line:usage-of-signkey-corrupt-pidp} and in
Line~\ref{line:usage-of-signkey-pidp}. We know that
Line~\ref{line:usage-of-signkey-corrupt-pidp} cannot be invoked as
long as $\mi{idp}$ is honest, which it is in $s_j$ and ever since
$s_0$. For Line~\ref{line:usage-of-signkey-pidp}, we see that the key
is not sent out to other processes. In $s_j$, the key can therefore
not have been leaked to any other DY processes.

Knowing that in or before $s_j$, only $\mi{idp}$ can derive
$k_\text{sign}$ from its knowledge, it is easy to see that only
$\mi{idp}$ can derive $\sig{x}{k_\text{sign}}$ for any $x$, and in
particular, $\mi{uc}$.

Now we want to see exactly how $\mi{idp}$ creates $\mi{uc}$ and which
data it uses in this process.

We have already seen that $\mi{idp}$ creates the $\mi{uc}$ in
Line~\ref{line:usage-of-signkey-pidp} of
Algorithm~\ref{alg:idp-pidp}. There may be more than one processing
step in $\rho$ where $\mi{idp}$ outputs $\mi{uc}$.

\begin{lemma}\label{lemma:req-uc-is-from-b-pidp}
  For all processing steps of the form
  \begin{align}\label{eq:request-beta-to-pidp}
    s_{\beta-1} \xrightarrow[\mi{idp} \rightarrow
    \{(x{:}a_{\mi{idp}}{:}\mi{resp}_\text{uc})\}]{(a_{\mi{idp}}{:}x{:}\mi{req}_\text{uc})
      \rightarrow \mi{idp}} s_{\beta}
  \end{align}
  (for some addresses $x$, $a_{\mi{idp}}$ with $s_{\beta} < s_j$,
  where $\mi{resp}_\text{uc}$ is an encrypted HTTP response with the
  body $\an{\mi{uc}}$) it holds that $\mi{req}_\text{uc}$ was emitted
  by $b$.
\end{lemma}

\begin{proof}
  To reach Line~\ref{line:usage-of-signkey-pidp} of
  Algorithm~\ref{alg:idp-pidp}, several conditions have to be met for
  $\mi{req}_\text{uc}$: It must be an encrypted HTTPS POST request
  with the path $\str{/certreq}$. The body of $\mi{req}_\text{uc}$
  must be congruent to $\an{i, \pub(k_u)}$. The request must contain a
  cookie with the name $\str{sessionid}$ and some value
  $\mi{sessionid}$. This value must be a valid key for the dictionary
  $s'.\str{sessions}$ and
  \begin{align}\label{eq:idInSessionAuthIds}
    i \inPairing s'.\str{sessions}[\mi{sessionid}]\ .
  \end{align}

  Initially, $s'.\str{sessions}$ is empty. It is only populated in
  Line~\ref{line:populate-sessions-pidp} of
  Algorithm~\ref{alg:idp-pidp}. This line must have been executed in a
  previous processing step of the following form:
  \begin{align}\label{eq:request-alpha-to-pidp}
    s_{\alpha-1} \xrightarrow[\mi{idp} \rightarrow
    \{(x{:}a_{\mi{idp}}{:}\mi{resp}_\text{auth})\}]{(a_{\mi{idp}}{:}x{:}\mi{req}_\text{auth})
      \rightarrow \mi{idp}} s_{\alpha}
  \end{align}
  (for some addresses $x$, $a_{\mi{idp}}$ with $s_{\alpha} <
  s_{\beta}$). In this step, $s'.\str{sessions}$ was populated with a
  new entry for the session id $\mi{sessionid}$.

  From Algorithm~\ref{alg:idp-pidp} we can see that
  $\mi{req}_\text{auth}$ must meet the following conditions: It must
  be an HTTPS POST request, must contain
  a specific Origin header and its body must contain a pair
  $\an{i_\text{in}, \mi{secret}_\text{in}}$ such that the id/password
  combination matches a combination stored in
  $S_{\alpha-1}(\mi{idp}).\str{users}$. As we have that
  $S_{\alpha-1}(\mi{idp}).\str{users} = S_{0}(\mi{idp}).\str{users}$
  and with the initial definition
  \begin{align}
    S_{0}(\mi{idp}).\str{users} = \an{\{\an{s, \an{
          \mathsf{IDsofSecret}(s) }} | \RPSecrets^i\}}
  \end{align}
  we can see that $i_\text{in} \in
  \mathsf{IDsofSecret}(\mi{secret}_\text{in})$. As the list of
  authenticated ids in the session is then (in
  Line~\ref{line:populate-sessions-pidp} of
  Algorithm~\ref{alg:idp-pidp}) populated with
  $\mathsf{IDsofSecret}(\mi{secret}_\text{in})$ and with
  (\ref{eq:idInSessionAuthIds}) we have that $\mi{i} \in
  \mathsf{IDsofSecret}(\mi{secret}_\text{in})$.  Now,
  $\mathsf{IDsofSecret}$ assigns the IDs to their secrets according to
  $\mathsf{secretOfID}$, i.e., it must hold that
  \begin{align}\label{eq:idMatchesSecret}
    \mapIDtoPLI(\mi{i}) = \mi{secret}_\text{in}\ .
  \end{align}
  This secret can be owned by at most one browser, and according to
  the definitions of the initial knowledge of the DY processes in
  \ref{sec:analysisbrowserid-pidp}, it is initially only known to the
  owner of the secret $\mapPLItoOwner(\mi{secret}_\text{in})$ (see
  Section~\ref{sec:browsers-pidp}) and to one specific IdP (see
  Section~\ref{app:idps}), in this case $\mi{i}_\text{domain} \in
  \mapDomain(\mi{idp})$ (because otherwise, $\mi{idp}$ would not
  accept this ID).

  From Algorithm~\ref{alg:idp-pidp} we can see that the IdP never uses
  this secret to create messages as long as it is honest, which it is
  by precondition.

  With (\ref{eq:idMatchesSecret}) we see that initially, only
  $\mapPLItoOwner(\mapIDtoPLI(\mi{i})) = \mapIDtoOwner(\mi{i})$ knows
  the secret $\mi{secret}_\text{in}$, which, by assumption, is not
  fully corrupted in $s_j$, and thus, with the request order given for
  (\ref{eq:request-beta-to-pidp}) and (\ref{eq:request-alpha-to-pidp})
  is not fully corrupted in $s_\alpha$. (Once fully corrupted,
  browsers stay fully corrupted.)

  (*): Honest browsers release secrets only to scripts that are loaded
  from a specific origin. In this case, according to the initial state
  given in Section~\ref{sec:browsers-pidp}, the secret
  $\mapIDtoPLI(\mi{i})$ is only released to scripts from the origin
  $\an{i_\text{domain}, \https}$. For any such script (or document),
  with Lemma~\ref{lemma:https-document-origin-general} and the
  definition of the browser's key mapping in
  Section~\ref{sec:browsers-pidp}, we can see that any script that has
  access to the secret was sent by $\mi{idp}$. This DY process is also
  the governor of $i$, which is, by assumption, not
  corrupted. Therefore, $\mi{idp}$ can only deliver either the script
  $\mi{script\_idp\_pif}$ or the script $\mi{script\_idp\_ad}$. We can
  now check, that both scripts, running in a browser, never send this
  secret to any other DY process than $\mi{idp}$, and trigger only
  encrypted requests to do so.

  In $\mi{script\_idp\_pif}$ (Algorithm~\ref{alg:scriptidppif-pidp}),
  the subterm $\mi{secret}$ of the state is not used at all;
  therefore, the script triggers no outgoing message containing the
  secret at all.

  In $\mi{script\_idp\_ad}$ (Algorithm~\ref{alg:scriptidpad-pidp}),
  $\mi{secret}$ is only used as a part of a an HTTP request to the
  document's own origin (which therefore is the origin for which the
  secret is stored in the browser's list of secrets, which therefore
  must be $\an{i_\text{domain}, \https}$). The request's data is not
  stored in the script's state.

  We now know that all entities that have access to $\mi{secret}$ (the
  browser $b$ and the IdP $\mi{idp}$) never leak it. As $\mi{idp}$
  never creates any HTTP(S) requests, $b$ must have created
  $\mi{req}_\text{auth}$ before the processing step $s_{\alpha-1}
  \rightarrow s_{\alpha}$.

  In this processing step, $\mi{idp}$ creates a new session id
  ($\mi{sessionid}$). This id is sent out only once (in
  Line~\ref{line:send-auth-response-pidp} of
  Algorithm~\ref{alg:idp-pidp}), which, in our case, is
  $\mi{resp}_\text{auth}$. With
  Corollary~\ref{cor:k-does-not-leak-from-honest-browser-general} we
  can see that from this (encrypted) response $\mi{resp}_\text{auth}$,
  only $b$ can derive the contents, especially the contents of the
  $\cSetCookie$ header. As in $b$, the cookie is stored as a
  \emph{secure}, \emph{HTTP only} cookie, $b$ releases the contents of
  this cookie only as a $\str{Cookie}$ header to the origin
  $\an{\mi{i}_\text{domain}, \https}$. Given the keymapping in $b$'s
  state, requests to this origin are handled by $\mi{idp}$, and with
  Algorithm~\ref{alg:idp-pidp} it is easy to see that the
  $\str{Cookie}$ header is only used for validating the UC request,
  but is not used anywhere else. All in all, $b$ and $\mi{idp}$ do not
  leak the session id $\mi{sessionid}$.

  As $\mi{sessionid}$ is an important part of $\mi{req}_\text{uc}$, we
  can see that this request must have been emitted by $b$. \qed
\end{proof}

\begin{lemma}\label{lemma:k-is-from-b-pidp}
  The secret key $k_u$ was chosen by the browser $b$ from its own
  nonces, i.e., $k_u \subset N^b$.
\end{lemma}

\begin{proof}
  First of all, we know that for $\mi{idp}$ to generate $\mi{uc}$,
  there must be a processing step in $\rho$ of the form (described in
  Lemma~\ref{lemma:req-uc-is-from-b-pidp}):
  \begin{align}
    s_{\beta-1} \xrightarrow[\mi{idp} \rightarrow
    \{(x{:}a_{\mi{idp}}{:}\mi{resp}_\text{uc})\}]{(a_{\mi{idp}}{:}x{:}\mi{req}_\text{uc})
      \rightarrow \mi{idp}} s_{\beta}
  \end{align}
  (for some addresses $x$, $a_{\mi{idp}}$ with $s_{\beta} < s_j$,
  where $\mi{resp}_\text{uc}$ is an encrypted HTTP response with the
  body $\an{\mi{uc}}$). For the request $\mi{req}_\text{uc}$, the
  method must be $\mPost$ and the path component must be
  $\str{/certreq}$.

  With Lemma~\ref{lemma:req-uc-is-from-b-pidp} we know that
  $\mi{req}_\text{uc}$ was emitted by $b$, which is honest at this
  point in the run. With the same arguments as in (*) we can see that
  either $\mi{script\_idp\_pif}$ or the script $\mi{script\_idp\_ad}$
  initiated $\mi{req}_\text{uc}$.

  For $\mi{script\_idp\_ad}$ it is easy to see that this script never
  sends a POST request to $\mi{idp}$.

  The script $\mi{script\_idp\_pif}$ can only send a POST request to
  $\str{/certreq}$ in Line~\ref{line:send-req-uc-pidp} of
  Algorithm~\ref{alg:scriptidppif-pidp}. In this case, the public key
  is chosen from the subterm $\str{pubkeys}$ of the script's
  state. This subterm is only populated in
  Line~\ref{line:populate-pubkeys-pidp} of
  Algorithm~\ref{alg:scriptidppif-pidp}. It can only be populated by a
  postMessage $\mi{pm}$ from an immediate parent window and from the
  origin $\an{\mapDomain(\fAP{LPO}), \https}$ (given how a browser
  checks and transmits postMessages, see
  Line~\ref{line:append-pm-to-scriptinput-condition}f.\ of
  Algorithm~\ref{alg:runscript}). Further, the message in $\mi{pm}$
  must be of the form $\an{n, \pub(k_u)}$ where $n$ is a nonce that
  was freshly chosen for a $\an{\str{genKeyPair}, n}$ postMessage in
  Line~\ref{line:start-genkeypair-pidp} of
  Algorithm~\ref{alg:scriptidppif-pidp}.

  Given that $b$'s keymapping assigns the private key of LPO to the
  domain of LPO and with Lemma~\ref{lemma:https-script-origin-general}
  we see that the only scripts that can send such a postMessage are
  $\mi{script\_lpo\_cif}$ and $\mi{script\_lpo\_ld}$.

  In the script $\mi{script\_lpo\_cif}$
  (Algorithm~\ref{alg:scriptlpocif-pidp}), postMessages of the form of
  $\mi{pm}$ can only be sent in
  Line~\ref{line:send-pm-with-nonce-cif-key-pidp} (the message sent in
  Line~\ref{line:send-pm-with-nonce-cif-email-pidp} would not carry
  the correct nonce for a response to a $\str{genKeyPair}$ message).
  
  The same holds true for the script $\mi{script\_lpo\_ld}$
  (Algorithm~\ref{alg:scriptlpocif-pidp}).

  Therefore, the key $k_u$ is a nonce that was chosen from the
  browser's nonces. \qed
\end{proof}

\begin{lemma}\label{lemma:ku-does-not-leak-from-b-pidp}
   $k_u$ does not leak from $b$.
\end{lemma}

\begin{proof}
  As we have seen above, the key $k_u$ was chosen either in the script
  $\mi{script\_lpo\_cif}$ or in the script $\mi{script\_lpo\_ld}$
  running in the honest browser $b$.
  
  In both scripts, any nonce that is chosen from the script's
  $\mi{nonces}$ will not be given to the script (as part of
  $\mi{nonces}$) by the browser again, thus, the nonce was chosen
  freshly. Further, the nonce is stored in the subterm $\str{key}$ of
  the script's state and (besides the derivation of the public key) is
  only used to sign IAs.

  There are no other scripts running in the origin of
  $\an{\mapDomain(\fAP{LPO}), \https}$. The (honest) browser $b$ does
  not leak the script's state. Therefore, $k_u$ does not leak from
  $b$. \qed
\end{proof}

With Lemma~\ref{lemma:req-uc-is-from-b-pidp},
\ref{lemma:k-is-from-b-pidp}, and
\ref{lemma:ku-does-not-leak-from-b-pidp}, we can see that only $b$
knows $k_u$ and the attacker cannot know $k_u$. Therefore, only $b$
can create the $\mi{ia} = \sig{\an{d_r, \https}}{k_u}$. As $k_u$ is
only accessible to scripts with the origin $\an{\mapDomain(\fAP{LPO}),
  \https}$, only the script $\mi{script\_lpo\_cif}$ or the script
$\mi{script\_lpo\_ld}$ can create $\sig{\an{d_r, \https}}{k_u}$. In
both scripts, after creation, $\mi{ia}$ is sent in postMessage only to
scripts that have the origin for which $\mi{ia}$ was created
($=\an{d_r, \https}$). With
Lemma~\ref{lemma:https-script-origin-general} and the definition of
relying parties (see Algorithm~\ref{alg:rp-pidp}) we see, that the
only potential receiver is $\mi{script\_rp\_index}$.

After receiving this $\str{response}$ postMessage,
$\mi{script\_rp\_index}$ stores the UC and the IA in the subterm
called $\str{cap}$ of its scriptstate (see
Algorithm~\ref{alg:scriptrpindex-pidp},
Line~\ref{line:set-cap-pidp}). After doing so, this subterm is read
only in Line~\ref{line:use-cap-1-pidp} (where only the identity is
extracted) and in Line~\ref{line:use-cap-2-pidp}. There, the $\mi{ia}$
is sent to $r$ (in the encrypted request $\mi{req}_\text{cap}$).

The RP $r$, which is not corrupted, and the browser $b$ do not leak
$\mi{ia}$. After receiving $\mi{ia}$, $r$ sends the newly created
service token $\an{n,i}$ to $b$, which ignores it (see
Algorithm~\ref{alg:scriptrpindex-pidp}
Line~\ref{line:ignore-service-token-pidp}f.). Therefore, $b$ and $r$ do
not leak $\an{n,i}$.

Therefore, the attacker cannot know $\an{n,i}$ in $S_j$, i.e.,
$\an{n,i} \not\in d_{N^\fAP{attacker}}(S_j(\fAP{attacker}))$. This is
a contradiction to our assumption.  \qed

\subsection{Condition B}

Similar to before, we assume that Condition B does not hold and lead
this to a contradiction. We therefore make the following assumption:
There is a run $\rho$ of $\bidwebsystem$, some state $s_j = (S_j,
E_j)$ in $\rho$, some $r\in \fAP{RP}$ that is honest in $S_j$, some RP
service token of the form $\an{n,i}$ recorded in $r$ in the state
$S_j(r)$, the request corresponding to $\an{n,i}$ was sent by some
$b\in \fAP{B}$ which is honest in $S_j$, and $b$ does not own $i$.

By definition of RPs, for $\an{n,i}$ there exists a
corresponding HTTPS request received by $r$, which we call
$\mi{req}_\text{cap}$, and a corresponding response
$\mi{resp}_\text{cap}$. The request must contain a valid CAP $c$ and
must have been sent by some atomic process $p$ to $r$. The response
must contain $\an{n,i}$ and it must be encrypted by some symmetric
encryption key $k$ sent in $\mi{req}_\text{cap}$.

In particular, it follows that the request and the response must be of
the following form, where $d_r \in \mathsf{dom}(r)$ is the domain of
$r$, $n_\text{cap}, k \in \nonces$ are some nonces, $\mi{path}$,
$\mi{params} \in \terms$, $c$ is some valid CAP, and $\mi{sts}$ is the
Strict-Transport-Security header (as in the definition of RP's
relation):

\begin{align}
  \label{eq:proofreqcapB-pidp} \mi{req}_\text{cap} &=
  \mathsf{enc}_\mathsf{a}(\langle \hreq{ nonce=n_\text{cap}, method=\mPost,
      xhost=d_r, path=\mi{path}, parameters=\mi{params}, headers=[\str{Origin}:
      \an{d_r, \https}], xbody=c},\nonumber\\ &\hspace{4.5em}k\rangle, \pub(\mapKey(d_r)))\\
  \label{eq:proofrespcaB-pidp} \mi{resp}_\text{cap} &=
  \ehrespWithVariable{\hresp{ nonce=n_\text{cap}, status=200,
      headers=\an{\mi{sts}}, xbody=\an{n,i}}}{k}
\end{align}
Moreover, there must exist a processing step of the following form,
where $m \leq j$, $a_r \in \mapAddresstoAP(r)$, and $x$ is some
address:
\[ s_{m-1} \xrightarrow[r \rightarrow
\{(x{:}a_r{:}\mi{resp}_\text{cap})\}]{(a_r{:}x{:}\mi{req}_\text{cap})
  \rightarrow r} s_{m}\enspace . \]

From the assumption and the definition of RPs it follows that $c$ is
of the following form:
\begin{align*}
  c &= \an{\mi{uc}, \mi{ia}}\\
  &\equiv \an{\sig{\an{i, \pub(k_u)}}{k_\text{sign}}, \sig{\an{d_r,
        \https}}{k_u}}
\end{align*}
where $k_u$ and $k_\text{sign}$ are some private keys. When we write
$i = \an{i_\text{name}, i_\text{domain}}$, we have that:
\begin{align*}
  c \equiv \an{\sig{\an{\an{i_\text{name}, i_\text{domain}},
        \pub(k_u)}}{k_\text{sign}}, \sig{\an{d_r, \https}}{k_u}}\ .
\end{align*}

With Lemma~\ref{lemma:https-script-origin-general} we see that this
request was initiated by a script that $b$ extracted from an HTTPS
response by $r$. The only script that $r$ sends in its responses is
$\mi{script\_rp\_index}$.

In this script (Algorithm~\ref{app:scriptrpindex-pidp}), the
only place where a request is initiated is in
Line~\ref{line:send-req-uc-pidp}. We can see that the cap $c$ is taken
from the script's state, i.e., $s'.\mi{cap} \equiv c$ before the
execution of Line~\ref{line:send-req-uc-pidp} must hold. Initially,
this term is empty, therefore the value must have been set during the
prior execution of the script. This happens in
Line~\ref{line:set-cap-pidp} and in Line~\ref{line:set-cap-2-pidp} of
the algorithm. For both lines to be executed, there must arrive a
postMessage at $\mi{script\_rp\_index}$ (either a $\mathtt{login}$ or
a $\mi{response}$ postMessage) from the origin of LPO. 

With Lemma~\ref{lemma:https-document-origin-general},
Lemma~\ref{lemma:https-script-origin-general}, and the definition of
the web browser, we can see that the message must indeed come from one
of LPO's scripts, that is, either $\mi{script\_lpo\_ld}$ or
$\mi{script\_lpo\_cif}$. Before we proceed by showing that both
scripts never send a UC for an identity that is not owned by browser
$b$ to the script $\mi{script\_rp\_index}$ (and later to $r$), we
first proof the following lemma:

\begin{lemma}\label{app:lemma-value-of-semail}
  The value of $s'.\str{email}$ in $\mi{script\_lpo\_ld}$ is always
  either one of the browser's identities or empty.
\end{lemma}
\begin{proof}

  We show this by induction:

  \emph{Base case:} The value of $s'.\str{email}$ is initially empty
  (see initial scriptstate).

  \emph{Induction step:} The value is set only in
  Lines~\ref{line:set-semail-ld-pidp-1}
  and~\ref{line:set-semail-ld-pidp-2}. In the first case, the identity
  is chosen non-deterministically from the browser's identities
  $\mi{ids}$, which are the identities that the browser owns (see
  Section~\ref{sec:browsers-pidp}). 

  In the second case, the value of $s'.\str{email}$ is taken from the
  localStorage, with the help of the key $\mi{idpnonce}$ that is taken
  from the sessionStorage. We can now show that what is retrieved from
  the localStorage is either empty or a previous value of
  $s'.\str{email}$:

  First, we show that the value of $\mi{idpnonce}$, taken from
  sessionStorage in Line~\ref{line:get-idpnonce-ld-pidp}, is always a
  nonce or empty: The browser's sessionStorage is separated by origins
  (and root windows), and therefore, only scripts under the origin of
  LPO have read or write access. Thus, the only two scripts that can
  possibly write the $\mi{idpnonce}$ value are $\mi{script\_lpo\_cif}$
  and $\mi{script\_lpo\_ld}$. The script $\mi{script\_lpo\_cif}$ does
  not write to sessionStorage. The script $\mi{script\_lpo\_ld}$ only
  writes to sessionStorage in
  Line~\ref{line:write-sessionstorage-idpnonce-ld-pidp}. It only
  writes a fresh nonce (chosen in
  Line~\ref{line:chose-idpnonce-ld-pidp}). Therefore, the value of
  $\mi{idpnonce}$ is always a nonce (or empty).

  As we are already in the second case of the if-statement in
  Line~\ref{line:check-idpnonce-ld-pidp} (we know that
  Line~\ref{line:set-semail-ld-pidp-2} was executed) $\mi{idpnonce}$
  cannot be empty and must be a nonce.

  Now, we can show that $\mi{localStorage}[\mi{idpnonce}]$ is either
  empty or a previous value of $s'.\str{email}$: The browser's
  localStorage is separated by origins, and therefore, only scripts
  under the origin of LPO have read or write access. As above, the
  only two scripts that can write values to the localStorage are
  $\mi{script\_lpo\_cif}$ and $\mi{script\_lpo\_ld}$. The script
  $\mi{script\_lpo\_cif}$ does not write to localStorage (it only
  removes subterms form localStorage in
  Line~\ref{line:remove-siteinfo-cif-pidp}). We can thus focus on
  $\mi{script\_lpo\_ld}$.

  There are two lines where this script writes to the localStorage:
  Lines~\ref{write-localstorage-siteinfo-ld-pidp}
  and~\ref{write-localstorage-idpnonce-ld-pidp}. We can safely ignore
  the first case, as it does not use a nonce as a key (but the fixed
  string $\str{siteInfo}$ instead). In the latter case, it writes a
  value of $s'.\str{email}$.

  This concludes the induction. \qed
\end{proof}

We can now show (for both scripts), that they never send a UC for an
identity that is not owned by the browser $b$:

\textbf{(I)} For $\mi{script\_lpo\_ld}$
(Algorithm~\ref{alg:scriptlpold-pidp}), it is easy to see that the UC
that is finally used to create a CAP for RP in
Line~\ref{line:use-uc-ld-pidp} is set in
Line~\ref{line:set-uc-ld-pidp}. There, the identity in the UC is
checked against the identity in $s'.\str{email}$ in the script's state
(and it is checked that $s'.\str{email}$ is not empty).

With Lemma~\ref{app:lemma-value-of-semail} and the observations above
we can conclude that in $\mi{script\_lpo\_ld}$, it is not possible
that a UC for an identity that the browser does not own is
accepted. Therefore, the UC that is sent to $\mi{script\_rp\_index}$
is issued for an identity of the browser $b$.

\textbf{(II)} For $\mi{script\_lpo\_cif}$
(Algorithm~\ref{alg:scriptlpocif-pidp}), it is easy to see that the UC
that is finally used in Line~\ref{line:use-uc-cif-pidp} is set in
Line~\ref{line:set-uc-cif-pidp}. There, the identity in the UC is
checked against the value of $s'.\str{email}$ (and, that
$s'.\str{email}$ is not empty). Initially, $s'.\str{email}$ is
empty. It is set only in Line~\ref{line:set-semail-cif-pidp}. There,
it is taken from the localStorage, using the key $\str{siteInfo}$. As
we have seen above, the only place where values are stored using this
key is in Line~\ref{write-localstorage-siteinfo-ld-pidp} of
Algorithm~\ref{alg:scriptlpold-pidp}. There, it is taken from the
script's $s'.\str{email}$, which, according to
Lemma~\ref{app:lemma-value-of-semail}, is either empty or one of the
browser's identities. Note that the value of $\str{siteInfo}$ is a
dictionary. The keys which are used inside of this dictionary are not
relevant here, but only the values.

Thus, in $\mi{script\_lpo\_cif}$, it is not possible that a UC for an
identity that the browser does not own is accepted. Therefore, the UC
that is sent to $\mi{script\_rp\_index}$ is issued for an identity of
the browser $b$.

With (I) and (II), we see that all UCs that are sent to
$\mi{script\_rp\_index}$ (and later to $r$) are issued for identities
of the browser $b$. This contradicts the assumption, which
proves that Condition B holds true.\qed

%% file: figure-pidp-detailed.tex
\begin{figure}[h!]\centering
\begin{tikzpicture}[scale=0.9, every node/.style={scale=0.9}]
\scriptsize

\matrix[column sep={5.35pc,between origins}, row sep=1.8ex] {
\node[anchor=base,fill=Gainsboro,rounded corners](lpo){LPO}; & \node[anchor=base,fill=Gainsboro,rounded corners](idp){IdP}; & \node[draw,anchor=base](rpdoc){RP-Doc}; && \node(ld-top){};                 & \node(pif-top){}; \\
                                   &                                    & \node(rp-doc-iframe-cif){}; & \node[draw,anchor=base](cif-iframe){CIF} ; &&\\
 \node(lpo-cif-init){};            &                                    &  &  \node(cif-init){}; && \\

                                   &                                    & \node(rp-doc-cif-rdy){}; & \node(cif-rdy){}; & & \\ 
                                   &                                    & \node(rp-doc-cif-ld){}; & \node(cif-ld){}; & & \\ 
\node(lpo-cif-ctx){};              &                                    & \node(cif-ctx){}; & & &  \\
\node(lpo-cif-addr-info-1-top){};  &                                    &  \node(cif-lpo-addr-info-1-top){};                                        && & \\ [1pt]
\node(lpo-idp-wk-1){};             & \node(idp-lpo-wk-1){};             &                                                                          && & \\
\node(lpo-cif-addr-info-1-btm){};  &                                    &  \node(cif-lpo-addr-info-1-btm){};                                        && & \\

                                   &                                    &                                   &\node(cif-pif-open){};               &   & \node(pif-0)[draw,anchor=base]{PIF}; \\
                                   & \node(idp-pif-init-0){};           &                                   &                                     &   & \node(pif-init-0){}; \\ [1pt]
                                   &                                    &                                   &\node(cif-pif-ping-0){};             &   & \node(pif-cif-ping-0){}; \\
                                   &                                    &                                   &\node(cif-pif-pong-0){};             &   & \node(pif-cif-pong-0){}; \\
                                   &                                    &                                   &\node(cif-pif-bprov-0){};            &   & \node(pif-cif-bprov-0){}; \\
                                   &                                    &                                   &\node(cif-pif-certinfo-0){};         &   & \node(pif-cif-certinfo-0){}; \\
                                   &                                    &                                   &\node(cif-pif-prov-fail-0){};        &   & \node(pif-cif-prov-fail-0){}; \\
                                   &                                    &                                   &\node(cif-pif-close){};              &   & \node(pif-end-0)[draw,anchor=base]{/PIF}; \\

                                   &                                    & \node(rp-doc-logout){}; & \node(cif-logout){}; & & \\
\node(phase-ld-open-1-left){};     &                                    &                                         &&                                  & \node(phase-ld-open-1-right){}; \\ [2pt]

                                   &                                    & \node(rp-doc-dlg-run){}; & \node(cif-dlg-run){}; & & \\
                                   &                                    & \node(rpdoc-ld-open){};                 && \node(ld)[draw,anchor=base]{LD}; & \\

\node(phase-ld-init-1-left){};     &                                    &                                         &&                                  & \node(phase-ld-init-1-right){}; \\ [2pt]
\node(lpo-ld-init-1){};            &                                    &                                         && \node(ld-init-1){};              & \\ [1pt]
                                   &                                    & \node(rpdoc-ld-recv-ready){};           && \node(ld-rpdoc-send-ready){};    & \\
                                   &                                    & \node(rpdoc-ld-send-request){};         && \node(ld-rpdoc-recv-request){};  & \\
\node(lpo-ld-ctx-1){};             &                                    &                                         && \node(ld-lpo-ctx-1){};           & \\
                                   &                                    &                                         && \node(ld-user-email){};          & \\
\node(lpo-ld-addr-info-1){};   &                                    &                                         && \node(ld-lpo-addr-info-1){}; & \\
\node(phase-prov-left){};          &                                    &                                         &&                                  & \node(phase-prov-right){}; \\
                                   &                                    &                                         && \node(ld-pif-open){};            & \node(pif)[draw,anchor=base]{PIF}; \\
                                   & \node(idp-pif-init-1){};           &                                         &&                                  & \node(pif-init-1){}; \\ [1pt]
                                   &                                    &                                         && \node(ld-pif-ping-1){};          & \node(pif-ld-ping-1){}; \\
                                   &                                    &                                         && \node(ld-pif-pong-1){};          & \node(pif-ld-pong-1){}; \\
                                   &                                    &                                         && \node(ld-pif-bprov-1){};          & \node(pif-ld-bprov-1){}; \\
                                   &                                    &                                         && \node(ld-pif-certinfo-1){};          & \node(pif-ld-certinfo-1){}; \\
                                   &                                    &                                         && \node(ld-pif-prov-fail-1){};          & \node(pif-ld-prov-fail-1){}; \\
                                   &                                    &                                         && \node(ld-pif-close){};           & \node[draw,anchor=base](pif-end){/PIF}; \\
\node(phase-auth-left){};          &                                    &                                         &&                                  & \node(phase-auth-right){}; \\
                                   &                                    &                                         && \node[draw,fill=Gainsboro]{redirect to AD}; & \\
                                   & \node(idp-ld-auth){};              &                                         && \node(ld-auth){}; & \\
\node(lpo-end){};                  & \node(idp-end){};                  & \node(rpdoc-end){};                     & \node(cif-end){}; & \node(ld-end){};              &  \\
\node(lpo-fade){};                 & \node(idp-fade){};                 & \node(rpdoc-fade){};                    & \node(cif-fade){}; &  \node(ld-fade){}; & \node(pif-col-end){};\\
};

\tikzstyle{xhrArrow} = [color=blue,decoration={markings, mark=at position 1 with {\arrow[color=blue]{triangle 45}}}, preaction = {decorate}]

   \draw [->,snake=snake,segment amplitude=0.2ex] (rp-doc-iframe-cif.40) to node [above=-2pt] { create} (cif-iframe);

  \draw [->] (cif-init.160) to node [above=-1.5pt]{ GET CIF } (lpo-cif-init.20);
  \draw [->] (lpo-cif-init.340) -- (cif-init.200);

  \draw [->,color=red,dashed] (cif-rdy) to node [above=-3pt]{ ready } (rp-doc-cif-rdy);

  \draw [->,color=red,dashed] (rp-doc-cif-ld) to node [above=-2pt]{ loaded } (cif-ld);

  \draw [->,color=blue,>=latex] (cif-ctx.160) to node [above=-2pt]{ GET session\_context } (lpo-cif-ctx.20);
  \draw [->,color=blue,>=latex] (lpo-cif-ctx.340) -- (cif-ctx.200);
\draw [->,color=blue,>=latex] (cif-lpo-addr-info-1-top) to node [above=-2pt]{ GET address\_info} (lpo-cif-addr-info-1-top);
\draw [->] (lpo-idp-wk-1.20) to node [above=-1.5pt]{ GET wk} (idp-lpo-wk-1.160);
\draw [->] (idp-lpo-wk-1.200) -- (lpo-idp-wk-1.340);
\draw [->,color=blue,>=latex] (lpo-cif-addr-info-1-btm) to node [above=-2pt]{} (cif-lpo-addr-info-1-btm);

\draw [->,snake=snake,segment amplitude=0.2ex] (cif-pif-open.40) to node [above=-2pt] { create} (pif-0);
\draw [->] (pif-init-0.160) to node [above=-1.5pt]{ GET PIF} (idp-pif-init-0.20);
\draw [->] (idp-pif-init-0.340) -- (pif-init-0.200);
    
\draw [->,color=red,dashed] (pif-cif-ping-0) to node [above=-2pt]{ ping } (cif-pif-ping-0);
\draw [->,color=red,dashed] (cif-pif-pong-0) to node [above=-2pt]{ pong } (pif-cif-pong-0);
\draw [->,color=red,dashed] (pif-cif-bprov-0) to node [above=-2pt]{ beginProvisioning  } (cif-pif-bprov-0);
\draw [->,color=red,dashed] (cif-pif-certinfo-0) to node [above=-2pt]{ email, certDuration } (pif-cif-certinfo-0);
\draw [->,color=red,dashed] (pif-cif-prov-fail-0) to node [above=-2pt]{ raiseProvisioningFailure } (cif-pif-prov-fail-0);
\draw [->,snake=snake,segment amplitude=0.2ex] (cif-pif-close.40) to node [above=-2pt] { close} (pif-end-0);

  \draw [->,color=red,dashed] (cif-logout) to node [above=-3pt]{ logout } (rp-doc-logout);

  \draw [->,color=red,dashed] (rp-doc-dlg-run) to node [above=-3pt]{ dlgRun } (cif-dlg-run);

\draw [->,snake=snake,segment amplitude=0.2ex] (rpdoc-ld-open.40) to node [above=-2pt] { open} (ld);

\draw [->] (ld-init-1.160) to node [above=-1.5pt]{ GET LD} (lpo-ld-init-1.20);
\draw [->] (lpo-ld-init-1.340) -- (ld-init-1.200);

\draw [->,color=red,dashed] (ld-rpdoc-send-ready) to node [above=-2pt]{ ready} (rpdoc-ld-recv-ready);

\draw [->,color=red,dashed] (rpdoc-ld-send-request) to node [above=-2pt]{ request} (ld-rpdoc-recv-request);

\draw [->,color=blue,>=latex] (ld-lpo-ctx-1.160) to node [above=-2pt]{ GET session\_context} (lpo-ld-ctx-1.20);
\draw [->,color=blue,>=latex] (lpo-ld-ctx-1.340) -- (ld-lpo-ctx-1.200);

\node (ld-user-email-drawn) at (ld-user-email) [draw,rounded corners,fill=Gainsboro]{ email address };

\draw [->,color=blue,>=latex] (ld-lpo-addr-info-1.160) to node [above=-2pt]{ GET address\_info} (lpo-ld-addr-info-1.20);
\draw [->,color=blue,>=latex] (lpo-ld-addr-info-1.340) to node [above=-2pt]{} (ld-lpo-addr-info-1.200);

\draw [->,snake=snake,segment amplitude=0.2ex] (ld-pif-open.40) to node [above=-2pt] { create} (pif);

\draw [->] (pif-init-1.160) to node [above=-1.5pt]{ GET PIF} (idp-pif-init-1.20);
\draw [->] (idp-pif-init-1.340) -- (pif-init-1.200);

\draw [->,color=red,dashed] (pif-ld-ping-1) to node [above=-2pt]{ ping } (ld-pif-ping-1);
\draw [->,color=red,dashed] (ld-pif-pong-1) to node [above=-2pt]{ pong } (pif-ld-pong-1);
\draw [->,color=red,dashed] (pif-ld-bprov-1) to node [above=-2pt]{ beginProvisioning  } (ld-pif-bprov-1);
\draw [->,color=red,dashed] (ld-pif-certinfo-1) to node [above=-2pt]{ email, certDuration } (pif-ld-certinfo-1);
\draw [->,color=red,dashed] (pif-ld-prov-fail-1) to node [above=-2pt]{ raiseProvisioningFailure } (ld-pif-prov-fail-1);

\draw [->,snake=snake,segment amplitude=0.2ex] (ld-pif-close.40) to node [above=-2pt] { close} (pif-end);

\node (ld-auth-drawn) at (ld-auth) [draw,fill=Gainsboro,rounded corners]{ auth IdP};
\draw [implies-implies,double] (ld-auth-drawn) -- (idp-ld-auth);

\begin{pgfonlayer}{background}
 \node (rpdoc-a) [above of=rpdoc, node distance=2ex]{};
 \node (rpdoc-al) [left of=rpdoc-a, node distance=6ex]{};
 \node (pif-col-end-b) [below of=pif-col-end, node distance=2ex]{};
 \node (pif-col-end-br) [right of=pif-col-end-b, node distance=4ex]{};
 \filldraw [color=Gainsboro,rounded corners] (rpdoc-al) rectangle (pif-col-end-br);

 \draw [color=gray] (lpo.270)  -- (lpo-end);
 \draw [color=gray] (idp.270)  -- (idp-end);
 \draw [color=gray] (rpdoc.270) -- (rpdoc-end);
 \draw [color=gray] (cif-iframe.270) -- (cif-end);
 \draw [color=gray] (ld.270) -- (ld-end);
 \draw [color=gray] (pif-0.270) -- (pif-end-0);
 \draw [color=gray] (pif.270) -- (pif-end);

 \draw [color=gray,dotted] (lpo-fade)  -- (lpo-end.90);
 \draw [color=gray,dotted] (idp-fade)  -- (idp-end.90);
 \draw [color=gray,dotted] (rpdoc-fade) -- (rpdoc-end.90);
 \draw [color=gray,dotted] (cif-fade) -- (cif-end.90);
 \draw [color=gray,dotted] (ld-fade) -- (ld-end.90);
\end{pgfonlayer}

\draw [dashed] (phase-ld-open-1-left.180) -- (phase-ld-open-1-right.0);
\draw [dashed] (phase-ld-init-1-left.180) -- (phase-ld-init-1-right.0);
\draw [dashed] (phase-prov-left.180) -- (phase-prov-right.0);
\draw [dashed] (phase-auth-left.180) -- (phase-auth-right.0);

\node at ($(ld-top)!0.7!(pif-top)$) {Browser};

\end{tikzpicture}

\raisebox{0.5ex}{\tikz{\draw [->] (0,0) -- (0.4,0);}}
   HTTPS messages, \raisebox{0.5ex}{\tikz{\draw
       [->,color=blue,>=latex] (0,0) -- (0.4,0);}} \xhrs (over HTTPS),
   \raisebox{0.5ex}{\tikz{\draw [->,color=red,dashed] (0,0) --
       (0.4,0);}} \pms, \raisebox{0.5ex}{\tikz{\draw
       [->,snake=snake,segment length=2ex,segment amplitude=0.2ex] (0,0) -- (0.4,0);}}
   browser commands
 \captcont[foo2]{BrowserID primary mode typical login flow
   overview {\bf (part 1 of 2)}.}
\label{fig:browserid-lowlevel-pidp-detailed-1}

\end{figure}

\begin{figure}[p!]\centering
\begin{tikzpicture}[scale=0.9, every node/.style={scale=0.9}]
\scriptsize

\matrix[column sep={5.35pc,between origins}, row sep=1.8ex] {
\node[anchor=base,fill=Gainsboro,rounded corners](lpo){LPO}; & \node[anchor=base,fill=Gainsboro,rounded corners](idp){IdP}; & \node[draw,anchor=base](rpdoc){RP-Doc}; & \node[draw,anchor=base](cif-top){CIF}; & \node[draw,anchor=base](ld-top){AD};                 & \node(pif-top){}; \\
\node(lpo-fade){};                 & \node(idp-fade){};                 & \node(rpdoc-fade){};                    & \node(cif-fade){}; &  \node(ld-fade){}; & \\
                                   & \node(idp-ld-auth){};              &                                         && \node(ld-auth){}; & \\
                                   &                                    &                                         && \node[draw,fill=Gainsboro]{redirect to LD}; & \\
\node(phase-ld-init-2-left){};     &                                    &                                         &&                                  & \node(phase-ld-init-2-right){}; \\
\node(lpo-ld-init-2-top){};        & \node(idp-ld-init-2-top){};        & \node(rpdoc-ld-init-2-top){};           && \node(ld-ld-init-2-top){};       & \node(pif-ld-init-2-top){}; \\
\node(lpo-ld-init-2){};            &                                    &                                         && \node(ld-init-2){};              & \\ [1pt]
                                   &                                    & \node(rpdoc-ld-recv-ready-2){};           && \node(ld-rpdoc-send-ready-2){};    & \\
                                   &                                    & \node(rpdoc-ld-send-request-2){};         && \node(ld-rpdoc-recv-request-2){};  & \\
\node(lpo-ld-ctx-2){};             &                                    &                                         && \node(ld-lpo-ctx-2){};           & \\
\node(lpo-ld-addr-info-2){};   &                                    &                                         && \node(ld-lpo-addr-info-2){}; & \\
\node(lpo-ld-init-2-btm){};        & \node(idp-ld-init-2-btm){};        & \node(rpdoc-ld-init-2-btm){};           && \node(ld-ld-init-2-btm){};       & \node(pif-ld-init-2-top){}; \\
\node(phase-prov-cont-left){};     &                                    &                                         &&                                  & \node(phase-prov-cont-right){}; \\
                                   &                                    &                                         && \node(ld-pif-open-2){};      & \node[draw,anchor=base](pif-2){PIF}; \\
                                   &                                    &                                       && \node(ld-pif-ping-2){};            & \node(pif-ld-ping-2){}; \\
                                   &                                    &                                       && \node(ld-pif-pong-2){};           & \node(pif-ld-pong-2){}; \\
                                   &                                    &                                       && \node(ld-pif-bprov-2){};          & \node(pif-ld-bprov-2){}; \\
                                   &                                    &                                       && \node(ld-pif-certinfo-2){};       & \node(pif-ld-certinfo-2){}; \\
                                   &                                    &                                       && \node(ld-pif-req-key-2){};      & \node(pif-ld-req-key-2){}; \\
                                   &                                    &                                         && \node(ld-prov-cont-key-pair){};  & \\ [3ex]
                                   &                                    &                                         && \node(ld-prov-cont-pkb){};       & \node(pif-prov-cont-pkb){};\\
                                   & \node(idp-prov-cont-req-uc){};     &                                         &&                                  & \node(pif-prov-cont-req-uc){};\\
                                   & \node(idp-prov-cont-cert-uc){};    &                                         &&                                  & \\
                                   & \node(idp-prov-cont-send-uc){};    &                                         &&                                  & \node(pif-prov-cont-send-uc){};\\
                                   &                                    &                                         && \node(ld-prov-cont-recv-uc){};   & \node(pif-prov-cont-recv-uc){};\\
                                   &                                    &                                         && \node(ld-pif-close-2){};      & \node[draw,anchor=base](pif-end-2){PIF}; \\
\node(phase-auth-lpo-left){};      &                                    &                                         &&                                  & \node(phase-auth-lpo-right){}; \\
                                   &                                    &                                         && \node(ld-gen-cap-lpo){};         & \\ [3ex]
\node(lpo-ld-auth){};              &                                    &                                         && \node(ld-lpo-auth){};            & \\ [-2pt]
\node(phase-cap-left){};           &                                    &                                         &&                                  & \node(phase-cap-right){}; \\ [2pt]
\node(lpo-ld-list-emails){};       &                                    &                                         && \node(ld-lpo-list-emails){};     & \\ [2pt]
\node(lpo-ld-addr-info-3){};       &                                    &                                         && \node(ld-lpo-addr-info-3){};     & \\
                                   &                                    &                                         && \node(ld-gen-cap){};             & \\ [3ex]
                                   &                                    & \node(rpdoc-ld-cap){};                  && \node(ld-rpdoc-cap){};           & \\
                                   &                                    & \node(rpdoc-ld-close){};                && \node[draw,anchor=base](ld-end){/LD}; & \\

                                   &                                    & \node(rp-doc-cif-liu){}; & \node(cif-liu){}; & &\\
                                   &                                    & \node(rp-doc-cif-dlg-cpt){}; & \node(cif-dlg-cpt){}; & &\\
 \node(lpo-cif-ctx2){};            &                                    &                                         & \node(cif-ctx2){}; & & \\

\node(lpo-end){};                  & \node(idp-end){};                  & \node(rpdoc-end){};                     & \node(cif-end){}; &               & \node(pif-col-end){}; \\
};

\node (ld-auth-drawn) at (ld-auth) [draw,fill=Gainsboro,rounded corners]{ auth IdP};
\draw [implies-implies,double] (ld-auth-drawn) -- (idp-ld-auth);

\draw [->] (ld-init-2.160) to node [above=-1.5pt]{ GET LD} (lpo-ld-init-2.20);
\draw [->] (lpo-ld-init-2.340) -- (ld-init-2.200);

\draw [->,color=red,dashed] (ld-rpdoc-send-ready-2) to node [above=-2pt]{ ready} (rpdoc-ld-recv-ready-2);

\draw [->,color=red,dashed] (rpdoc-ld-send-request-2) to node [above=-2pt]{ request} (ld-rpdoc-recv-request-2);

\draw [->,color=blue,>=latex] (ld-lpo-ctx-2.160) to node [above=-2pt]{ GET session\_context} (lpo-ld-ctx-2.20);
\draw [->,color=blue,>=latex] (lpo-ld-ctx-2.340) -- (ld-lpo-ctx-2.200);

\draw [->,color=blue,>=latex] (ld-lpo-addr-info-2.160) to node [above=-2pt]{ GET address\_info} (lpo-ld-addr-info-2.20);
\draw [->,color=blue,>=latex] (lpo-ld-addr-info-2.340) to node [above=-2pt]{} (ld-lpo-addr-info-2.200);

\draw [->,snake=snake,segment amplitude=0.2ex] (ld-pif-open-2.40) to node [above=-2pt] { create} (pif-2);

\draw [->,color=red,dashed] (pif-ld-ping-2) to node [above=-2pt]{ ping } (ld-pif-ping-2);
\draw [->,color=red,dashed] (ld-pif-pong-2) to node [above=-2pt]{ pong } (pif-ld-pong-2);
\draw [->,color=red,dashed] (pif-ld-bprov-2) to node [above=-2pt]{ beginProvisioning  } (ld-pif-bprov-2);
\draw [->,color=red,dashed] (ld-pif-certinfo-2) to node [above=-2pt]{ email, certDuration } (pif-ld-certinfo-2);
\draw [->,color=red,dashed] (pif-ld-req-key-2) to node [above=-2pt]{ requestKey } (ld-pif-req-key-2);

\node at (ld-prov-cont-key-pair) [draw,rounded corners,fill=Gainsboro]{ gen. key pair};

\draw [->,color=red,dashed] (ld-prov-cont-pkb) to node [above=-2pt]{ $\text{pk}_\text{b}$, email} (pif-prov-cont-pkb);

\draw [->,color=blue,>=latex] (pif-prov-cont-req-uc) to node [above=-2pt]{ $\text{pk}_\text{b}$, email} (idp-prov-cont-req-uc);

\node (certify-uc) at (idp-prov-cont-cert-uc) [draw,rounded corners,fill=white]{ create UC};

\draw [->,color=blue,>=latex] (idp-prov-cont-send-uc) to node [above=-2pt]{ UC} (pif-prov-cont-send-uc);

\draw [->,color=red,dashed] (pif-prov-cont-recv-uc) to node [above=-2pt]{ UC} (ld-prov-cont-recv-uc);

\draw [->,snake=snake,segment amplitude=0.2ex] (ld-pif-close-2.40) to node [above=-2pt] { close} (pif-end-2);

\node at (ld-gen-cap-lpo) [draw,rounded corners,fill=Gainsboro]{ gen. $\text{IA}_\text{LPO}$};

\draw [->,color=blue,>=latex] (ld-lpo-auth.160) to node [above=-2pt]{ POST auth\_with\_assertion ($\text{CAP}_\text{LPO}$)} (lpo-ld-auth.20);
\draw [->,color=blue,>=latex] (lpo-ld-auth.340) -- (ld-lpo-auth.200);

\draw [->,color=blue,>=latex] (ld-lpo-list-emails.160) to node [above=-2pt]{ GET list\_emails} (lpo-ld-list-emails.20);
\draw [->,color=blue,>=latex] (lpo-ld-list-emails.340) -- (ld-lpo-list-emails.200);

\draw [->,color=blue,>=latex] (ld-lpo-addr-info-3.160) to node [above=-2pt]{ GET address\_info} (lpo-ld-addr-info-3.20);
\draw [->,color=blue,>=latex] (lpo-ld-addr-info-3.340) -- (ld-lpo-addr-info-3.200);

\node at (ld-gen-cap) [draw,rounded corners,fill=Gainsboro]{ gen. $\text{IA}_\text{RP}$};

\draw [->,color=red,dashed] (ld-rpdoc-cap) to node [above=-2pt]{ response ($\text{CAP}_\text{RP}$)} (rpdoc-ld-cap);

\draw [->,snake=snake,segment amplitude=0.2ex] (rpdoc-ld-close.40) to node [above=-2pt] { close} (ld-end);

  \draw [->,color=red,dashed] (rp-doc-cif-liu) to node [above=-3pt]{ loggedInUser} (cif-liu);

  \draw [->,color=red,dashed] (rp-doc-cif-dlg-cpt) to node [above=-3pt]{ dlgCmplt} (cif-dlg-cpt);

  \draw [->,color=blue,>=latex] (cif-ctx2.160) to node [above=-2pt]{ GET session\_context } (lpo-cif-ctx2.20);
  \draw [->,color=blue,>=latex] (lpo-cif-ctx2.340) -- (cif-ctx2.200);

\begin{pgfonlayer}{background}
 \node (rpdoc-a) [above of=rpdoc, node distance=2ex]{};
 \node (rpdoc-al) [left of=rpdoc-a, node distance=6ex]{};
 \node (pif-col-end-b) [below of=pif-col-end, node distance=2ex]{};
 \node (pif-col-end-br) [right of=pif-col-end-b, node distance=4ex]{};
 \filldraw [color=Gainsboro,rounded corners] (rpdoc-al) rectangle (pif-col-end-br);

 \draw [color=gray,dotted] (lpo-fade.270)  -- (lpo);
 \draw [color=gray,dotted] (idp-fade.270)  -- (idp);
 \draw [color=gray,dotted] (rpdoc-fade.270) -- (rpdoc);
 \draw [color=gray,dotted] (cif-fade.270) -- (cif-top);
 \draw [color=gray,dotted] (ld-fade.270) -- (ld-top);

 \draw [color=gray] (lpo-fade.270)  -- (lpo-end);
 \draw [color=gray] (idp-fade.270)  -- (idp-end);
 \draw [color=gray] (rpdoc-fade.270) -- (rpdoc-end);
 \draw [color=gray] (cif-fade.270) -- (cif-end);
 \draw [color=gray] (ld-fade.270) -- (ld-end);
 \draw [color=gray] (pif-2.270) -- (pif-end-2);
\end{pgfonlayer}

\draw [dashed] (phase-ld-init-2-left.180) -- (phase-ld-init-2-right.0);
\draw [dashed] (phase-prov-cont-left.180) -- (phase-prov-cont-right.0);
\draw [dashed] (phase-auth-lpo-left.180) -- (phase-auth-lpo-right.0);
\draw [dashed] (phase-cap-left.180) -- (phase-cap-right.0);

\node at ($(ld-top)!0.7!(pif-top)$) {Browser};

\end{tikzpicture}

\raisebox{0.5ex}{\tikz{\draw [->] (0,0) -- (0.4,0);}}
   HTTPS messages, \raisebox{0.5ex}{\tikz{\draw
       [->,color=blue,>=latex] (0,0) -- (0.4,0);}} \xhrs (over HTTPS),
   \raisebox{0.5ex}{\tikz{\draw [->,color=red,dashed] (0,0) --
       (0.4,0);}} \pms, \raisebox{0.5ex}{\tikz{\draw
       [->,snake=snake,segment length=2ex,segment amplitude=0.2ex] (0,0) -- (0.4,0);}}
   browser commands
 \caption[foo2]{BrowserID primary mode typical login flow
   overview {\bf (part 2 of 2)}.}
\label{fig:browserid-lowlevel-pidp-detailed-2}

\end{figure}
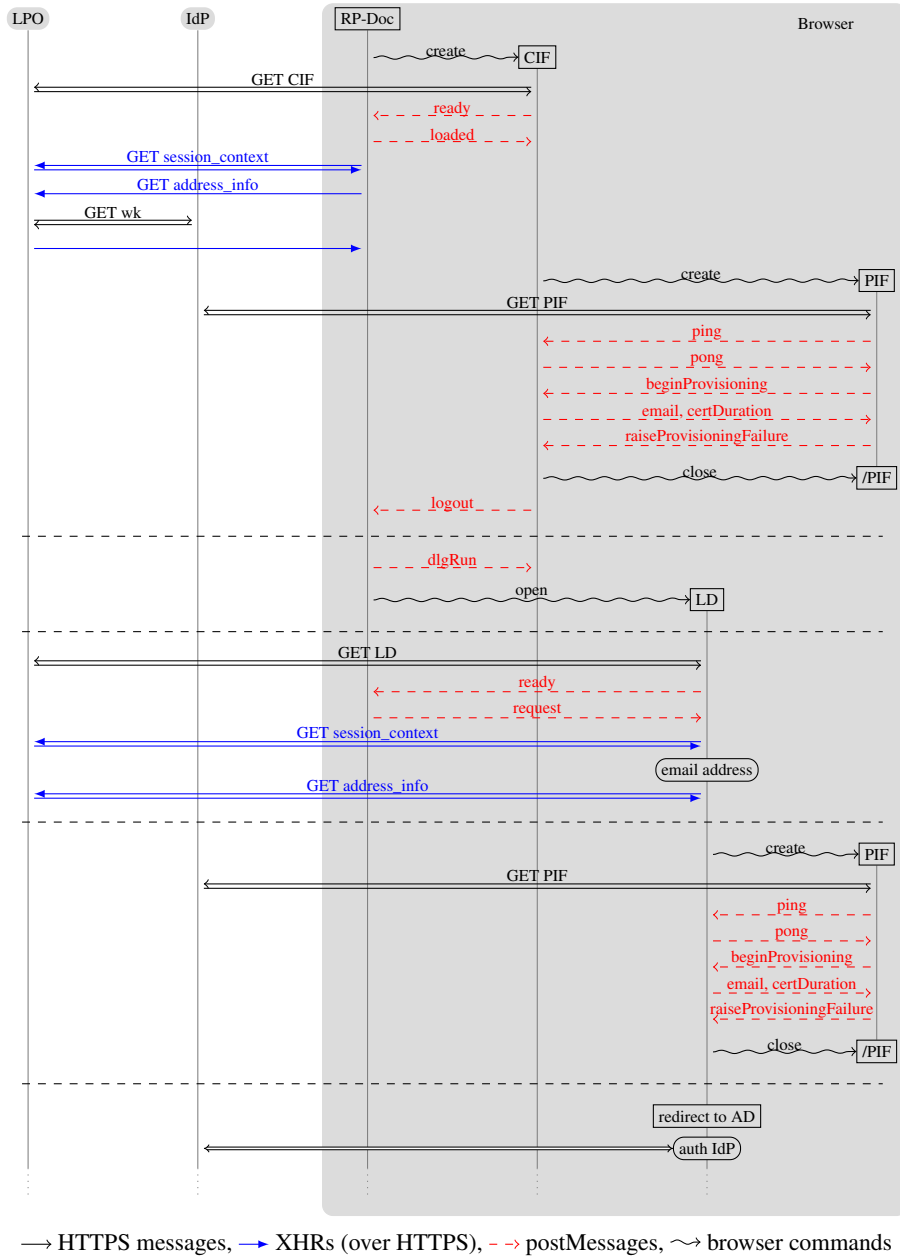

%% file: figure-sidp.tex
\begin{figure}[p!]\centering
 \scriptsize{
 \begin{tikzpicture}[scale=0.9, every node/.style={scale=0.9}]

  \matrix [column sep={5.35pc,between origins}, row sep=1.8ex]
  {
 \node[anchor=base,fill=Gainsboro,rounded corners](lpo){LPO}; & \node[draw, anchor=base](rp-doc-init-1){RP-Doc}; & \node (cif-top){}; & \node (ld-top) {}; \\
 & \node(rp-doc-iframe-cif){}; & \node[draw,anchor=base](cif-iframe){CIF}; & \\
 \node(lpo-cif-init){}; & & \node(cif-init){}; & \\
 & \node(rp-doc-cif-rdy){}; & \node(cif-rdy){}; & \\ 
 & \node(rp-doc-cif-ld){}; & \node(cif-ld){}; & \\ 
 \node(lpo-cif-ctx){}; & & \node(cif-ctx){}; & \\
 & \node(rp-doc-logout){}; & \node(cif-logout){}; & \\
 & \node(rp-doc-dlg-run){}; & \node(cif-dlg-run){}; & \\
 & \node(rp-doc-iframe-dlg){}; & & \node[draw,anchor=base](dlg-iframe){LD}; \\ [0.5ex]
 \node(lpo-dlg-init){}; & & & \node(dlg-init){}; \\
 & \node(rp-doc-dlg-rdy){}; & & \node(dlg-rdy){}; \\
 & \node(rp-doc-dlg-req){}; & & \node(dlg-req){}; \\
 \node(lpo-dlg-ctx){}; & & & \node(dlg-ctx){}; \\ [1ex]
 \node(lpo-dlg-auth){}; & & & \node(dlg-auth){}; \\
 & & & \node(dlg-gen-key){}; \\ [2ex]
 \node(lpo-dlg-cert-1){}; & & & \node(dlg-cert-1){}; \\
 \node(lpo-dlg-cert-sign){}; & & & \\
 \node(lpo-dlg-cert-2){}; & & & \node(dlg-cert-2){}; \\
 & & & \node(dlg-gen-ia){}; \\
 & \node(rp-doc-dlg-res){}; & & \node(dlg-res){}; \\
 & \node(rp-doc-dlg-close){}; && \node[draw,fill=Gainsboro,anchor=base](dlg-close){/LD}; \\
 & \node(rp-doc-cif-liu){}; & \node(cif-liu){}; & \\
 & \node(rp-doc-cif-dlg-cpt){}; & \node(cif-dlg-cpt){}; & \\
 \node(lpo-cif-ctx2){}; & \node(rp-doc-vrfy){}; & \node(cif-ctx2){}; & \\
 \node(lpo-end){};& \node(rp-doc-end){}; & \node(cif-end){}; & \node(dlg-end){}; \\
  };

  \tikzstyle{xhrArrow} = [color=blue,decoration={markings, mark=at
    position 1 with {\arrow[color=blue]{triangle 45}}}, preaction
  = {decorate}]

  \draw [->,snake=snake,segment amplitude=0.2ex] (rp-doc-iframe-cif.40) to node [above=-2pt] { create} (cif-iframe);

  \draw [->] (cif-init.160) to node [above=-2pt]{ GET CIF } (lpo-cif-init.20);
  \draw [->] (lpo-cif-init.340) -- (cif-init.200);

  \draw [->,color=red,dashed] (cif-rdy) to node [above=-3pt]{ ready } (rp-doc-cif-rdy);

  \draw [->,color=red,dashed] (rp-doc-cif-ld) to node [above=-2pt]{ loaded } (cif-ld);

  \draw [->,color=blue,>=latex] (cif-ctx.160) to node [above=-2pt]{ GET session\_context } (lpo-cif-ctx.20);
  \draw [->,color=blue,>=latex] (lpo-cif-ctx.340) -- (cif-ctx.200);

  \draw [->,color=red,dashed] (cif-logout) to node [above=-3pt]{ logout } (rp-doc-logout);

  \draw [->,color=red,dashed] (rp-doc-dlg-run) to node [above=-3pt]{ dlgRun } (cif-dlg-run);

  \draw [->,snake=snake,segment amplitude=0.2ex] (rp-doc-iframe-dlg.40) to node [above=-2pt] { open } (dlg-iframe);

  \draw [->] (dlg-init.160) to node [above=-2pt]{ GET LD } (lpo-dlg-init.20);
  \draw [->] (lpo-dlg-init.340) -- (dlg-init.200);

  \draw [->,color=red,dashed] (dlg-rdy) to node [above=-3pt]{ ready } (rp-doc-dlg-rdy);

  \draw [->,color=red,dashed] (rp-doc-dlg-req) to node [above=-3pt]{ request} (dlg-req);

  \draw [->,color=blue,>=latex] (dlg-ctx.160) to node [above=-2pt]{ GET session\_context } (lpo-dlg-ctx.20);
  \draw [->,color=blue,>=latex] (lpo-dlg-ctx.340) -- (dlg-ctx.200);

  \draw [->,color=blue,>=latex] (dlg-auth.160) to node [above=-1.5pt]{ POST auth } (lpo-dlg-auth.20);
  \draw [->,color=blue,>=latex] (lpo-dlg-auth.340) -- (dlg-auth.200);

  \node [draw,fill=Gainsboro,rounded corners] at (dlg-gen-key) { gen. key pair };

  \draw [->,color=blue,>=latex] (dlg-cert-1) to node [above=-3pt]{ POST certreq} (lpo-dlg-cert-1);
  \node [draw,fill=White,rounded corners] (lpo-dlg-cert-sign-draw) at (lpo-dlg-cert-sign) { create UC };
  \draw [->,color=blue,>=latex] (lpo-dlg-cert-2) to node [above=-2pt]{UC} (dlg-cert-2);

  \node [draw,fill=Gainsboro,rounded corners] at (dlg-gen-ia) { gen. IA};

  \draw [->,color=red,dashed] (dlg-res) to node [above=-3pt]{ response} (rp-doc-dlg-res);

  \draw [->,snake=snake,segment amplitude=0.2ex] (rp-doc-dlg-close.40) to node [above=-2pt] { close} (dlg-close);

  \draw [->,color=red,dashed] (rp-doc-cif-liu) to node [above=-3pt]{ loggedInUser} (cif-liu);

  \draw [->,color=red,dashed] (rp-doc-cif-dlg-cpt) to node [above=-3pt]{ dlgCmplt} (cif-dlg-cpt);

  \draw [->,color=blue,>=latex] (cif-ctx2.160) to node [above=-2pt]{ GET session\_context } (lpo-cif-ctx2.20);
  \draw [->,color=blue,>=latex] (lpo-cif-ctx2.340) -- (cif-ctx2.200);

  \begin{pgfonlayer}{background}
   \node (rp-doc-a) [above of=rp-doc-init-1, node distance=2ex]{};
   \node (rp-doc-al) [left of=rp-doc-a, node distance=6ex]{};
   \node (dlg-end-b) [below of=dlg-end, node distance=2ex]{};
   \node (dlg-end-br) [right of=dlg-end-b, node distance=8ex]{};
   \filldraw [color=Gainsboro,rounded corners] (rp-doc-al) rectangle (dlg-end-br);

  \draw [color=gray] (rp-doc-init-1.270) -- (rp-doc-end);
  \draw [color=gray] (cif-iframe.270) -- (cif-end);
  \draw [color=gray] (dlg-iframe.270) -- (dlg-close);
  \draw [color=gray] (lpo.270) -- (lpo-end);

  \end{pgfonlayer}

\node at ($(cif-top)!0.9!(ld-top)$) {Browser};

 \end{tikzpicture}}

\raisebox{0.5ex}{\tikz{\draw [->] (0,0) -- (0.4,0);}}
   HTTPS messages, \raisebox{0.5ex}{\tikz{\draw
       [->,color=blue,>=latex] (0,0) -- (0.4,0);}} \xhrs (over HTTPS),
   \raisebox{0.5ex}{\tikz{\draw [->,color=red,dashed] (0,0) --
       (0.4,0);}} \pms, \raisebox{0.5ex}{\tikz{\draw
       [->,snake=snake,segment length=2ex,segment amplitude=0.2ex] (0,0) -- (0.4,0);}}
   browser commands

\caption{BrowserID secondary mode typical login flow
   overview. Similar abstraction level as in Figure~\ref{fig:browserid-lowlevel-pidp-detailed-1}.}
\label{fig:browserid-lowlevel-sidp}

\end{figure}
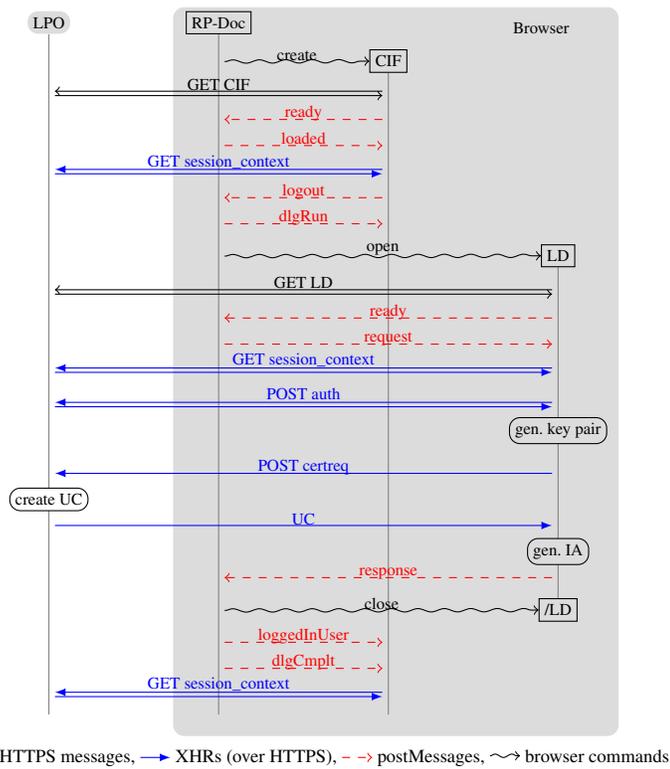